\documentclass[mnsc,nonblindrev]{informs3}
\RequirePackage{tgtermes}
\RequirePackage{newtxtext}
\RequirePackage{newtxmath}
\RequirePackage{bm}
\RequirePackage{endnotes}
\OneAndAHalfSpacedXI
\usepackage{appendix}
\usepackage[utf8]{inputenc} 
\usepackage[T1]{fontenc}    
\usepackage{hyperref}       
\usepackage{url}            
\usepackage{booktabs}       
\usepackage{xcolor}
\usepackage{amsmath}
\usepackage{amssymb}
\usepackage{amsfonts}       
\usepackage{nicefrac}       
\usepackage{microtype}      
\usepackage{xcolor}         
\usepackage{enumitem}
\usepackage{wrapfig}
\definecolor{framecolor}{rgb}{0, 0.502, 0.502}
\definecolor{bgcolor}{rgb}{0.960, 0.984, 1.000}
\usepackage{multirow}
\usepackage{array}
\usepackage{siunitx} 
\usepackage{longtable} 
\usepackage{tcolorbox}
\newtcolorbox{mybox}[1]{colback=bgcolor,colframe=framecolor,fonttitle=\bfseries,title=#1}

\usepackage{graphicx}
\usepackage{caption}
\usepackage{subcaption}
\usepackage{wrapfig}
\usepackage{multirow}
\usepackage{algorithm}
\usepackage{algpseudocode}
\usepackage{tikz}
\usepackage{threeparttable}
\usepackage{algpseudocode}
\usepackage{color,soul}
\usepackage{natbib}
\usepackage{xr-hyper}
\usepackage{hyperref}
\usepackage{graphicx} 
\algnewcommand{\algorithmicand}{\textbf{ and }}
\algnewcommand{\algorithmicor}{\textbf{ or }}
\algnewcommand{\OR}{\algorithmicor}
\algnewcommand{\AND}{\algorithmicand}
\algrenewcommand{\algorithmiccomment}[1]{\hfill// #1}
 \bibpunct[, ]{(}{)}{,}{a}{}{,}%
 \def\bibfont{\small}%
\EquationsNumberedThrough    
\TheoremsNumberedThrough     
\ECRepeatTheorems

\MANUSCRIPTNO{IJDS-0001-2024.00}

\makeatletter
\newcommand*{\addFileDependency}[1]{
  \typeout{(#1)}
  \@addtofilelist{#1}
  \IfFileExists{#1}{}{\typeout{No file #1.}}
}

\makeatother
\algrenewcommand\algorithmicrequire{\textbf{Input:}}
\algrenewcommand\algorithmicensure{\textbf{Output:}}
\algnewcommand{\LineComment}[1]{\Statex \(\triangleright\) #1}
\newcommand*{\myexternaldocument}[1]{%
    \externaldocument{#1}%
    \addFileDependency{#1.tex}%
    \addFileDependency{#1.aux}%
}

\myexternaldocument{./appendix}

\begin{document}
\RUNAUTHOR{Lin et al.} 
\RUNTITLE{CrowdLLM: Building Digital Populations Augmented with Generative Models}
\TITLE{CrowdLLM: Building LLM-Based Digital Populations Augmented with Generative Models}


\ARTICLEAUTHORS{%
\AUTHOR{Ryan Feng Lin\textsuperscript{a}\thanks{Equal contribution.}, Keyu Tian\textsuperscript{b}\footnotemark[1], Hanming Zheng\textsuperscript{b}, Congjing Zhang\textsuperscript{a}, Li Zeng\textsuperscript{b}\thanks{Corresponding authors.}, Shuai Huang\textsuperscript{a}\footnotemark[2]}
\AFF{\textsuperscript{a} Department of Industrial and
Systems Engineering, University of Washington, Seattle, WA 98195, USA\\
\textsuperscript{b} Department of Data Science, City University of Hong Kong, Kowloon, Hong Kong\\
\EMAIL{\{ryanflin, congjing, shuaih\}@uw.edu, \{ktian6-c, hanming.zheng\}@my.cityu.edu.hk, li.zeng@cityu.edu.hk}} 
} 

\ABSTRACT{%
The emergence of large language models (LLMs) has sparked much interest in creating LLM-based digital populations that can be applied to many applications such as social simulation, crowdsourcing, marketing, and recommendation systems. A digital population can reduce the cost of recruiting human participants and alleviate many concerns related to human subject study. However, research has found that most of the existing works rely solely on LLMs and could not sufficiently capture the accuracy and diversity of a real human population. To address this limitation, we propose CrowdLLM that integrates pretrained LLMs and generative models to enhance the diversity and fidelity of the digital population. We conduct theoretical analysis of CrowdLLM regarding its great potential in creating cost-effective, sufficiently representative, scalable digital populations that can match the quality of a real crowd. Comprehensive experiments are also conducted across multiple domains (e.g., crowdsourcing, voting, user rating) and simulation studies which demonstrate that CrowdLLM achieves promising performance in both accuracy and distributional fidelity to human data.}


\maketitle

\section{Introduction}\label{intro}

Recent years have witnessed the immense potential of large language models (LLMs) in performing human tasks \citep{song2023llm, liu2024make} and human behavior simulation \citep{zhousotopia, sun2024building}. This capacity of LLMs has sparked much interest in creating virtual human-like decision-making agents that can be used in many applications such as social simulation \citep{wang2025yulan, anthisposition}, behavioral studies \citep{chen2024evaluating, meng2024ai}, crowdsourcing \citep{grunde2025designing,xu2024role, moskovskiy2024llms}, marketing \citep{deshmukh2024harnessing, cai2025rtbagent}, recommendation \citep{shu2024rah, portugal2024agentic}, etc. A common theme of these applications is that they all involve a large group of human participants so as to solicit their decision-making powers to provide solutions to a task, and then, aggregate their solutions to solve the task \citep{zhang2014spectral}. Apparently, an implicit assumption made in these ``human-intensive'' operations is that the system could not bet on one single participant to solve the problem. Instead, it relies on the wisdom of the crowd, which by definition means a diverse collection of individuals who would provide different responses on the same problem. There are many reasons: sometimes it is because there is no ground truth (like in voting); sometimes it is because the quality of the crowd is not guaranteed (like in crowdsourcing); and sometimes it is simply because the goal is to solicit input from the diverse population (like collecting user ratings for products in order to develop accurate recommendation systems). The growing interest in using LLMs in these applications is motivated by reasons that vary across the applications, but one common reason is that in many of these applications the involvement of real humans may raise many concerns about privacy \citep{xia2020privacy}, confidentiality \citep{sims2019legal}, quality \citep{iren2014cost}, transparency \citep{xie2023dark}, etc. These concerns also extend to other fields, such as bias in social science \citep{alizadeh2025web} and privacy risks in recommendation systems \citep{wang2025privacy}. Beyond these legal, ethical, and cost-related challenges, recruiting workers itself presents significant complexities. The LLM-based synthetic crowd can circumvent these challenges and therefore inspire the many recent aforementioned developments. As articulated in \citep{Anthis2025ICML_LLM_Social}, LLMs should always be used as a concept testing tool for pilot and exploratory studies before we recruit real humans. 

Generally, an LLM-based model can be constructed based on a pretrained LLM, such as OpenAI ChatGPT \citep{achiam2023gpt}, Meta Llama \citep{touvron2023llama}, Google Gemini \citep{team2023gemini, team2025gemma}, Deepseek \citep{guo2025deepseek}, etc. It is followed by supervised fine-tuning (SFT) on a specific task \citep{li2024getting, ding2023parameter}, possibly combined with human preference alignment through learning methods such as reinforcement learning from human feedback (RLHF) \citep{hua2024intuitive, rafailov2024direct, schulman2017proximal}, to achieve better performance in completing the given task. Although such a pipeline has been widely adopted, the underlying drawbacks are not negligible. While SFT is much cheaper than pretraining \citep{xia2024understanding}, it can still be cost-prohibitive, which further demands parameter-efficient techniques to reduce the cost \citep{hu2021lora, ding2023parameter}. Meanwhile, SFT or RLHF can be largely impacted by the quality and the curation of the task-specific data used for tuning \citep{liu2024coachlm, yeh2024reliable, chang2022data}. Thus, building a high-achieving LLM-based model to perform human tasks remains a challenge, particularly when there is a lack of high-quality data and computational resources. 




Noting these issues, many endeavors have been devoted to the improvement of the model tuning \citep{wu2025llm, yin2024lofit} and prompt designing \citep{zhaopareto, zamfirescu2023johnny, zhang2024neural} so as to craft well-tuned task-specific LLMs. However, LLMs are pure ``black box'' models whose controllability (i.e., of their behaviors and output) is known to be a challenge. While instruction prompts play a significant role in inducing LLMs to show desired behaviors, it is usually difficult to find a universal design of the prompts for various tasks. It is also beyond our reach to know whether LLMs really understand the prompts and generate the outputs causally based on the instructions. Additionally, LLMs usually generate outputs with little diversity \citep{kirk2023understanding, peterson2024ai, padmakumardoes2024writing}, which is actually one key challenge in the development of the envisioned digital populations that we are interested here. It can be seen that many existing efforts are devoted to the framework of LLM which relies on large-scale high-quality data and substantial computational resources, whereas in many aforementioned applications, sometimes a lightweight solution of the LLM-based digital population is sufficient for the task. After all, in these applications, such as crowdsourcing a data labeling task or surveying a potential market for a new product, the human participants need not be experts; they may perform poorly on individual tasks, yet through effective aggregation of their inputs, they can collectively achieve satisfactory results. Still, the main challenges for developing such an LLM-based digital population are, as pointed out in \citep{Anthis2025ICML_LLM_Social} and many other recent efforts, that these LLM-based models usually exhibit a lack of diversity and undetected bias, and inaccuracies due to excessively user-pleasing outputs. 

Therefore, targeting the applications where a lightweight solution of the LLM-based digital population is sufficient, we pursue a strategy that is different from the existing efforts that aim to solve the problem within the LLM framework itself. Rather, our approach is to augment LLM with a generative machine learning model that can provide the diversity it needs, mitigate the bias it implicitly has, and improve its accuracy. We develop a principled design of a computational pipeline that is lightweight enough to be cost-effective but also sufficiently accurate and robust to guide the generation and aggregation of a diverse pool of LLM-based virtual participants to match the diversity and accuracy of real-world operations. 

Our main contributions include: (1) we propose CrowdLLM to emulate decision-making diversity and distributional fidelity observed in many real-world operations in crowdsourcing, voting, and product reviews; (2) CrowdLLM is built on a rigorous probabilistic framework that integrates the best of the two worlds, the LLM and the generative ML models; (3) we conduct theoretical analysis of CrowdLLM regarding its great potential in creating cost-effective, sufficiently representative, scalable digital populations that can match the quality of real populations; and (4) we conduct comprehensive experiments across multiple domains (e.g., crowdsourcing, voting, user rating) and simulation studies which demonstrate that CrowdLLM achieves promising performance in both accuracy and distributional fidelity to human data.

\begin{figure}[!t]
    \centering
    \begin{subfigure}{0.45\textwidth} 
        \centering
        \includegraphics[width=\textwidth]{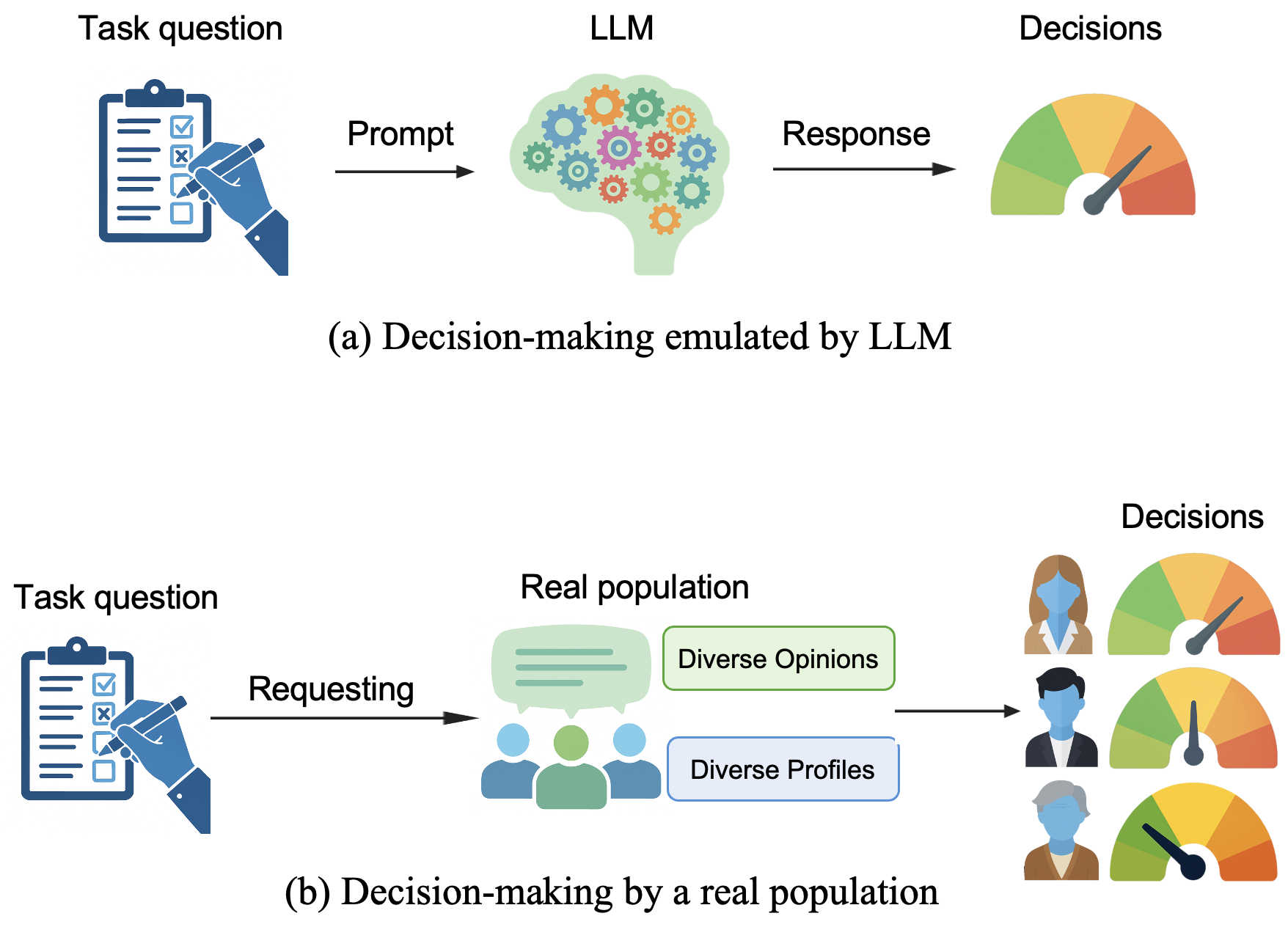}
        \label{fig1a}
    \end{subfigure}%
    \hfill
    \begin{subfigure}{0.45\textwidth} 
        \centering
        \includegraphics[width=\textwidth]{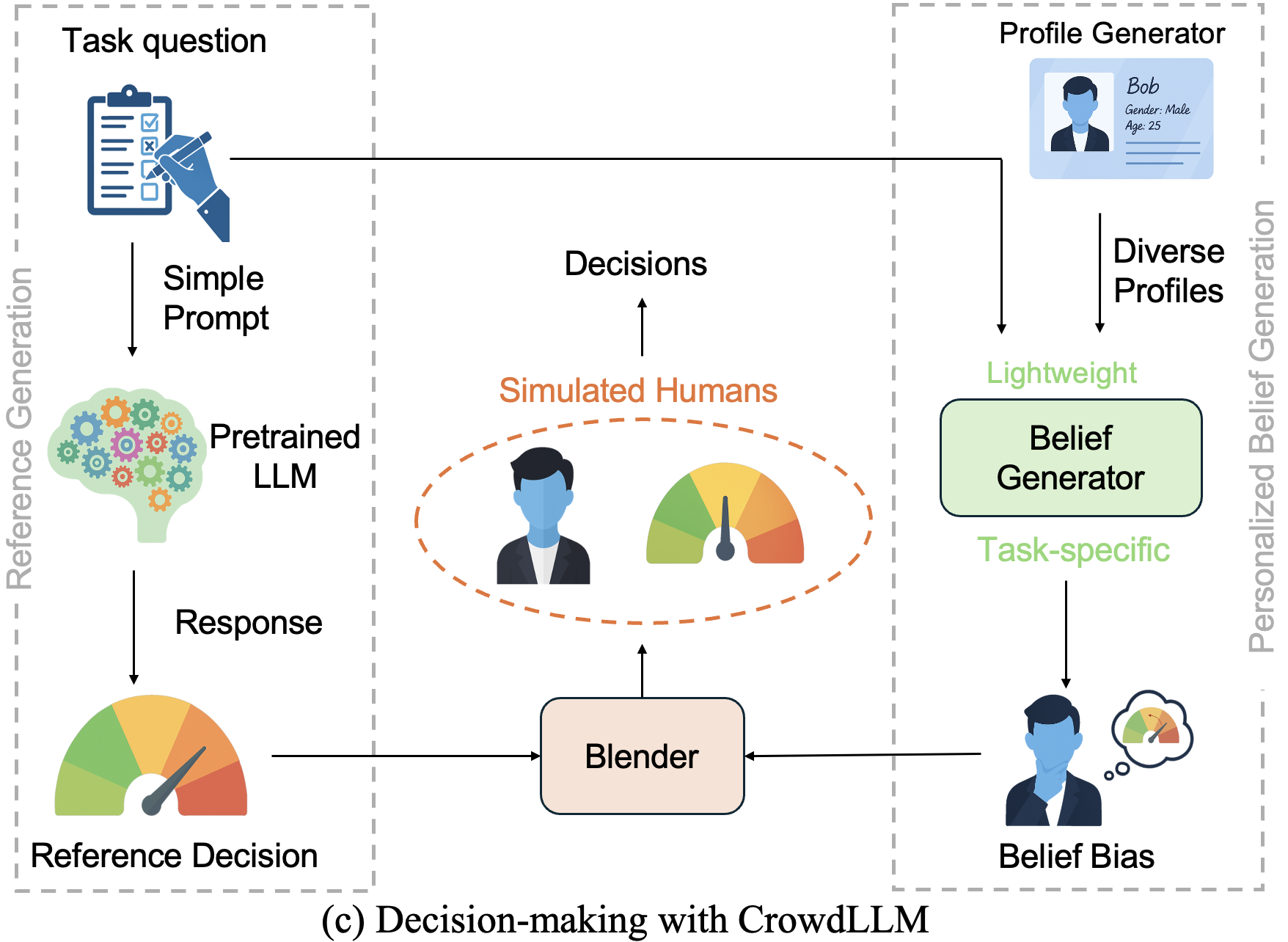}
        \label{fig1b}
    \end{subfigure}%
    \caption{A comparison of different decision-making workflows. (a) LLM: Decisions are purely made by LLM through the input of prompts. (b) Real population: Diverse decisions are made by a population of humans with diverse profiles. (c) CrowdLLM: Diverse decisions are made by simulated humans. Each simulated human's decision is a blend of a reference decision generated by a pretrained LLM and the personal belief bias generated by a belief generator. The simulated humans are sampled probabilistically by a profile generator.}
\end{figure}


\section{Related Work} 
\subsection{The Promises and Pitfalls of LLMs in Simulating Humans} LLMs have been used to simulate human behavior \citep{lu2025beyond, karten2025llm}, decision-making processes \citep{eigner2024determinants} and complex social interactions \citep{leng2023llm, bui2025mixture}. For example, through a series of Trust Games grounded in behavioral economics and modeled with Belief-Desire-Intention reasoning, \cite{xie2024can} show that GPT-4 agents exhibit strong behavioral alignment with humans in both actions and underlying rationales. Agent-based modeling with LLMs also shows great promise for large-scale social simulations. Frameworks such as AgentSociety \citep{piao2025agentsociety} and SocioVerse \citep{zhang2025socioverse} exemplify this potential, demonstrating simulations involving tens of thousands of LLM-driven agents or drawing upon millions of real users to inform agent behavior. These platforms aim to model complex societal dynamics, simulate millions of interactions, and study collective responses to events like policy changes or natural disasters. However, general-purpose LLMs can exhibit low accuracy on specific behavioral simulations. \cite{lu2025beyond} shows that fine-tuned LLMs \citep{binz2024centaur} on behavioral data enriched with synthesized reasoning traces substantially improve the accuracy of action generation compared to training on actions alone. While many works demonstrated the promises of LLMs, there have also been many evidences that pointed out their pitfalls in generating human-grade data. For example, \cite{gao2025take} use economic games to demonstrate that LLM behaviors are not consistent with humans and fine-tuned LLMs may only mimic specific patterns or contexts with reduced diversity even in simple scenarios. Beyond these inconsistencies, LLMs have also been observed to be associated with diminished output diversity \citep{padmakumardoes, chen2025dlcrec, zhang2025noveltybench}. 

To address the limited diversity and reliability in LLM-generated simulation outputs, \cite{dong2024can} propose the LLM-as-a-Personalized-Judge framework and reveal that integrating verbal uncertainty estimation improves alignment with human judgments. \cite{wang2025multilingual} propose a multilingual prompting strategy to increase diversity by activating cultural knowledge embedded in model training data. Similarly, \cite{shypula2025evaluating} introduce methods for evaluating and mitigating representational bias in LLM-driven outputs, ensuring that outputs better reflect a wide range of demographic and cultural perspectives. In addition, \cite{liu2025cultural} emphasizes the importance of refining training datasets to reduce biases and improve both accuracy and fairness. Moreover, \cite{mai2024improving} highlight the role of hybrid human-LLM teaming to enhance model performance, demonstrating that human feedback can mitigate errors in complex simulations. These approaches aim to tackle issues of diversity and accuracy of LLMs, making them more reliable and representative for practical applications.


\subsection{LLM-Based Digital Populations} 

In this subsection, we review existing works that create digital populations/synthetic crowds for applications such as voting, crowdsourcing, and product reviews. For example: (1) \textit{Crowdsourcing}: Crowdsourcing \citep{howe2006rise} leverages the collective intelligence of workers who are usually non-experts to perform tasks such as labeling, classification, and data verification. The quality of crowdsourced data is often a challenge due to worker inconsistency, spammers, and labeling noise. Recruiting workers and ensuring response quality is time-consuming and costly. LLM-based agents could circumvent many of these issues. Costabile et al. \citep{costabile2025assessing} suggest that an LLM-based crowd might outperform human crowds in fact-checking tasks by exhibiting less bias and higher consistency. Moskovskiy et al. \citep{moskovskiy2024llms} find that with techniques like activation patching, LLMs can generate parallel data with quality rivaling human-annotated corpora. However, Veselovsky et al. \citep{veselovsky2025prevalence} suggest that only using LLM-based agents may be problematic in crowdsourcing scenarios considering their limitations in capturing the full range of human preferences and viewpoints.  Wu et al. \citep{wu2023llms} investigate LLMs as workers in complex human-computational algorithms, observing variable success but highlighting potential for LLMs to handle sub-tasks within larger pipelines. To solve challenges in tasks requiring aligned and nuanced rewriting, Zeng et al. \citep{zeng2024combining} propose hybrid aggregation strategies that combine LLM and crowd judgments for misinformation detection. These studies suggest LLM-based agents alone are insufficient, so some researchers explore different ways of human-LLM collaboration. E.g., Creator-Aggregator Multi-Stage \citep{li2024human}, where LLMs and humans team up, aims to leverage mutual strengths by having humans come up with initial drafting and use LLMs, humans, and models to generate text answer aggregation. Li \citep{li2024comparative} shows that selective integration of LLM annotations can enhance overall annotation quality in both full and few-crowd settings. Tamura et al. \citep{tamurasimulation} uses simulation-based approaches to understand optimal aggregation strategies in a human+AI crowd. (2) \textit{Synthetic users in recommendation systems}: LLM-powered user simulators have become an important tool for recommendation systems by generating high-fidelity and interpretable synthetic interaction data that alleviates data sparsity and reduces the cost of online exploration. Recent advances take this idea in different directions: Agent4Rec \citep{zhang2024generative} emphasizes population diversity by initializing agents with heterogeneous traits; RecAgent \citep{wang2025user} prioritizes behavioral fidelity through cognitive components such as memory, reflection, and planning; SUBER  \citep{corecco2024suber} focuses on controllability and reproducibility for long-horizon evaluation, and \citep{zhang2025llm} enhances transparency by explicitly modeling user preference logic and mitigating hallucination through an ensemble of logical and statistical components. Together, these simulators offer interpretable and adaptable user behavior models for evaluating recommendation policies, but they also face shared limitations, including the high computational cost of cognitively rich agents and the challenge of balancing population diversity with stable, non-drifting within-agent preferences. (3) \textit{Voting}: LLMs have been explored in electoral contexts, but their use raises concerns regarding consistency, fairness, and reliability in collective decision-making. Studies reveal issues across different elections: \citep{cen2025large} observed biases and inconsistencies in LLM responses during the 2024 U.S. presidential election, while \citep{von2024united} reported failures in LLM-based predictions of the 2024 European Parliament elections, particularly in handling diverse national and linguistic contexts. Approaches such as fair voting aggregation have been proposed to mitigate these effects \citep{majumdar2024generative}. Despite these interventions, these observations highlight broader limitations of general-purpose LLMs in voting scenarios, particularly their limited capacity to capture human variability and their reliance on overly simplistic decision aggregation mechanisms. First, LLMs demonstrate a lack of diversity in synthetic outputs. \citep{ball2025human} observed that LLM-generated data fails to replicate the variance seen in real human responses, with limited differentiation in persona-to-party mappings. Consistently, \citep{yang2024llm} showed that LLM outputs produce less diverse collective outcomes in simulated voting scenarios, and their data is used in our experiments, where our method improves both consistency and diversity. Second, LLM-based collective decision-making systems exhibit limitations in decision mechanisms, often relying on simplistic aggregation methods such as plurality or dictatorial voting, which constrains collective reasoning and robustness \citep{zhao2024electoral}. 

\subsection{Generative Models}
Generative modeling methods aim to learn the underlying data distribution to capture complex latent structures and variability, enabling models to generalize across diverse scenarios. Over time, this goal has driven the development of several major paradigms. Variational Autoencoders (VAEs) \citep{kingma2013auto} introduced a stable, likelihood-based framework in which data are encoded into a latent distribution and sampled through the reparameterization trick, enabling efficient conditional generation. Although VAEs may produce slightly smoothed outputs and fewer extreme samples due to the variational approximation, they remain tractable, stable to train, and easily conditioned on input variables. Generative Adversarial Networks (GANs) \citep{goodfellow2014generative} enhance sample fidelity by training a generator against a discriminator, but this adversarial setup introduces instability, mode collapse, and high sensitivity to hyperparameters, making controlled diversity difficult. Diffusion models \citep{ho2020denoising} achieve strong distributional coverage through iterative denoising, yet they require heavy computation, large datasets, and slow sampling, limiting their practicality in low-dimensional or conditional settings. Optimal transport–based models \citep{li2023dpm} and normalizing flows provide exact likelihoods through invertible mappings but impose structural constraints that restrict flexibility in conditional scenarios. While each method addresses certain weaknesses, they also introduce trade-offs in stability, controllability, or computational cost. For generating structured latent variations, a VAE provides a stable and tractable probabilistic framework, making it the natural choice for the belief-generation component of CrowdLLM.

\section{CrowdLLM}

Our target applications involve a group of human participants to solicit their decision-making powers to provide solutions to a task, and then aggregate their solutions to solve tasks. Before formally presenting our framework, let's first provide an analytic characterization of these applications.  For a specific decision-making task, we consider a set of $T$ problems $\mathcal{T}=\{1,2,\cdots, T\}$. Each problem $t (t\in \mathcal{T})$ is associated with a description $\boldsymbol{x}_t\in\mathcal{X}$, where $\mathcal{X}$ is the problem space. Given $\boldsymbol{x}_t$, one needs to give their individual response. For example, in a choice-making scenario, they need to make a choice $y_t$ from the set of $M_t$ alternatives $\mathcal{Y}_t=\{1,2,\cdots, M_t\}$. The ultimate goal of decision-making is to find an optimal rule $\psi:\mathcal{X}\to\mathcal{Y}$ to give the decision $\hat{y}=\psi(\boldsymbol{x}_t)$. Suppose we have $N$ participants, and each is with a profile vector $\boldsymbol{v}_i (i=1,\cdots, N)$, i.e., which includes their demographics or user characteristics. Each problem will be assigned to all the participants, but they can choose whether to perform the task or not. Thus, for each problem $t$, we will only collect responses from a set of $N_t$ participants, denoted by $Y_t=\Bigl\{y_{t,i_n}\in \mathcal{Y}_t|n=1,\cdots, N_t\Bigr\}$. With an aggregation function $h(\bigcup_{i=1}^N\{y_{t,i}\})$, the final decision for the problem can be represented by $\hat{y}_t=h(Y_t)$. As a result, the decision-making rule is de facto an ensemble of personalized decision-making rules. Both personalization and aggregation are important in our target applications. While aggregation synthesizes the collective wisdom, personalization emphasizes the diversity of participants in both their profiles and opinions, as illustrated in Figure 1(b). 




The overall framework of CrowdLLM is shown in Figure 1(c). Different from a pure LLM-based model shown in Figure 1(a), in CrowdLLM, the LLM-emulated participants are augmented with a generative model to mimic the task-specific behaviors of real human participants. Note that rather than building numerous personalized LLMs tailored for each individual (pure LLM agents), CrowdLLM allows all these virtual individuals to share one single pretrained LLM as their engine, which is a more cost-effective solution. A full description of CrowdLLM is shown in the box below, and more details are given in the rest of this section.
\begin{mybox}{CrowdLLM: A Synthetic Crowd of Human Participants}
\textbf{Input:} Candidate pool $\mathcal{S}$; Recruitment budget $N$; Task-specific problem set $\mathcal{T}=\{1,\cdots, T\}$, each problem $t$ with its description $\boldsymbol{x}_t$, requirements $\mathcal{R}_t$ and context $\mathcal{C}_t$; Frozen LLM $\mathcal{M}$.
\begin{enumerate}[leftmargin=*]
\item \textbf{Virtual Participant Recruitment:} Produce a set of $N$ participants with qualified profiles $\boldsymbol{v}_1,\cdots, \boldsymbol{v}_T$ from the candidate pool for problem $t$. For each problem $t$, assign the problem to the participants and ask them to make decisions. 
\item \textbf{Reference Generation:} Instruct the LLM $\mathcal{M}$ with problem-specific prompt $\mathcal{P}=f(\boldsymbol{x}_t, \mathcal{R}_t, \mathcal{C}_t)$ to generate reference decisions $y_{ref}\sim\pi_{\mathcal{M}}(y|\mathcal{P})$ for the problem $t$.
\item \textbf{Belief Generation:} For the $i$-th participant, if their participation status $\varphi_{t,i}=1$, generate their belief bias over the problem as $\boldsymbol{\delta}_{i,t}=G_{belief}(\boldsymbol{x}_t, \boldsymbol{v}_i)$. 
\item \textbf{Personalized Decision-Making:} For the $i$-th participant, the personalized decision is made by a blending of the reference decisions and personalized belief bias which follows a probabilistic model $\hat{y}_{i,t}\sim \pi(y|B_{\sigma}(y_{ref}, \boldsymbol{\delta}_{i,t}), \boldsymbol{x}_t, \boldsymbol{v}_i)$ where $B_{\sigma}(y_{ref}, \boldsymbol{\delta}_{i,t})$ is the blender.
\item \textbf{Decision Aggregation:} Depending on the task, we can aggregate the decisions for any problem $t$ by a function $h$ through $y=h(Y_t), \text{where}\,\, Y_t=\{\hat{y}_{t,i}|\varphi_{t,i}=1\}$. 
\end{enumerate}
\end{mybox}
\subsection{Details of Each Component in CrowdLLM}

%


\paragraph{Reference generation.} Recall that for any given problem/task $\boldsymbol{x}_t$, any decision CrowdLLM generates combines two inputs, i.e., as illustrated in Figure 1(c), the reference decision and the personal belief. To produce the reference decision for problem $t$, we leverage LLM, in particular, a pretrained LLM $\mathcal{M}$, since it is computationally cheaper and imposes fewer requirements on specialized hardware compared to fine-tuning \citep{seedat2024curated}. For a single LLM-emulated participant, with $\mathcal{P}=f(\boldsymbol{x}_t, \mathcal{R}_t, \mathcal{C}_t)$ as context, i.e., recall that each problem $t$ has its description $\boldsymbol{x}_t$, requirements $\mathcal{R}_t$ and context $\mathcal{C}_t$ (see an example of $\mathcal{R}_t, \mathcal{C}_t$ in Appendix), we prompt $\mathcal{M}$ to generate several decisions, which we call reference decisions. This can be viewed as sampling from a reference distribution $\pi_{M}$ over $\mathcal{Y}$, i.e., $y_{ref}\sim \pi_{\mathcal{M}}(y|\mathcal{P})$. To ensure the reliability of the reference decision, we perform $K$ times of generation, which yields a set of decisions $\{y'_1,\cdots, y_K'\}$. In summary, the reference decision can be expressed as an aggregated decision:
\vspace{0.3cm}
\begin{align*}
\begin{aligned}
&y_{ref}=h_{\mathcal{M}}(y'_1,\cdots, y'_K), \\
&y_k'\sim\pi_{\mathcal{M}}(y|\mathcal{P}),\quad k=1,\cdots, K,  
\end{aligned}
\end{align*}
where $h_{\mathcal{M}}(\cdot)$ is an aggregation function, e.g., mean or majority voting. Though as a good common sense respondent, it is known that LLM-based agents often fail to generate differentiated decisions but instead follow the same common sense \citep{veselovsky2025prevalence, shypula2025evaluating, xu2024echoes}, even when we vary the ways of prompting (multi-persona prompting) and temperature settings. Our experiments in Section 4 also show that the LLM-emulated participants lack diversity compared with real humans' decisions.


\paragraph{Belief generation.}  To ensure the LLM-emulated participants can make diverse decisions as humans, we introduce a belief generator $G_{belief}(\cdot)$ to generate personalized belief biases. The generator can be implemented as a lightweight generative network that can adapt to a specific task. It takes both the participant's profile $\boldsymbol{v}_i$ and the problem description $\boldsymbol{x}_t$ as the input, and encodes them into a belief bias. The generation of the belief bias follows an inference model:
\vspace{0.3cm}
\begin{align}
\begin{aligned}
\boldsymbol{\delta}_{i,t}\sim p(\boldsymbol{\delta}|g_x(\boldsymbol{x}_t), g_z(\boldsymbol{v}_i))=\mathcal{N}\Bigl(\boldsymbol{\mu}\Bigl(g_x(\boldsymbol{x}_t), g_z(\boldsymbol{v}_i)\Bigr), \mathbf{\Sigma}\Bigl(g_x(\boldsymbol{x}_t), g_z(\boldsymbol{v}_i)\Bigr)\Bigr),
\end{aligned}
\label{eq3}
\end{align}
where $g_x(\cdot)$ and $g_z(\cdot)$ are embedding functions parameterized with $\boldsymbol{\beta}$. When fixing the participants' profiles $\boldsymbol{v}_i$ as the context, this naturally leads to a variational autoencoder (VAE) conditioning on the profiles. But it should be also noted that if the profile generator is not frozen and $\boldsymbol{v}_i$ can vary with the noise $\boldsymbol{\varepsilon}_i$ that generates the profiles, $p(\boldsymbol{\delta}|g_x(\boldsymbol{x}_t), g_z(\boldsymbol{v}_i))$ is not necessarily Gaussian after marginalization and thus can result in a semi-implicit variational autoencoder which is able to accommodate non-Gaussian distributions through a hierarchy of stochastic layers \citep{yin2018semi}. To simplify the problem, we directly go with the VAE structure without considering such hierarchical inference. The reconstruction of the problem description is performed by a decoder $D(\cdot)$ through $\boldsymbol{x}_i=D(\boldsymbol{\delta}_{i,t}, g_z(\boldsymbol{v}_i))$. And the generated belief bias $\boldsymbol{\delta}_{i,t}$ is then fed into the blender (shown in Figure 1(c)) to produce the final decision from this virtual participant.


\paragraph{Personalized decision-making.} To finalize the decision of a single participant, we need to blend the reference decision generated by LLM with the personal belief bias that is sampled from \eqref{eq3} as a latent vector. Since \eqref{eq3} is a probabilistic model, to offset the impact of its randomness, we generate the final decision of the participant as an expected decision:
\begin{align*}
\begin{aligned}
\hat{y}_{t,i}=\mathbb{E}_{\boldsymbol{\delta}_{i,t}\sim p(\boldsymbol{\delta})}\Bigl[B_{\sigma}(y_{ref}, \boldsymbol{\delta}_{i,t})\Bigr]\approx\frac{1}{J}\sum_{j=1}^JB_{\sigma}(y_{ref}, \boldsymbol{\delta}^{(j)}_{i,t}).
\end{aligned}
\end{align*}

In practice, we can generate the personal belief bias $J$ times and approximate the expectation by the sample average. Here, $B_{\sigma}(\cdot)$ is the blender parameterized by $\sigma$. In our framework, the blender follows
\begin{align}
\tilde{y}_{t,i}= B_{\sigma}(y_{ref}, \boldsymbol{\delta}_{i,t})\sim \mathcal{F}(y_{ref}+\boldsymbol{\delta_i}, \sigma^2),
\label{blender}
\end{align}
where $\mathcal{F}$ is a preset distribution depending on the task. For example, if the decision to make is a continuous variable, $\mathcal{F}$ can be a normal distribution. The variance $\sigma^2$ reflects the noise level.

\paragraph{Crowd-level decision aggregation.} In the final step, for each problem $t$, the participants' decisions are aggregated through an aggregation function $h:\mathcal{Y}^N\to\mathcal{Y}$. Specifically, the aggregated response for problem $t$ can be written as 
\begin{align*}
y=h(Y_t), \text{where}\,\, Y_t=\{\hat{y}_{t,i}|\varphi_{t,i}=1\}.
\end{align*}
where $\varphi_{t,i}=1$ means participant $i$ responds to problem $t$. Various aggregation functions can be used, such as mean score, majority voting, Dawid-Skene model, etc. When the problems do not have a ground truth solution, the consensus or the decision distribution of human workers can be the gold standard to evaluate the performance of CrowdLLM.  



\paragraph{Virtual participant recruitment.} Last but not least, recruitment of LLM-emulated participants is realized through a random profile generator $G_u(\cdot)$, i.e., this profile generator should be responsible for generating the information of a participant and selecting participants from a pool of qualified profiles denoted by $\mathcal{S}$. $\mathcal{S}$ can be built based on task-specific prior knowledge (see an example in Figure~\ref{figs0}). Formally, the $i$-th participant's profile is expressed as 
\begin{align}
    \boldsymbol{v}_i = G_u(\mathcal{S},\boldsymbol{\varepsilon}_i; \boldsymbol{\theta}),
    \label{eq:zi}
\end{align}
where $\boldsymbol{\varepsilon}_i\sim q(\boldsymbol{\varepsilon})$ is random noise encouraging profile diversity and $\boldsymbol{\theta}$ is the generator's parameter. Once we obtain the profile of a participant, we can simulate one's decision-making behaviors through the generation components of CrowdLLM (i.e., from reference generation to decision aggregation). We also consider that in reality not all participants participate in all problems. We assume the participation of the $i$-th participant on problem $t$, $\varphi_{t,i}$, satisfies a Bernoulli distribution $\varphi_{i,t}\sim Bernoulli (p_{i,t})$, where $p_{i,t}$ is the probability of participation which can be defined by prior knowledge. 

\subsection{Model Training}

Training of CrowdLLM will only require a small set of real human data, since the LLM-emulated participants are built on a frozen pre-trained LLM and only the generators and the blender need to be trained. The training data includes multiple problems/instances and the decisions of a set of real human participants for each problem/instance. To train CrowdLLM, the loss functions are:
\begin{align}
\begin{aligned}
\mathcal{L}_1=&\frac{1}{N}\sum_{i=1}^N\frac{1}{T_i}\sum_{t=1}^{T_i}\varphi_{i,t}\Bigl\{\mathbb{KL}\Bigl(\mathbb{E}_{\Omega\sim q_{\phi}(\Omega|\boldsymbol{x}_t,\boldsymbol{v}_i)}[q_{\boldsymbol{\beta}}(\boldsymbol{\delta}_{i,t}|\Omega)]\Big\Vert p(\boldsymbol{\delta}_{i,t})\Bigr)\\
&-\mathbb{E}_{\Omega\sim q_{\phi}(\Omega|\boldsymbol{x}_t,\boldsymbol{v}_i)}\Bigl[\mathbb{E}_{\boldsymbol{\delta}_{i,t}\sim q_{\boldsymbol{\beta}}(\boldsymbol{\delta}_{i,t}|\Omega)}\Bigl[\log p(\boldsymbol{x}_t|\boldsymbol{\delta}_{i,t}, \boldsymbol{v}_i) \Bigr]\Bigr]\Bigr\},\\
\mathcal{L}_2=&\frac{1}{N}\sum_{i=1}^N\frac{1}{T_i}\sum_{t=1}^{T_i}\varphi_{i,t}\ell(\hat{y}_{i,t}, y_{i,t}),
\label{eq:loss}
\end{aligned}
\end{align}
where $T_i=\sum_{t=1}^T\varphi_{i,t}$ and $q(\Omega|\boldsymbol{x}_i,\boldsymbol{v}_i)$ is an implicit prior distribution. Here, $\mathcal{L}_1$ follows the semi-implicit VAE style \citep{yin2018semi} and ensures the model sufficiently represents the personal belief of the human participants, while $\mathcal{L}_2$ ensures the final accuracy of CrowdLLM, i.e., the final decision made by a virtual human participant should be close to its real human counterpart’s. The specific form of $\ell(\cdot, \cdot)$ in $\mathcal{L}_2$ depends on the task scenario. For regression-type continuous or ordinal judgment problems, $\ell(\hat{y}_{i,t}, y_{i,t})=||\hat{y}_{i,t}-y_{i,t}||_2^2$ is a common choice; For classification-type choice-making problems, $\ell(\hat{y}_{i,t}, y_{i,t})=\mathbb{I}[\hat{y}_{i,t}=y_{i,t}]$ is widely adopted. The overall loss function is $\mathcal{L}=\mathcal{L}_1+\lambda\mathcal{L}_2$, where $\lambda>0$ is a regularizer. By minimizing the loss function $\mathcal{L}$, we can optimize the parameters of CrowdLLM.

\section{Theoretical Analysis}\label{theory}
In this section, we perform theoretical analysis to further reveal the underlying mechanisms of CrowdLLM about why it can create realistic digital populations. Specifically, we ask the following questions: (Q1) Is CrowdLLM able to generate a target population with envisioned profile characteristics? (Q2) How does the diversity of the digital population generated by CrowdLLM affect the decision-making performance (e.g., accuracy)? And (Q3) How is the decision-making performance impacted by the quality of the LLM backbone and the generative models in CrowdLLM? All the proofs are in the Appendix.

To answer Q1, we can readily extend the theoretical results of generative models in the literature \citep{dahal2022deep, aamari2019estimating}. Specifically, the following theorem shows that it is possible to generate a diverse population of a target profile distribution $\mathcal{T}$ through the profile generator in CrowdLLM:
\begin{theorem}
Suppose the profiles are $d$-dimensional bounded vectors following a target mixed-type distribution $\mathcal{T}$. Consider $\rho$ as an easy-to-sample distribution taken to be uniform on $(0,1)^{d+1}$. For any $\varepsilon\in(0,1)$, there exists a profile generator $G$ building on a generative model that satisfies
\begin{align*}
    W_1(G_{\sharp}\rho, \mathcal{T})< (1+ \sqrt{\frac{2}{\pi}}\eta d)\varepsilon.
\end{align*}
Here, $W_1$ is the Wasserstein-1 distance, $G_{\sharp}\rho$ is the pushforward of $\rho$ by $G_{\sharp}$ which represents the resulting distribution transferred from $\rho$ to the generated profile space, and $\eta$ is a constant.
\end{theorem}
Theorem 1 indicates that we can build a profile generator to generate meaningful profiles of a target population. However, generating a diverse profile of the digital population doesn't guarantee CrowdLLM's good performance, as the virtual human participants, despite having a diverse profile, may still give similar responses on the same task. Thereby in CrowdLLM we further have the belief generation component to ensure that human participants can generate different responses as they have different profiles. 

\begin{figure}[!b]
    \centering
    \includegraphics[width=0.7\linewidth]{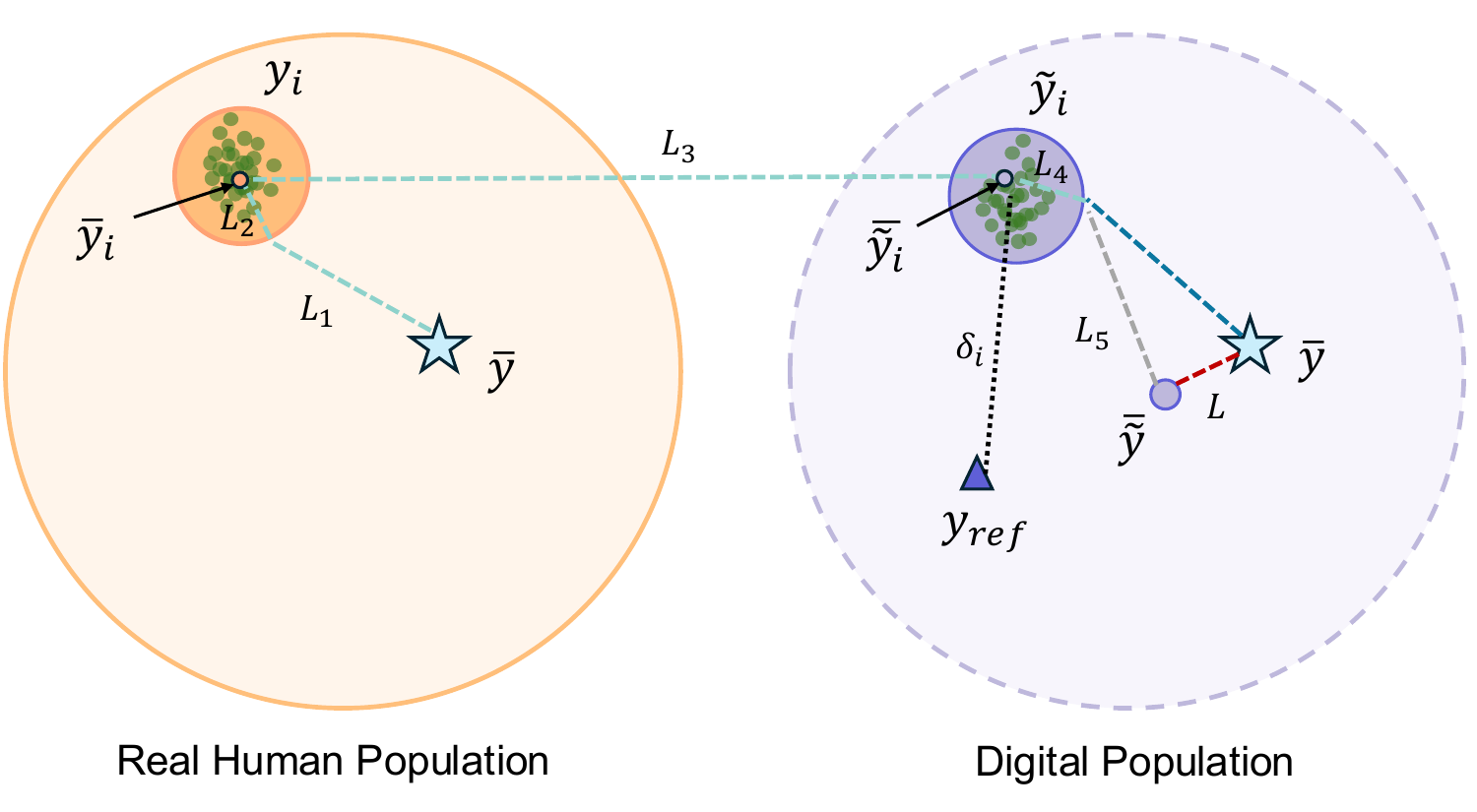}
    \caption{An illustration of the risk decomposition for a specific problem $\boldsymbol{x}$. The yellow circle represents the sample human population $U$ while the purple dashed circle represents their digital counterpart. The balls in the circles represent a physical human individual $u_i$ and their digital counterpart $\tilde{u}_i$. $\overline{y}_i$ and $\overline{\tilde{y}}_i$ are the expected responses of the human individual and the digital individual, respectively. $y_i$ and $\tilde{y}_i$ are their corresponding noisy observations. The empirical mean of the individual noisy responses $\tilde{y}_i$ across the whole digital population is represented by the light purple point $\overline{\tilde{y}}$. The dark purple triangle $y_{ref}$ is the reference response generated by the LLM. The star $\overline{y}$ represents the average response of the sample population, adopted as a ``ground truth". The five components $L_1$ to $L_5$ are explained in Theorem 2.}
\label{fig:fig2}
\end{figure}

Now we answer Q2. Consider a task that is characterized by the problem description $\boldsymbol{x}$. The Bayes-optimal response on this task is the conditional mean response given by the target human population as $y^*=\mathbb{E}_{\boldsymbol{v}\sim\mathcal{T}}\mathbb{E}_{y\sim\mathcal{Y}|\boldsymbol{x},\boldsymbol{v}}[y]$.  With a slight abuse of notation, we can write it as $y^*(\boldsymbol{x})$ interchangeably (similar simplification will be used in the rest of the paper without causing confusion). In practice, however, we typically have no access to this ground truth response. Instead, we rely on a finite sample population $U=\{u_i\}_{i=1,\cdots,N}$ whose profiles $\boldsymbol{v}_1,\cdots, \boldsymbol{v}_N$ are drawn from the target distribution $\mathcal{T}$. Each individual $u_i$ provides a response $y_i$. If we know  $\overline{y}_i=\mathbb{E}_{y_i\sim\mathcal{Y}|\boldsymbol{x},\boldsymbol{v}_i}[y_i]$ which is considered as their rational decision since the operator $\mathbb{E}$ averages out the randomness of their decisions, we can adopt the average of these expected responses across the sampled population $U$, i.e., $y^{**}=\frac{1}{N}\sum_{i=1}^N\overline{y}_i$, as a gold-standard approximation of $y^*$. However, the real human individuals' responses are typically noisy and can be expressed as $y_i=\overline{y}_i+\varepsilon_i$, where $\varepsilon_i\sim\mathcal{N}(0,\eta_i^2)$ captures the individual randomness. Ideally, if we could collect individual responses repeatedly many times, we could accurately estimate the individual-level expected response $\overline{y}_i$, and consequently, obtain an accurate estimate of $y^{**}$. Nevertheless, in practice, an individual typically provides only a single noisy response to a specific problem, which hinders the estimation of $\overline{y}_i$. Therefore, we can only rely on the noisy response $y_i$, and substitute $y^{**}$ with empirical mean $\overline{y}=\frac{1}{N}\sum_{i=1}^Ny_i$. $\overline{y}$ provides an unbiased estimate of $y^{**}$, which provides a ground truth (see the blue star in Figure \ref{fig:fig2}). Recall that to build a digital population, CrowdLLM generates virtual individuals $\tilde{U}=\{\tilde{u}_i\}_{i=1,\cdots, N}$ that mirror the real human individuals $U$. Suppose the digital counterpart of the individual $u_i$, denoted by $\tilde{u}_i$, is generated by CrowdLLM with the same profile $\boldsymbol{v}_i$. Their individual responses, either the expected $\overline{\tilde{y}}_i$ or the noisy $\tilde{y}_i$, can be linked to those of $u_i$, i.e., $\overline{y}_i$ or $y_i$, despite the deviations caused by any model's inherent limitations. Such a physical-digital pair is illustrated in Figure \ref{fig:fig2}. Following the decision-making process of CrowdLLM, we can simplify the notations and express the response to the problem $\boldsymbol{x}$ of the individual $u_i$ generated by CrowdLLM as
\begin{align*}
\tilde{y}_i=y_{ref}+\delta(\boldsymbol{x},\boldsymbol{v}_i)+\tilde{\varepsilon}_i,
\end{align*}
where $y_{ref}=\mathbb{E}[\Phi(\boldsymbol{x})]$ is the reference decision generated by the LLM backbone $\Phi$, $\delta(\cdot)$ denotes the belief generator, and $\tilde{\varepsilon}_i$ is the noise with $\mathbb{E}[\tilde{\varepsilon}_i]=0$ and $Var[\tilde{\varepsilon}_i]=\tilde{\eta}_i^2$ which represents the inherent uncertainty of individual $i$. For simplicity, following Eq. \eqref{blender}, we only consider the blender is additive and $\mathcal{F}$ is normal. Then, we compare these responses with real humans' responses. We only consider the analysis of the average response for a specific problem $\boldsymbol{x}$ and compare $\overline{\tilde{y}}$ with the ground truth $\overline{y}$. Their discrepancy can be measured by a loss function $\ell(\cdot,\cdot)$ as $\ell(\overline{\tilde{y}}, \overline{y})$, e.g., here we focus on the squared loss to conduct our theoretical inquiry. The same proof strategies can be extended to other loss functions, such as KL-divergence. Inspired by the unified theory of diversity \citep{wood2023unified}, we can prove the following theorem:
\begin{theorem}
Consider a digital population generated by CrowdLLM, i.e., $\tilde{U}=\{\tilde{u}_i\}_{i=1,\cdots,N}$ with profiles $\mathcal{Z}=\{\boldsymbol{v}_i\}_{i=1}^N\sim\mathcal{T}$, and their real human counterparts $U=\{u_i\}_{i=1,\cdots,N}$. Suppose the overall expected risk is $L=\mathbb{E}_{\mathcal{T}}\Bigl[\mathbb{E}_{\mathcal{D}}\bigl[\mathbb{E}_{\boldsymbol{x}\sim\mathcal{X},y\sim\mathcal{Y}}[\ell(\overline{y},\overline{\tilde{y}})]\bigr]\Bigr]$. Given the training data $\mathcal{D}=\bigcup_{i=1}^N\mathcal{D}_i$ where $\mathcal{D}_i$ is the data contributed by individual $u_i$, we have the following decomposition over the risk $L$:  
\begin{align*}
L=&\mathbb{E}_{\mathcal{X}}\Bigl[\underbrace{\mathbb{E}_{\mathcal{Z}\sim\mathcal{T}}\bigl[\frac{1}{N}\sum_{i=1}^N\mathbb{E}_{y_i\sim\mathcal{Y}|\mathcal{X},\boldsymbol{v}_i}[\ell(\overline{y},y_i)]\bigr]}_{L_1:\text{Average Human Bias}}+\underbrace{\mathbb{E}_{\mathcal{T}}\bigl[\frac{1}{N}\sum_{i=1}^N\eta_i^2\bigr]}_{L_2:\text{Human Individual Noise}}+\underbrace{\mathbb{E}_{\mathcal{T},\mathcal{D}}\bigl[\frac{1}{N}\sum_{i=1}^N\ell(\overline{y}_i,\overline{\tilde{y}}_i)\bigr]}_{L_3:\text{Twin Discrepancy}}\\ 
&+\underbrace{\mathbb{E}_{\mathcal{T}}\bigl[\frac{1}{N}\sum_{i=1}^N\tilde{\eta}_i^2\bigr]}_{L_4:\text{Allowed Individual Uncertainty}}-\underbrace{{E}_{\mathcal{T}, \mathcal{D}}\bigl[\frac{1}{N}\sum_{i=1}^N\ell(\overline{\tilde{y}},\tilde{y}_i)\bigr]}_{L_5:\text{Digital Population Diversity}}\Bigr]
\end{align*}
\label{thm2}
\end{theorem}
Theorem 2 decomposes the expected risk of CrowdLLM into five parts that correspond to the average human bias, the human individual noise, the discrepancy between the individuals of the physical and digital populations, the allowed individual uncertainty, and the diversity of the digital population. When other components ($L_1$ to $L_4$) are fixed, Theorem 2 shows greater diversity of the digital population ($L_5$) will reduce the expected risk. We can similarly compute the decision-making risk for pure LLMs as:
\begin{proposition}
With the same population $U=\{u_i\}_{i=1,\cdots,N}$ as in Theorem 2, pure LLM-based decision-making with zero-shot prompting yields the following decomposition over its risk $L'=\mathbb{E}_{\mathcal{T},\mathcal{X},\mathcal{Y}}[\ell(\overline{y}, y_{ref})]$:  
\begin{align*}
&L'=L_1+L_2+\mathbb{E}_{\mathcal{T}}\bigl[\frac{1}{N}\sum_{i=1}^N\ell(\overline{y}_i,y_{ref})\bigr]+\eta_{\Phi}(t),
\end{align*}    
where $\eta_{\Phi}(t)\geq0$ measures the randomness of the LLM outputs under temperature $t$.
\label{prop1}
\end{proposition}
With Theorem \ref{thm2} and Proposition \ref{prop1}, we can further provide a sufficient condition under which CrowdLLM outperforms pure LLM-based decision-making. Interestingly, this sufficient condition is built on the quality of the LLM backbones.
\begin{assumption}[Quality of the LLM backbone]
Given any specific task $\boldsymbol{x}$, for $\alpha\in(0,1)$ and $\gamma\in (0,\alpha)$, there exists a constant $\kappa_\alpha$, such that the deviation of the response given by the LLM backbone $\Phi$ from the gold-standard response $y^{**}(\boldsymbol{x})$ given by the human population is bounded with probability at least $1-\alpha$, i.e., $P(|\mathbb{E}[\Phi(\boldsymbol{x})]-y^{**}(\boldsymbol{x})|\leq \kappa_\alpha)\geq 1-\alpha$. 
\end{assumption}
With a guarantee on the quality of the LLM backbone, we have the following theorem: 
\begin{theorem}
Consider a digital population generated by CrowdLLM $\tilde{U}=\{\tilde{u}_i\}_{i=1,\cdots,N}$. For a specific problem $\boldsymbol{x}$, assume the belief biases of $\tilde{U}$ is $\delta_1,\cdots,\delta_N$, with a mean $\mu_{\delta}=\frac{1}{N}\sum_{i=1}^N\delta_i$ and the second moment $\varepsilon^2_\delta=\frac{1}{N}\sum_{i=1}^N\delta_i^2$. Suppose the deviation between the gold-standard response and the reference decision given by the LLM backbone is $\Delta=y^{**}(\boldsymbol{x})-y_{ref}$. We can construct an interval 
\begin{align*}
\mathcal{B}_\alpha(\Delta)=\Bigl[\Delta - h_{\Delta}(\kappa_\alpha), \Delta + h_{\Delta}(\kappa_\alpha)\Bigr],    
\end{align*}
where
\begin{align*}
&h_{\Delta}(\kappa_\alpha)=
\left\{\begin{aligned}
&h_1, \quad   \text{when }\frac{N-2}{N}\sqrt{\frac{(N-2)\varepsilon_{\delta}^2+N\eta(t)}{2}}\geq \kappa_{\alpha}\\
\\
&h_2, \quad \text{when }\frac{N-2}{N}\sqrt{\frac{(N-2)\varepsilon_{\delta}^2+N\eta(t)}{2}}< \kappa_{\alpha},\\
\end{aligned}\right.\\
&h_1=\frac{\sqrt{N^2\kappa_{\alpha}^2+2(N-1)[(N-2)\varepsilon_{\delta}^2+N\eta(t)]}-(N-2)\kappa_{\alpha}}{2(N-1)},\\
&h_2=\frac{\sqrt{2[(N-2)\varepsilon_{\delta}^2+N\eta(t)]}}{N},
\end{align*}
such that when $\mu_{\delta}\in \mathcal{B}_\alpha(\Delta)$, with probability at least $1-\alpha$, CrowdLLM leads to a smaller expected risk than its pure LLM-based counterpart, i.e., $L\leq L'$.
\label{thm3}
\end{theorem}
Theorem \ref{thm3} indicates that to ensure that CrowdLLM outperforms the LLM backbone, the mean belief bias $\mu_{\delta}$ should not be too far away from $\Delta$ (i.e., $\Delta$ quantifies the deviation of the reference decision by the LLM backbone to the ground truth). It is easy to see that the interval width merely relies on $h_{\Delta}(\kappa_\alpha)$. When the size of digital population $N$ is small, $h_{\Delta}(\kappa_\alpha)=h_2$ does not depend on $\kappa_\alpha$, but is instead directly affected by $N$. When $N$ is large enough, $h_{\Delta}(\kappa_\alpha)=h_1$ is monotonically decreasing in $\kappa_\alpha$, which means that if $\kappa_\alpha$ turns larger, we need a tighter interval that covers the belief bias to ensure $L\leq L'$. It suggests that if CrowdLLM is built with a lower-quality LLM backbone, the mean belief bias $\mu_{\delta}$ needs to be closer to $\Delta$, whereas a higher-quality LLM backbone will provide greater tolerance to guarantee CrowdLLM's performance. Moreover, the diversity of the digital population, reflected in $\varepsilon^2_\delta$ and $\eta_{\Phi}(t)$, can somehow alleviate these constraints and offer even more capacity for CrowdLLM to surpass the LLM backbone. This theoretical analysis also reveals why LLM itself can't generate the needed diversity, since it is not just statistical derivations from a mean but needs to be productive in a specific context (e.g., which corresponds to a diverse distribution of participants' profiles). We know that an LLM can balance coherence and novelty in its responses by adjusting the temperature $t$, but from Theorem 3, we see that increasing $t$ to improve novelty in responses can actually enlarge the gap between the LLM and CrowdLLM since it leads to more incoherence and noise of the LLM backbone. In contrast, in CrowdLLM, more diversity associated with a larger $\varepsilon^2_\delta$ might also result in an increase in $h_{\Delta}(\kappa_\alpha)$, which allows $\mu_{\delta}$ to deviate more from $\Delta$ while still ensuring the superiority of CrowdLLM. This answers Q3 as well.

We can further develop a confidence interval for CrowdLLM to cover the ground truth $y^*$. First, we restate the Theorem 1 in \citep{angelopoulos2023prediction} in the context of our problem as follows:
\begin{theorem}
Given any specific task $\boldsymbol{x}$, suppose $\theta^*\triangleq y^{*}=\mathbb{E}[y|\boldsymbol{x}]$ is the population mean response to be estimated. Consider a model $f$ learned from the data to predict $y$. With $\alpha\in(0,1)$ and $\gamma\in(0,\alpha)$ fixed, suppose that for any possible $\theta$, we can construct confidence sets $B^1_{\gamma}(\theta)$ and $B^2_{\alpha-\gamma}(\theta)$ satisfying $ P(\mathbb{E}_{\boldsymbol{v}\sim\mathcal{T}}[f(\boldsymbol{x},\boldsymbol{v})-y]\in B^1_{\gamma}(\theta))\geq 1-\gamma$ and $P(\mathbb{E}_{\boldsymbol{v}\sim\mathcal{T}}[\theta-f(\boldsymbol{x},\boldsymbol{v})]\in B^2_{\alpha-\gamma}(\theta))\geq 1-(\alpha-\gamma)$. Let $B^{y}_\alpha=\{\theta|\exists \theta_1\in B^1_{\gamma}(\theta),\theta_2\in B^2_{\alpha-\gamma}(\theta)\;\; s.t.\;\; \theta_1+\theta_2=0\}$. Then, we have $P(\theta^*\in B^{y}_\alpha)\geq 1-\alpha$.
\end{theorem}
Inspired by this result, we can further show the following theorem:
\begin{theorem} 
Consider CrowdLLM with an LLM backbone $\Phi$ and a belief generator $\delta$ under an additive blender. Assume $\Phi$ has the randomness $\eta(t)$ that can only be changed through adjusting the temperature $t$. Fix $\alpha\in(0,1)$ and $\gamma\in(0,\alpha)$. Given any specific task $\boldsymbol{x}$, suppose we have $n$ real human responses $y_1,\cdots, y_n$ taking numeric values for this task in the training data. Consider a digital population of size $N$ generated by CrowdLLM, $\tilde{U}=\{u_1,\cdots, u_N\}$ with profiles $\boldsymbol{v}_1,\cdots, \boldsymbol{v}_N$. Here, we assume $\frac{n}{N}\to p\in(0,1)$. Suppose their decisions are $\tilde{y}_1,\cdots, \tilde{y}_N$, where $\tilde{y}_i=\Phi^{(i)}(\boldsymbol{x})+\delta_i$ with $\Phi^{(i)}(\boldsymbol{x})$ being the $i$th response sampled from $\Phi$ and $\delta_i=\delta(\boldsymbol{x},\boldsymbol{v}_i)$ being the belief biases. Define $\sigma^2_{\delta}=\frac{1}{N}\sum_{i=1}^N(\delta_i-\overline{\delta})^2$ and $\sigma^2_{r}=\frac{1}{n}\sum_{i=1}^n(r_i-\overline{r})^2$ where $\overline{\delta}=\frac{1}{N}\sum_{i=1}^N\delta_i$, $r_i=y_i-\tilde{y}_i$ and $\overline{r}=\frac{1}{n}\sum_{i=1}^nr_i$. Suppose the training error can be bounded by a small tolerance $\varepsilon_0^2$, i.e., $\frac{1}{n}\sum_{i=1}^n(y_i-\tilde{y}_i)^2\leq \varepsilon_0^2$. Then, we can build a confidence interval centered on the aggregated decision $\overline{\tilde{y}}=\frac{1}{N}\sum_{i=1}^N\tilde{y}_i$:
\begin{align*}
    B^y_\alpha=\Bigl\{y:|y-\overline{\tilde{y}}|\leq \varepsilon_0+z_{1-\frac{\alpha}{2}}\sqrt{\frac{\eta(t)}{N}+\frac{\sigma^2_{\delta}}{N}+\frac{\sigma^2_{r}}{n}}\Bigr\},
\end{align*}
where $z_{1-\frac{\alpha}{2}}$ is the $1-\frac{\alpha}{2}$ quantile of the standard normal distribution, such that the ground truth $y^*$ satisfies
\begin{align*}
   \mathop{\lim\inf}_{n,N\to\infty} P(y^*\in B^y_{\alpha})\geq 1-\alpha.
\end{align*}
\end{theorem}
The above theorem indicates the effectiveness of CrowdLLM under asymptotic settings. From the interval built in this theorem, we notice the uncertainty in estimating $y^*$ is mainly contributed by three parts: the randomness of LLM backbone, the variance of belief biases, and the variance of the residuals $\sigma^2_{r}$. Since the LLM backbone is frozen and will only show a certain randomness controlled by the temperature $t$, the uncertainty from this component is fully contributed by $\eta(t)$. When $N$ grows, the average converges and the uncertainty is gradually reduced. The uncertainty contributed by the belief generator is rooted in the belief biases. $\frac{\sigma^2_{\delta}}{N}$ measures the variability of the average belief bias. However, this uncertainty will not be explosively increasing, since the average belief bias will gradually converge to the true deviation of LLM's decision from the ground truth decision. As $N$ grows, the diversity matters less and $\frac{\sigma^2_{\delta}}{N}$ is less influential. Beyond the first two sources of uncertainty, the quality of the belief generator determines how large the uncertainty will be. If the belief generator can well characterize the real human data, the individual residuals $r_i$ should be small, which indicates the decisions given by CrowdLLM can well approximate the real human data, and $\sigma_r^2$ should be small when the individual residuals are consistently small. If $n$ is large, the real human data used for training is large, which shrinks this variance and leads to a more accurate model.

\section{Case Studies}\label{experiments}
In this section, we conduct thorough experiments to evaluate the ability of CrowdLLM to generate data of human-grade quality across different applications, including crowdsourcing \citep{vaughan2018making}, collecting product reviews from users \citep{isinkaye2015recommendation}, and voting \citep{yang2024designing}. Each application entails a different kind of human data, but all involve a set of decision-making tasks that are attributed to a population of workers (in crowdsourcing), users (in recommendation systems), or voters (in politics). We evaluate CrowdLLM and its vanilla versions and some other benchmark methods, including LLM-based and non-LLM methods based on a range of performance metrics such as accuracy, diversity, fidelity to real human data, sample efficiency, and cost. We also thoroughly study different configurations of CrowdLLM and its various components (i.e., how prompting strategies are employed to generate the LLM-based virtual human participants) and study its sample efficiency (i.e., how much training data is needed to reach a superior performance). In what follows, we introduce details of our experimental evaluations and main findings.

\subsection{Experiment Setup}

\subsubsection{Overall Design and Evaluation Method.}

Recall that CrowdLLM generates a collection of decisions from the synthetic participants for some given tasks, such as voting for alternatives, rating products, and evaluating texts. To evaluate how well its outcome matches the data collected in a real human population, we assess its performance in both aggregated decision (accuracy) and the distribution of decisions (diversity). Specifically, we use \textit{Average Wasserstein Distance (Avg. WD)} to measure the average of the distributional deviation from the gold standard across each test problem. A lower value of this metric indicates better distributional similarity and a more effective characterization of population diversity. For the evaluation of aggregated decisions, we consider evaluation metrics such as \textit{Mean Absolute Error (MAE)}, \textit{Root Mean Square Error (RMSE)}, and \textit{Cosine Similarity (CS)}. The first two metrics emphasize the average performance over different tasks, while CS focuses on the general comparison with the gold standard (i.e., real human data) on the whole test set. For Crowdsourcing tasks, we follow prior work \citep{veselovsky2025prevalence} and adopt the common practice of considering various aggregation approaches that include mean for all the case studies, and also more specialized ones for crowdsourcing applications such as the majority voting (MV), Dawid-Skene (DS), and Generative model of Labels, Abilities, and Difficulties (GLAD) \citep{whitehill2009whose}. Majority voting is a basic label aggregation approach that selects the label chosen by most participants. The DS improves on this by learning how accurate each participant is from their labeling history, giving more weight to reliable ones. GLAD goes further by also considering task difficulty, using a probability model to combine participant's ability and task difficulty for more robust label estimation. 

\subsubsection{Implementation of CrowdLLM.}
Unless otherwise specified, we use Gemma3-12B \citep{team2025gemma} as the LLM backbone of CrowdLLM due to its strong performance, its reliable capabilities across diverse tasks, and manageable size which facilitates extensive experimentation. In each of our case studies, we also conduct a comparison of Gemma3-12B with three other prominent LLMs, Deepseek-Distill-R1-Llama-8B (Deepseek R1) \citep{guo2025deepseek}, Llama3-8B-lexi-uncensored \citep{dubey2024llama}, and Qwen3-8B \citep{yang2025qwen3technicalreport}, and found Gemma3-12B indeed outperforms others. In all the experiments, we set the temperature to 0 for CrowdLLM, and set the reference generation parameter $K$ to $8$, the offset parameter $J$ for personalized decision-making in training to $10$, and the regularization controller $\lambda$ to $1$. These settings are based on our extensive empirical experiments across datasets which consistently provide good performances of CrowdLLM. We use Adam \citep{kingma2014adam} as the optimizer with a learning rate $0.001$ to train CrowdLLM with Eq. \eqref{eq:loss}. Another important aspect of implementing CrowdLLM on a particular application is the profile generation. Recall that the profile generator aims to generate diverse profiles for the synthetic human participants. The individual profile will be used as the contextual information which is further fed into the belief generator of CrowdLLM. For a given application, one can either obtain summary statistics of the human participants or design an ideal distribution of the profile variables that are representative of a real crowd. Figure~\ref{figs0} shows an example of such a distribution of participants' profiles based on 5 variables: \textit{Gender}, \textit{Age}, \textit{Race}, \textit{Occupation}, \textit{Education} that we used in Case Study I. 

\begin{figure}[!htbp]
    \centering
    \includegraphics[width=0.5\textwidth]{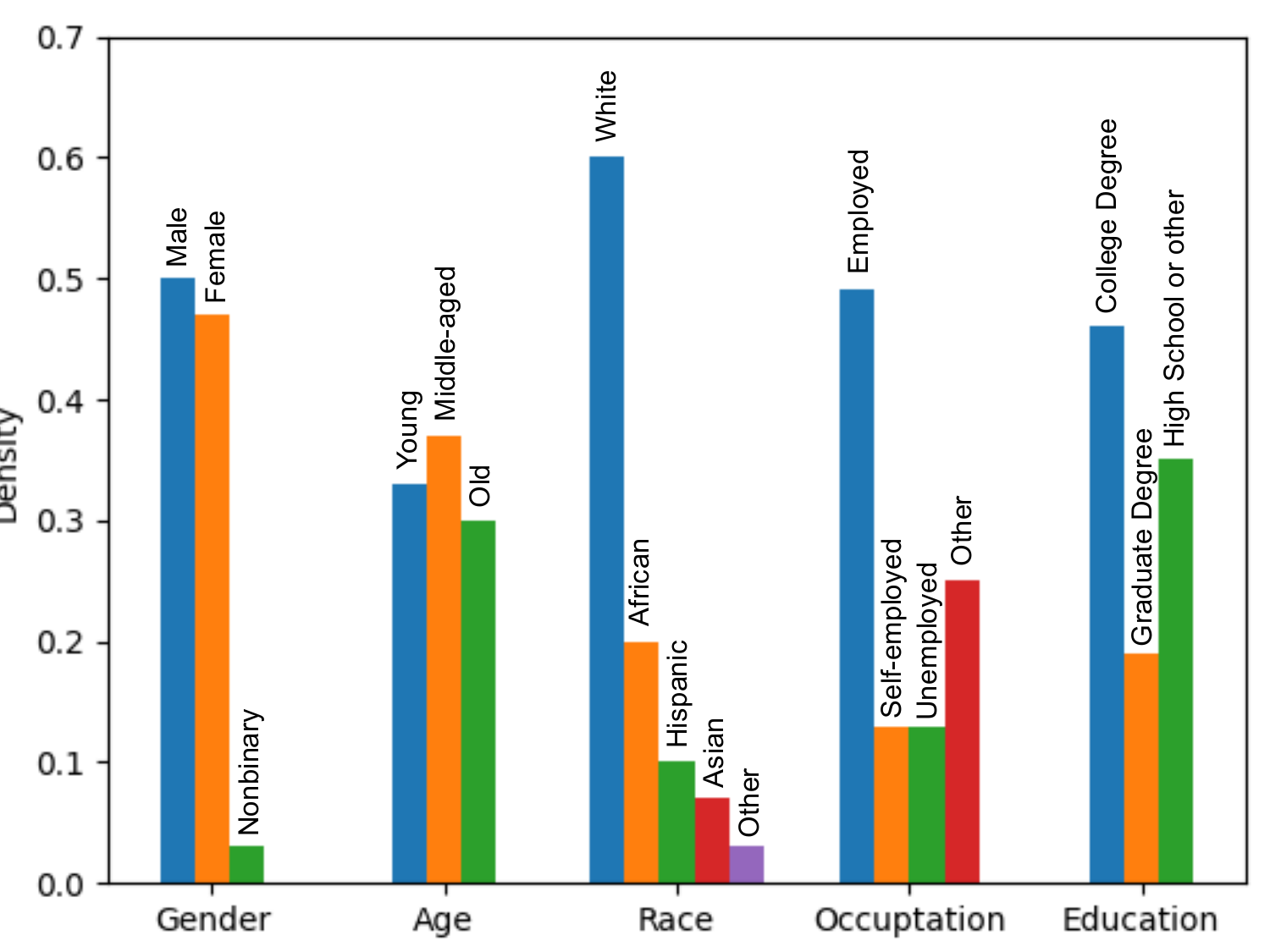}
    \caption{An example of the distribution of participants' profiles.}
    \label{figs0}
\end{figure}

\subsubsection{Baselines.} We adopt a variety of prompting strategies to create representative baselines of LLM-based synthetic crowd. This includes the zero-shot (direct) prompting, referred as LLM (zero-shot), which prompts the LLM to generate decisions on problems with the most basic information. We also include multi-persona prompting \citep{li2025llm, hu2024quantifying} and self-consistency (SC) \citep{wangself}. Multi-persona prompting involves a collaboration of multiple LLMs prompted with different profiles to create diverse personas. Self-consistency prompting is performed by taking the majority vote of multiple decisions generated by LLM (i.e., with temperature set at 0.5) following the practice in \citep{wangself, liu2024dellma}. More specifically, the zero-shot prompt provides LLM with problem description $\boldsymbol{x}_t$, specific requirements $\mathcal{R}_t$ on the decisions, and context $\mathcal{C}_t$. Here, $\boldsymbol{x}_t$ describes the problem that the participants need to make decisions on. $\mathcal{R}_t$ describes what kind of decisions need to be made, e.g., if the decision is a score, the scale of the score should be included in $\mathcal{R}_t$. $\mathcal{C}_t$ encodes information about the decision-making scenario. An example of the zero-shot prompt is shown in the Appendix. For multi-persona prompting, i.e., using multiple personas in prompting for self-collaboration as shown in \cite{olea2024evaluating, hu2024quantifying}, we build multiple personas with additional context on their profiles. We use similar prompts as zero-shot prompting for making decisions, but change the context $\mathcal{C}_t$ to assign persona. An example of a multi-persona prompt is shown in the Appendix. For self-consistency prompting, we follow the strategy adopted in \cite{wangself, liu2024dellma} by taking the majority vote of multiple decisions generated by LLM with temperature 0.5. For the generation of each decision, we use the same prompt as zero-shot prompting. We can certainly adopt more prompting strategies, such as the CoT prompting in \cite{wei2022chain}. In our experiments, we found no significant differences among the prompting strategies, and our focus is not on what the best prompting strategies in LLMs are to simulate humans, but the integration of those LLM-based frameworks with generative models, so throughout our experiments, we use zero-shot, multi-persona, and SC prompts.




\subsection{Case Study I: Crowdsourcing}

We evaluated CrowdLLM and the other baselines using two publicly available crowdsourcing datasets, \textit{Offensiveness Rating} and \textit{Question Answering Difficulty} \citep{pei-jurgens-2023-annotator}. The \textit{Offensiveness} dataset contains 13,036 instances annotated by 263 workers across 1,500 problems to identify offensive text. \textit{QA Difficulty} includes 4,576 instances annotated by 458 workers of 1,000 problems to assess question-answer pair difficulty. For each dataset, we use responses from 80\% of the distinct workers as the training set and reserve the remaining 20\% of workers’ responses for testing.






\begin{table}[!t]
\centering
\caption{Performance Comparison on Crowdsourcing: Offensiveness Rating}
\label{tab:parallel_metrics Offensiveness} 
\resizebox{\textwidth}{!}{\begin{tabular}{lcccccccccccc}
\hline
 \multirow{2}{*}{Method} & \multicolumn{5}{c}{MAE} & \multicolumn{5}{c}{RMSE} & \multirow{2}{*}{CS} & \multirow{2}{*}{Avg. WD} \\ 
\cline{2-6}\cline{7-11} 
 &Mean &Median &MV &DS &GLAD &Mean &Median &MV &DS &GLAD &&\\
\hline
Random & 1.27 & 1.61 & 1.57 & 1.73 & 1.89 & 1.44 & 1.86 & 2.05 & 2.17 & 2.31 & 0.56 & 1.31 \\
LLM (zero-shot) & 0.86 & 1.03 & 1.23 & 1.06 & 1.23 & 1.08 & 1.32 & 1.53 & 1.38 & 1.52 & 0.39 & 1.13 \\
LLM (multi-persona) & 0.65 & 0.64 & 0.67 & \textbf{0.69} & 0.68 & 0.86 & 1.00 & 1.16 & \textbf{1.15} & \textbf{1.16} & 0.69 & 0.78 \\
LLM (SC) & 0.99 & 1.26 & 1.48 & 1.22 & 1.49 & 1.16 & 1.48 & 1.68 & 1.52 & 1.69 & 0.34 & 1.24 \\
VAE   & 0.71 & 0.81 & 0.60 & 0.96 & 0.92 & 0.94 & 1.15 & 1.24 & 1.37 & 1.54 & 0.76 & 0.79  \\ 
CrowdLLM & \textbf{0.45} & \textbf{0.48} & \textbf{0.43} & 1.00 & \textbf{0.53} & \textbf{0.59} & \textbf{0.82} & \textbf{1.10} & 1.53 & 1.23 & \textbf{0.85} & \textbf{0.51} \\
\hline
\hline
\end{tabular}}
\vspace{1mm}
\parbox{0.95\linewidth}{\small \textit{LLM backbone: Gemma 3-12B.}}
\end{table}%

\begin{table}[!t]
\centering
\caption{Performance Comparison on Crowdsourcing: QA Difficulty}
\label{tab:parallel_metrics Difficulty} 
\resizebox{\textwidth}{!}{\begin{tabular}{lcccccccccccc}
\hline
 \multirow{2}{*}{Method} & \multicolumn{5}{c}{MAE} & \multicolumn{5}{c}{RMSE} & \multirow{2}{*}{CS} & \multirow{2}{*}{Avg. WD} \\ 
\cline{2-6}\cline{7-11} 
 &Mean &Median &MV &DS &GLAD &Mean &Median &MV &DS &GLAD &&\\
\hline
Random & 1.21 & 1.41 & 1.52 & 1.98 & 1.81 & 1.43 & 1.73 & 2.02 & 2.38 & 2.26 & 0.49 & 1.29 \\
LLM (zero-shot) & 0.71 & 0.81 & 1.04 & 1.02 & 1.05 & 0.90 & 1.04 & 1.24 & 1.28 & 1.26 & 0.33 & 1.06 \\
LLM (multi-persona) & 0.73 & 0.82 & 1.00 & 1.06 & 1.02 & 1.02 & 1.16 & 1.38 & 1.46 & 1.41 & 0.48 & 0.93 \\
LLM (SC) & 0.68 & 0.79 & 1.02 & \textbf{0.98} & 1.03 & 0.87 & \textbf{1.02} & 1.22 & \textbf{1.23} & \textbf{1.25} & 0.36 & 1.01 \\
VAE & 0.76 & 0.90 & 0.71 & 1.41 & 1.11 & 0.98 & 1.20 & \textbf{1.14} & 1.85 & 1.53 & 0.68 & 0.88  \\ 
CrowdLLM & \textbf{0.65} & \textbf{0.74} & \textbf{0.63} & 1.49 & \textbf{0.89} & \textbf{0.83} & 1.06 & 1.18 & 2.08 & 1.43 & \textbf{0.71} & \textbf{0.74} \\

\hline
\hline
\end{tabular}}
\vspace{1mm}
\parbox{0.95\linewidth}{\small \textit{LLM backbone: Gemma 3-12B.}}
\end{table}%

\begin{figure*}[!b]
    \centering
    \begin{subfigure}{0.5\textwidth}
        \centering
        \includegraphics[width=\linewidth]{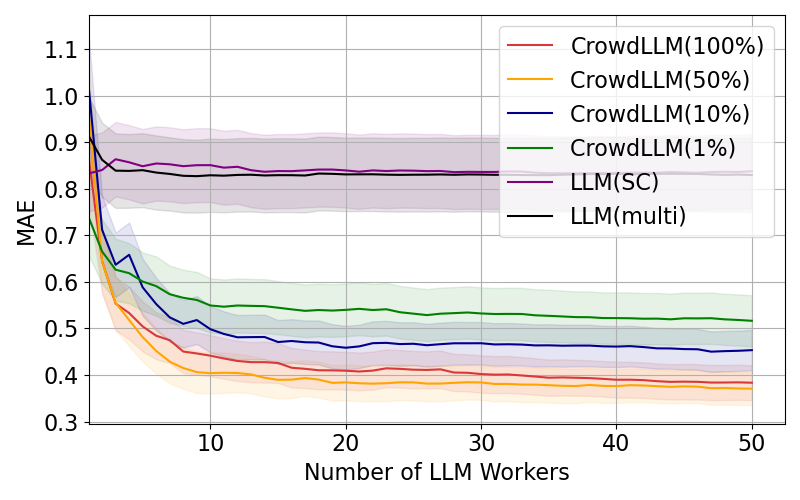}
        \caption{Offensiveness dataset} 
        \label{Offensive_MAE}
    \end{subfigure}%
    \hfill
    \begin{subfigure}{0.5\textwidth}
        \centering
        \includegraphics[width=\linewidth]{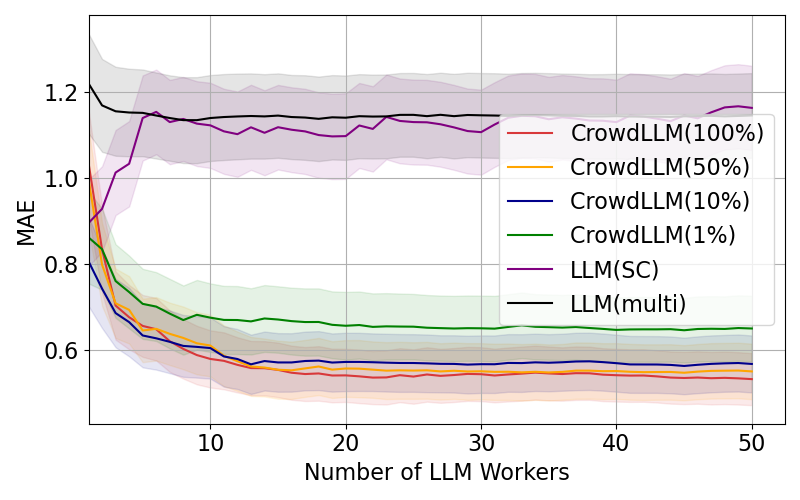}
        \caption{QA Difficulty dataset} 
        \label{Difficulty_MAE}
    \end{subfigure}%
    \caption{MAE with increasing simulated workers across training worker sizes; CrowdLLM (x\%) means the model is trained with x\% of the real human workers' data}
    \label{MAE} 
\end{figure*}
\begin{figure*}[!b]
    \centering
    \begin{subfigure}{0.5\textwidth}
        \centering
        \includegraphics[width=\linewidth]{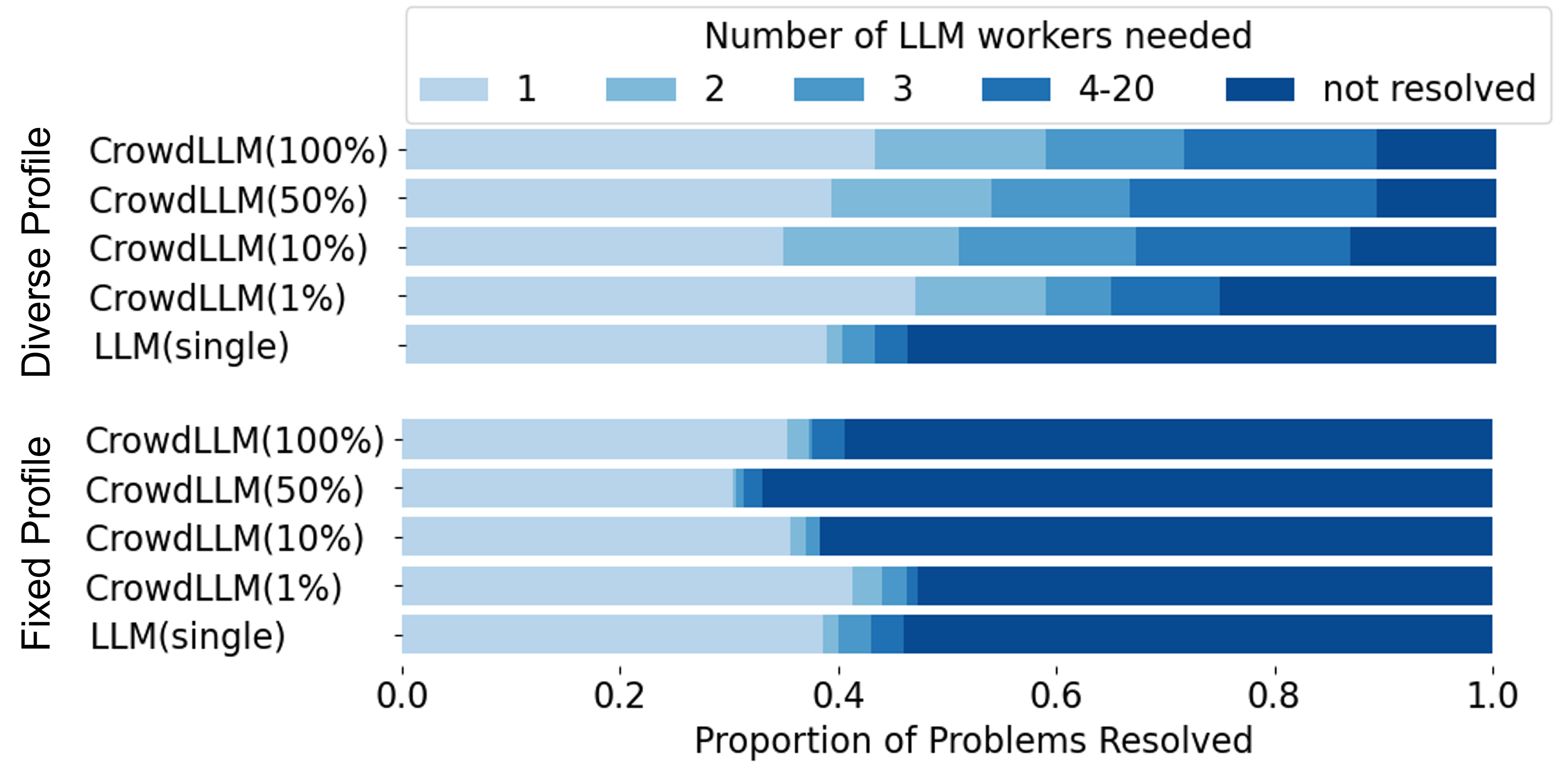}
        \caption{Offensiveness dataset} 
        \label{Offensive_cost}
    \end{subfigure}%
    \hfill
    \begin{subfigure}{0.5\textwidth}
        \centering
        \includegraphics[width=\linewidth]{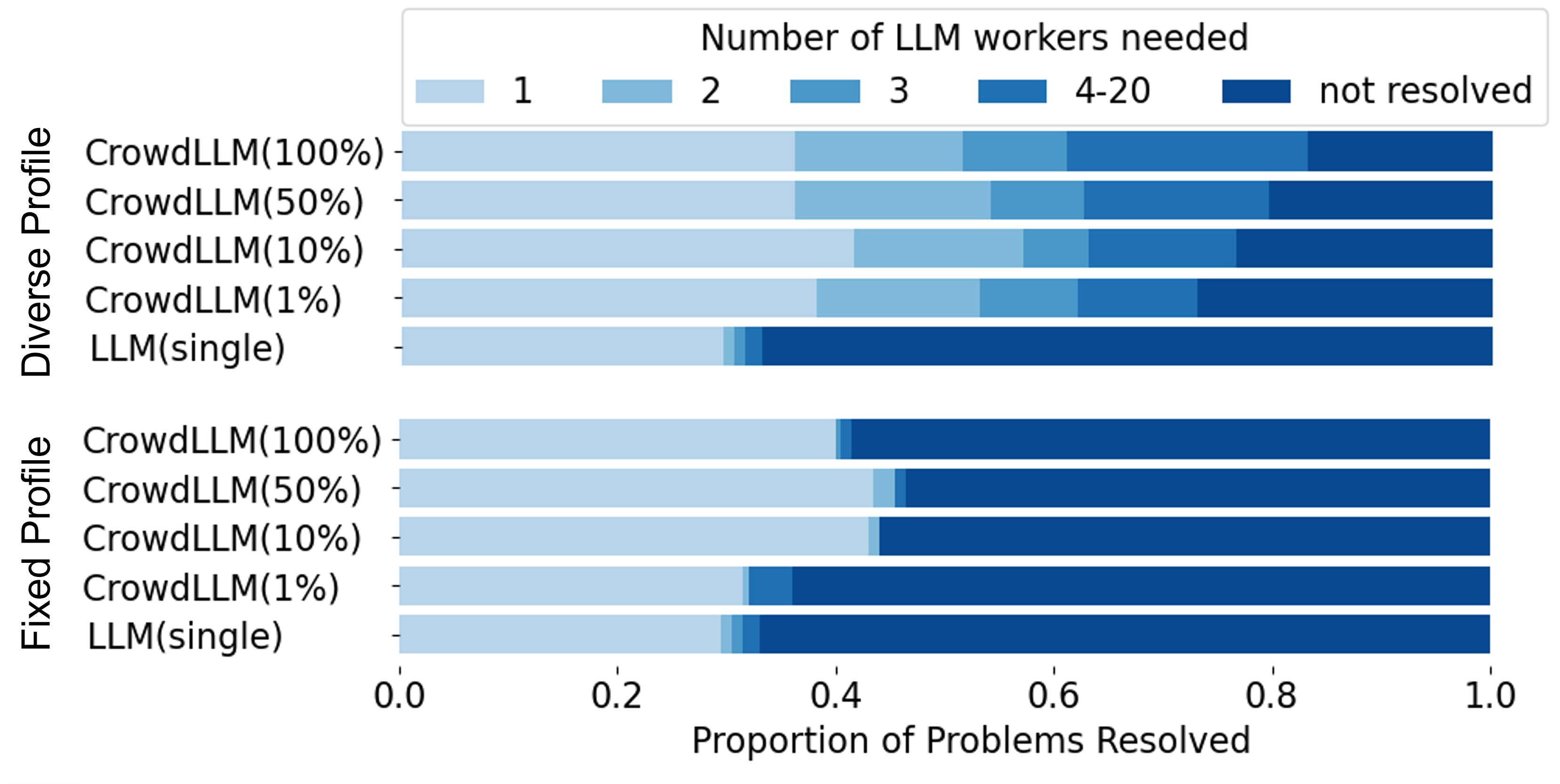}
        \caption{QA Difficulty dataset} 
        \label{Difficulty_cost}
    \end{subfigure}%
    \caption{Resolution rate with increasing simulated workers in CrowdLLM under diverse and fixed profiles; CrowdLLM (x\%) means the model is trained with x\% of the human workers' data}
    \label{cost} 
\end{figure*}

Tables~\ref{tab:parallel_metrics Offensiveness}-\ref{tab:parallel_metrics Difficulty} show that CrowdLLM consistently outperforms other baselines by achieving the lowest Avg. WD and the highest CS. This shows that our CrowdLLM can better capture real human diversity. We can also see that by incorporating diverse profile information to allow for more variability in decisions, even LLM (multi-persona) shows better performances compared to LLM (zero-shot), though the improvement is not as substantial as in CrowdLLM. And pure generative models without LLM (i.e., the VAE model) often outperform the baseline LLM (zero-shot) method in terms of Avg. WD, highlighting the power of generative models in diversifying predictions that improve CrowdLLM's alignment with real humans' responses. In terms of the aggregated decision, CrowdLLM generally offers competitive or improved MAE/RMSE compared with other approaches. On both \textit{Offensiveness} and \textit{QA Difficulty} datasets, CrowdLLM's MAEs are notably better than the baseline LLM methods under most of the aggregation methods. 

One might be interested in how much real human data is needed to train CrowdLLM well. To answer this question, we further conduct some computational experiments and show the results in Figure~\ref{MAE}. Specifically, for each task/instance in the crowdsourcing case studies, we incrementally add more workers and monitor CrowdLLM's performance.  Figure~\ref{MAE} shows that the performance of CrowdLLM (i.e., evaluated by MAE) consistently improves as the number of workers increases. Note that in Figure~\ref{MAE} the labels, CrowdLLM (x\%), mean the model is trained with x\% of the human workers' data. In contrast to the low diversity and high uncertainty exhibited by LLM baselines, it is impressive to see that, even with only 1\% of the human workers which collectively provided around 100 responses in the training data, CrowdLLM can achieve the level of performance comparable to the full-data setting where CrowdLLM is trained with all the training data (100\%), highlighting the data efficiency and cost-saving potential of CrowdLLM.

Another question one may ask is how many virtual human workers CrowdLLM needs to generate to resolve the problems (i.e., reducing absolute error below 0.5)? We conduct more experiments as well to answer this question and show the results in Figure~\ref{cost}. We can see that, across the different levels of the training size, the resolution rate increases steadily as the number of virtual workers with diverse profiles grows, whereas on the other hand, if we fix the worker profiles, it leads to a significant degradation. It underscores the value of promoting diversity among the virtual workers, which is a core strength of CrowdLLM. Recall that in the overall design of CrowdLLM, profiles serve as proxies for workers' beliefs: different profiles correspond to different backgrounds and thus to diverse beliefs about a task. When the virtual workers have diverse profiles (i.e., by generating $\epsilon_i$ following Eq. \eqref{eq:zi}), their varying beliefs allow the collective answer to be iteratively refined, and one can see in Figure~\ref{cost} that the resolution rate increases steadily as the number of virtual workers grows. In contrast, under fixed profiles (i.e., by setting $\epsilon_i$ in Eq. \eqref{eq:zi} as a fixed number across individuals), the virtual workers hold nearly identical beliefs, so adding more of them does not improve the solution. This result highlights that the performance improvement of CrowdLLM comes more from the diverse profiles rather than simply from having more virtual workers.

\subsection{Case Study II: Product Ratings by Users}


In this subsection, we showcase the effectiveness of CrowdLLM on generating product reviews that can be used to train recommendation systems. We consider \textit{Amazon Beauty} and \textit{Amazon Music}, two datasets extracted from the Amazon Reviews 2023 dataset \citep{hou2024bridging, mcauley2015image}. \textit{Amazon Beauty} is the subset of ``All Beauty" category in which each product has been reviewed by 20 to 30 distinct users. It contains 448 products, with a total of 11,154 reviews from 10,957 unique users. \textit{Amazon Music} is a filtered subset of the ``Musical Instruments" category containing products reviewed by exactly 20 distinct users. It comprises 629 products, encompassing 12,580 reviews written by 12,396 unique users. For both datasets, we perform necessary preprocessing steps and hold out 20\% of the full data as a test set based on unique problem IDs. All the experimental results are reported based on the held-out test set. The objective of CrowdLLM is to generate the distribution of ratings of each product by providing its product information to the digital population CrowdLLM creates. Results of CrowdLLM and the other methods are shown in Tables~\ref{tab:Recommendation Beauty} and ~\ref{tab:Recommendation_Music}. All methods use the same product information. LLM multi-persona additionally conditions on user-evaluation text to simulate more diverse responses of the virtual users. CrowdLLM also incorporates the user-evaluation text into its training process. We can see in Tables~\ref{tab:Recommendation Beauty} and ~\ref{tab:Recommendation_Music} that on both datasets, CrowdLLM achieves the best performance in terms of all the evaluation metrics. We further show in Figures ~\ref{Amazon beauty} and ~\ref{Amazon music} the generated distributions of the product ratings by CrowdLLM and the other methods (i.e., we randomly selected three products from each of the two datasets), together with the distribution of real human users. We can see that LLMs, even when we adjust temperature or provide user-evaluation text in multi-persona prompting, still exhibit limited diversity in their generated ratings. VAE, by contrast, produces more diverse outputs but falls short of accurately matching true human ratings. CrowdLLM achieves a better balance: it captures both diversity and accuracy, resulting in rating distributions that closely align with those of human participants.

\begin{table}[!t]
\centering
\caption{Performance Comparison on Recommendation(All Beauty)}
\label{tab:Recommendation Beauty} 
\begin{tabular}{lcccc}
\hline
 Method &  MAE &  RMSE &   CS &   WD \\
\hline
random & 1.19 & 1.30 & 0.59 & 0.14 \\
LLM (zero-shot) & 1.04 & 1.19 & 0.13 & 0.15 \\
LLM (multi-persona) & 0.87 & 0.97 & 0.24 & 0.09 \\
LLM (SC) & 1.03 & 1.16 & 0.13 & 0.14 \\
VAE & 0.88 & 1.02 & 0.28 & 0.06 \\
CrowdLLM & \textbf{0.21} & \textbf{0.26} & \textbf{0.93} & \textbf{0.04} \\
\hline
\hline
\end{tabular}
\vspace{1mm}
\parbox{0.55\linewidth}{\small \raggedright \hspace{2.5em} \textit{LLM backbone: Gemma 3-12B.}}
\end{table}%

\begin{table}[!t]
\centering 
\caption{Performance Comparison on Recommendation(Musical Instruments)}
\label{tab:Recommendation_Music} 
\begin{tabular}{lcccc}
\hline
 Method &  MAE &  RMSE &   CS &   WD \\
\hline
           random & 1.28 &  1.39 & 0.58 & 0.15 \\
  LLM (zero-shot) & 0.52 &  0.65 & 0.27 & 0.16 \\
LLM (multi-persona) & 0.47 &  0.58 & 0.37 & 0.14 \\
         LLM (SC) & 0.46 &  0.57 & 0.29 & 0.16 \\
              VAE & 0.47 &  0.61 & 0.63 & 0.07 \\
         CrowdLLM & \textbf{0.22} &  \textbf{0.27} & \textbf{0.92} & \textbf{0.05} \\
\hline
\hline
\end{tabular}
\vspace{1mm}
\parbox{0.55\linewidth}{\small \raggedright \hspace{2.5em} \textit{LLM backbone: Gemma 3-12B.}}
\end{table}%

\begin{figure*}[!b]
    \centering
    \begin{subfigure}{0.33\textwidth}
        \centering
        \includegraphics[width=\linewidth]{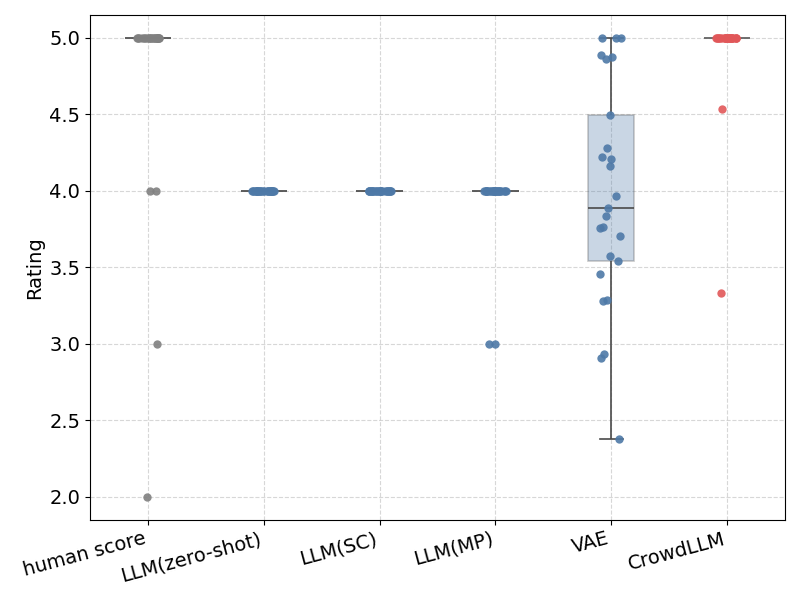}
    \end{subfigure}%
    \hfill
    \begin{subfigure}{0.33\textwidth}
        \centering
        \includegraphics[width=\linewidth]{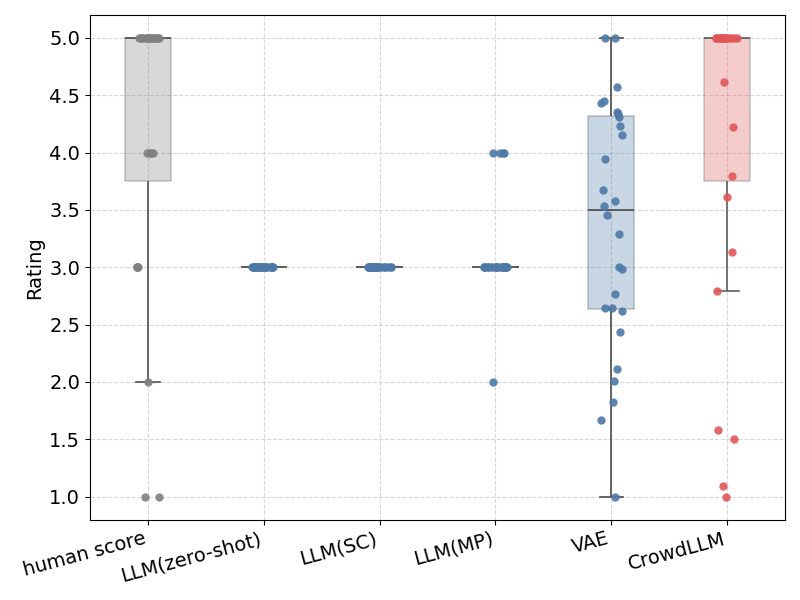}
    \end{subfigure}%
    \hfill
    \begin{subfigure}{0.33\textwidth}
        \centering
        \includegraphics[width=\linewidth]{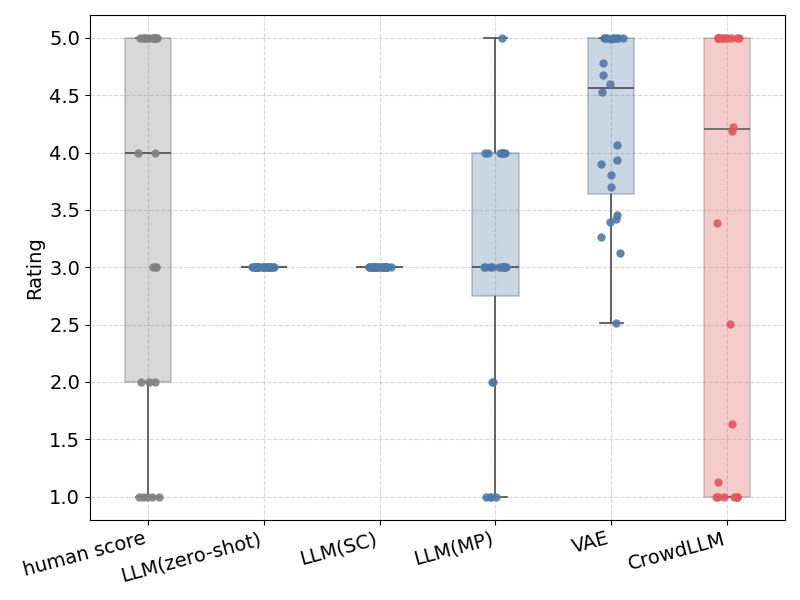}
    \end{subfigure}%
    \caption{Rating scores of three randomly selected products from Amazon Beauty}
    \label{Amazon beauty} 
\end{figure*}
\begin{figure*}[!b]
    \centering
    \begin{subfigure}{0.33\textwidth}
        \centering
        \includegraphics[width=\linewidth]{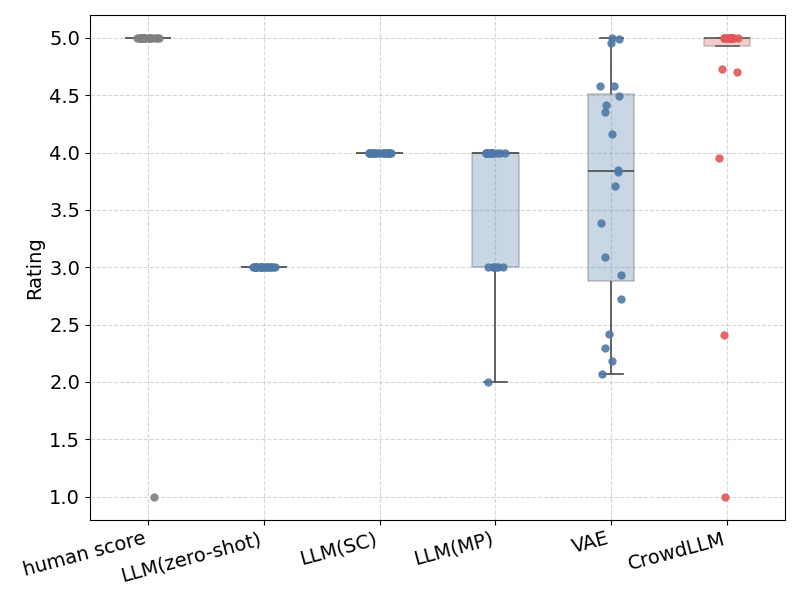}
    \end{subfigure}%
    \hfill
    \begin{subfigure}{0.33\textwidth}
        \centering
        \includegraphics[width=\linewidth]{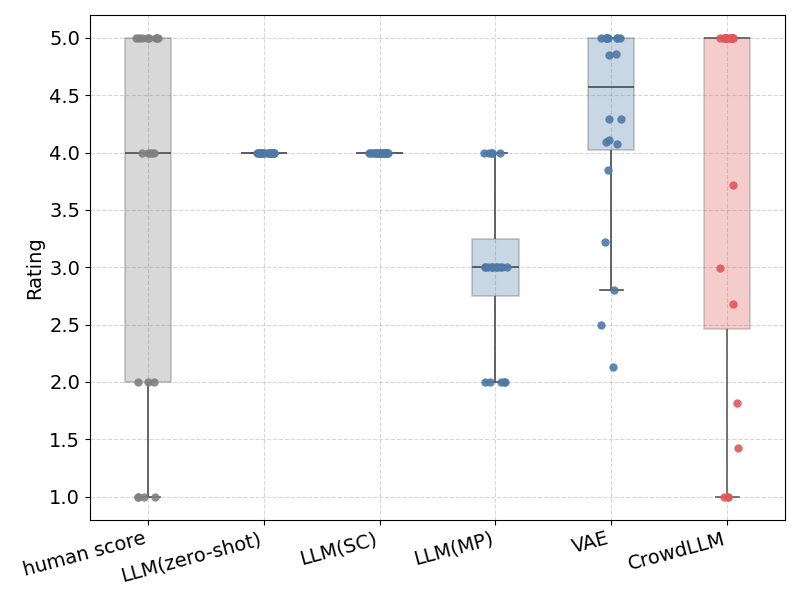}
    \end{subfigure}%
    \hfill
    \begin{subfigure}{0.33\textwidth}
        \centering
        \includegraphics[width=\linewidth]{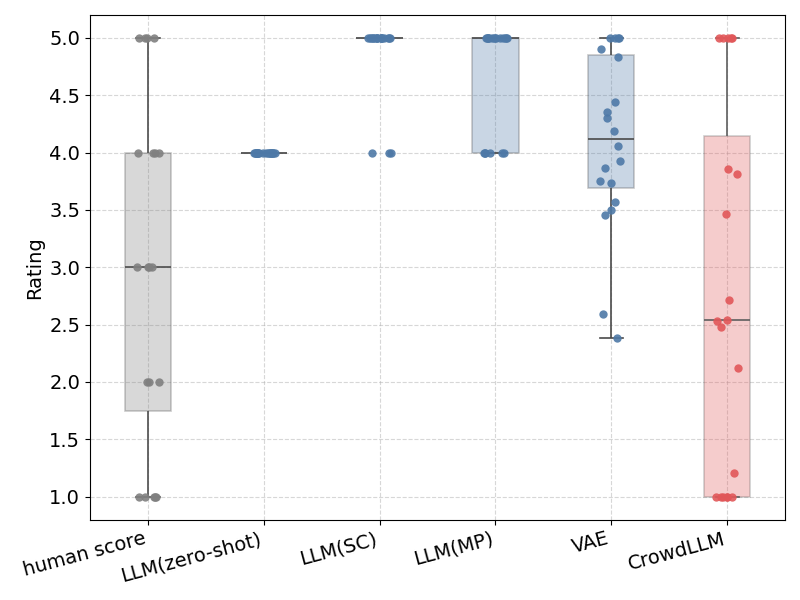}
    \end{subfigure}%
    \caption{Rating scores of three randomly selected products from Amazon Music}
    \label{Amazon music} 
\end{figure*}

\subsection{Case Study III: Voting}
\begin{table}[!t]
\centering
\caption{Performance Comparison on Voting}
\label{tab:Voting} 
\begin{tabular}{lcccc}
\hline
 Method &  MAE &  RMSE &   CS &   WD \\
\hline
           random & 5.00 & 6.55 & 0.72 & 4.00  \\
  LLM (zero-shot) & 11.92 & 15.41 & 0.38 & 9.00 \\
LLM (multi-persona) &  10.00 & 13.69 & 0.40 & 6.50 \\
         LLM (SC)  & 11.42 & 14.97 & 0.35 & 7.92 \\
              VAE & 11.67 & 14.88 & 0.32 & 7.67\\
         CrowdLLM & \textbf{4.00} & \textbf{5.02} & \textbf{0.83} & \textbf{1.67} \\
\hline
\hline
\end{tabular}
\vspace{1mm}
\parbox{0.55\linewidth}{\small \raggedright \hspace{2.5em} \textit{LLM backbone: Gemma 3-12B.}}
\end{table}%

\begin{figure}[!htbp]
    \centering
    \includegraphics[width=0.7\textwidth]{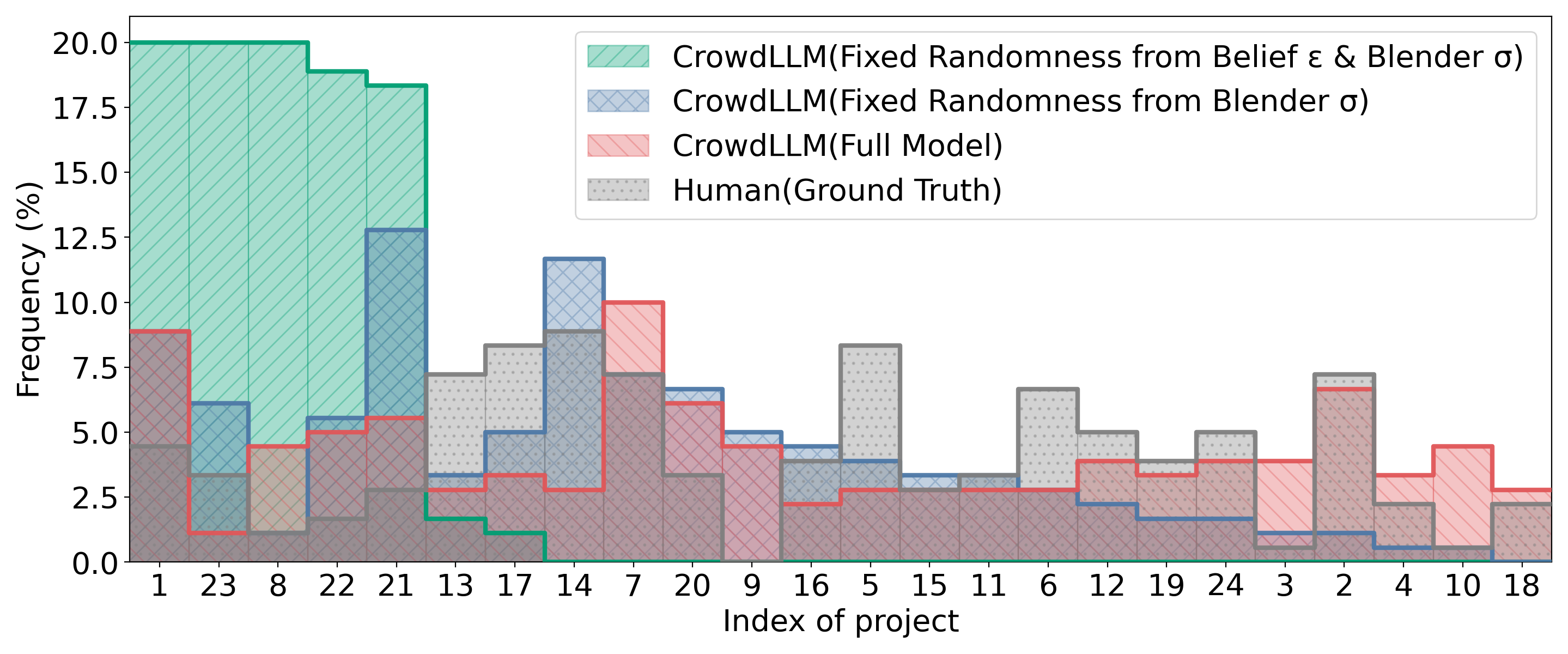}
    \caption{Voting migration across models: from fixed belief with blender to relaxed belief and full CrowdLLM — increasing diversity with subsequent accuracy refinement}
    \label{fig:Voting_belief_blender}
\end{figure}

We further evaluate CrowdLLM and other methods on a voting dataset. This Zurich PB Voting dataset \citep{yang2024llm} was conducted in March 2023 with 180 participants where each participant evaluated 24 projects and expressed preferences through several predefined preference selection methods. Existing work has developed LLM-based synthetic crowds to replicate the voting outcomes in this dataset, such as \cite{yang2024llm}. One problem found in these studies is that, despite the great promise of LLMs in generating human voting results, there is a lack of diversity as LLM-generated votes tend to concentrate heavily on only a few voting options, leaving many other alternatives with no votes at all, as reported in \cite{yang2024llm}. This is a significant shortcoming of the LLM-based virtual voters. Therefore, in this case study, we aim to evaluate CrowdLLM using the same data studied in \cite{yang2024llm}. In our experiments, we adopted the top-five selection method, where each participant selects five preferred projects. The dataset also contains participant-specific preference information, such as project location, topic, and cost, which is incorporated into our modeling. In summary, the dataset contains 180 voting records (five selections each), and we split it with 80\% for training and 20\% for testing. We use LLMs, VAE, random baseline, and CrowdLLM to create a simulated voter population. In each experiment, each virtual vote generates a set of five preferred projects. We then aggregate these generated votes into per-project vote counts and compare them with the real human vote counts, evaluating performance using MAE, RMSE, CS, and WD. Results are shown in Table~\ref{tab:Voting} which shows that CrowdLLM achieves the best results across these metrics. In terms of diversity, LLMs and VAE tend to concentrate on a few projects, with limited variation across participants’ preferences. We also conduct an ablation study to show how different components of CrowdLLM impact its result. Figure~\ref{fig:Voting_belief_blender} illustrates the voting results of different configurations of the CrowdLLM. Under fixed beliefs, i.e., CrowdLLM (fixed randomness from belief $\epsilon$ and blender $\sigma$), voters cast in their votes in the same way, which is similar to the behavior of LLMs. By allowing each voter’s belief to be randomly generated, we can see that CrowdLLM (fixed randomness from blender $\sigma$) improves diversity, yet still there are several projects with very few votes. In contrast, the full CrowdLLM model not only preserves diversity but also accurately captures the few votes for the less popular projects, leading to greater consistency with real human voting patterns.

\subsection{Performance Study Using Simulated Datasets}
\label{Sec:Discussion}

The three real-world case studies demonstrated CrowdLLM's effectiveness in generating digital populations that have high fidelity to real human populations. In this subsection, we aim to further study which factors of the training dataset mostly impact CrowdLLM's effectiveness. These factors concern the signal-to-noise level of the data, the sample size, etc. We generate datasets by different combinations of four factors: the number of participants, the number of tasks assigned to each participant, the accuracy of participants' responses, and the diversity of participants’ beliefs. We evaluate CrowdLLM's performance by varying these factors together with changing the signal-to-noise level and sample size of the training data, etc.



\subsubsection{Simulation Design.}

\begin{figure*}[!b]
    \centering    \includegraphics[width=\linewidth]{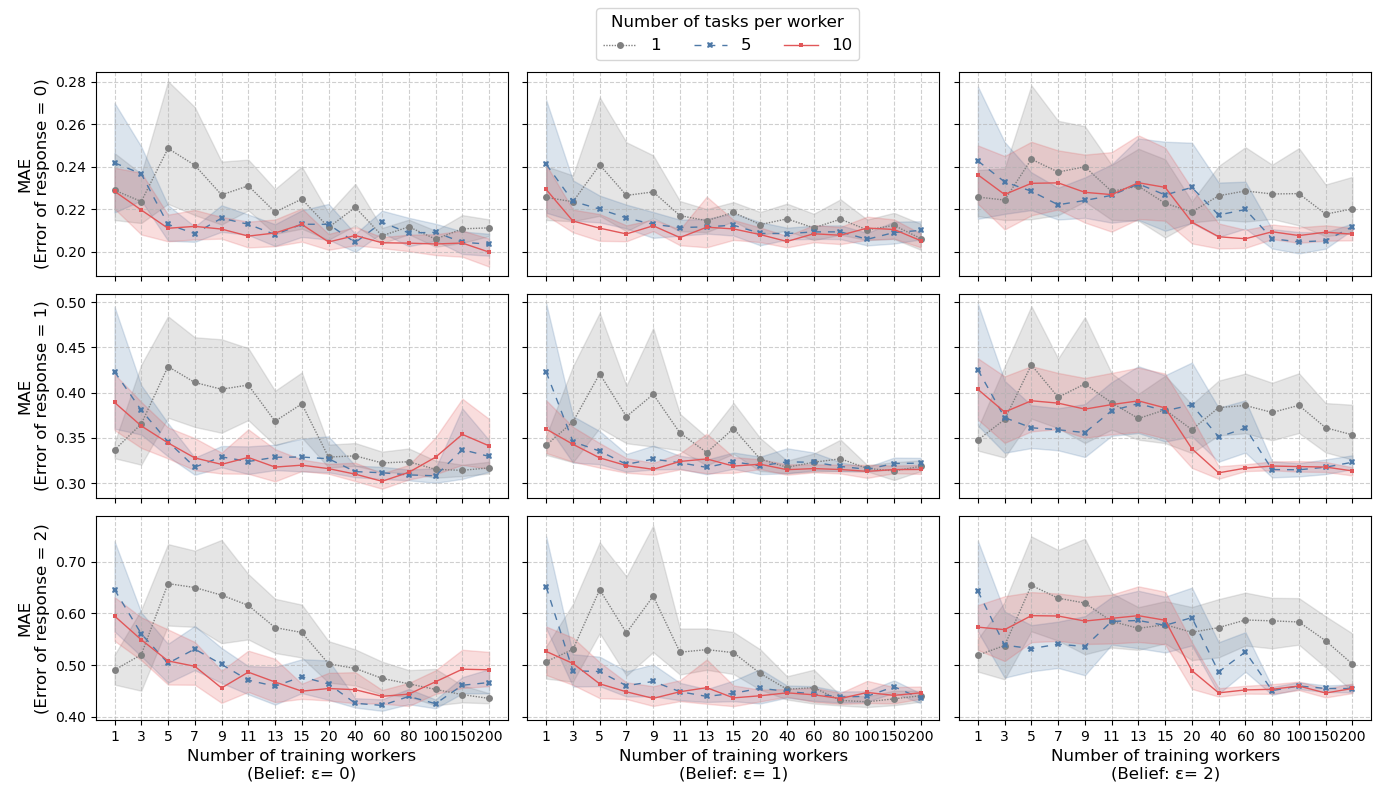}
    \caption{Simulation studies across different signal-to-noise levels of the simulated datasets}
    \label{fig:simulationall} 
\end{figure*}

Without loss of generality, we use the CrowdLLM model trained on the Offensiveness dataset in Case Study I as a ground truth model. In other words, we use the digital population created by CrowdLLM, which is trained on the Offensiveness dataset as the ``real'' population in this simulation study. In this way, we have the ground truth of the simulated datasets and can accurately evaluate the performance of CrowdLLM and other models. Then we generate datasets by different combinations of four factors: the number of participating workers, the number of tasks assigned per worker, the accuracy of worker responses, and the diversity of workers’ beliefs. Each dataset is generated according to a different design, and the dataset is randomly divided into training and testing partitions, with 20\% of the questions held out for testing. CrowdLLM is trained on the training set, and is then evaluated on the held-out test set. To account for randomness, each experiment is repeated 10 times with independent initializations. During testing, 20 virtual workers are instantiated in CrowdLLM for each question, and their responses are averaged to form the final answer. To manipulate the signal-to-noise levels of the simulated datasets by the CrowdLLM model, we further add Gaussian noise to the generated data. We consider three degrees of noise (i.e., standard deviation = 0, 1, 2, as indicated by the error of response). These three settings correspond to the first, middle, and last rows of Figure~\ref{fig:simulationall}, respectively. Within each accuracy scenario, we evaluate how the four factors, i.e., the belief diversity, the number of participating workers, and the number of questions per worker, impact the performance of CrowdLLM. The factor, belief diversity, is modeled using Gaussian distributions with standard deviations of 0, 1, and 2. We can see in Figure~\ref{fig:simulationall} how the increasing of the number of workers versus the increasing of the number of questions per worker affects the performance of CrowdLLM (i.e., evaluated by MAE), and how different levels of belief diversity shift the trade-off between accuracy and cost, as the three columns in Figure~\ref{fig:simulationall} represent different degrees of belief randomness ($\epsilon = 0, 1, 2$). For each experimental setting, the horizontal axis represents the number of workers used in the training data of CrowdLLM, while the color shade of the lines indicates the number of tasks per worker. 


\subsubsection{Key Insights from the Simulation Studies.}

In summary, we found five key insights from Figure~\ref{fig:simulationall} that clarify the roles of belief diversity, worker numbers, and workload in determining CrowdLLM's performance under varying signal-to-noise conditions of the simulated datasets. First, when the simulated datasets contain considerable noise (error of response =1 or 2), if we set belief diversity to 0, adding more workers in the training data does not improve CrowdLLM's performance. Instead, the uniformity of beliefs causes errors to reinforce one another, and the MAE increases as the number of workers increases. This is shown in the second and third rows of the first column in Figure~\ref{fig:simulationall}. Second, if datasets is noise-free (error of response =0), we should set belief diversity to 0 such that CrowdLLM achieves the best performance. However, in practice, this is an impossible scenario where everyone is perfectly aligned on the same response. This is shown in the first figure in Figure~\ref{fig:simulationall}. Third, when belief diversity is high (i.e, set to 2), CrowdLLM demands a large annotation budget to reach good performance. This is shown in the last column of Figure~\ref{fig:simulationall}, i.e., we need roughly 800 annotations (i.e., the product of the 20 virtual workers each answering 40 questions during testing) before MAE begins to stabilize. This pattern holds true across all levels of response quality, showing that high diversity systematically raises the annotation cost required to achieve reliable performance, as fewer workers require each to answer more questions to reach stability. Fourth, in the moderate diversity setting (belief diversity = 1), CrowdLLM achieves a more efficient balance between the number of workers and the number of questions answered per worker. As shown in the middle column of Figure~\ref{fig:simulationall}, roughly 10 workers (each answers 5 or 10 questions) can achieve strong performance, which is more efficient than the high-diversity setting that requires about 800 total annotations. Fifth, across all scenarios, the best performance that CrowdLLM can achieve is ultimately limited by the overall quality of the dataset. Datasets with more accurate worker responses provide better guidance for the model, leading to better performance of CrowdLLM. In summary, the quality of the dataset sets performance ceiling for CrowdLLM. To approach this ceiling efficiently, other factors must be carefully balanced. In particular, moderate belief diversity achieves a favorable trade-off between the number of workers and the number of questions each worker answers in the training data, yielding strong performance with relatively few annotations.

\section{Conclusion}\label{conclusion}
This paper introduces CrowdLLM, a novel framework that leverages lightweight generative models with LLM-based virtual crowdworkers to emulate the decision-making diversity and distributional fidelity typically observed in a range of crowd-based decision-making applications such as crowdsourcing, voting, and product review. Empirical evaluations across real-world and simulated datasets demonstrate that CrowdLLM achieves promising performance in both accuracy and distributional fidelity to human judgments. CrowdLLM outperforms strong baselines and remains robust under data scarcity. In future work, we plan to further refine CrowdLLM with more specialized mechanisms for particular application contexts, e.g., by integrating with human behavior models and choice models, or characterize the data-generating process in finer details. One possibility is to reconstruct responses from a human-centered perspective, reinforcing within-group coherence and enhancing between-group differentiation, so that simulated human participants retain population-level statistical tendencies while exhibiting individualized, human-like variability. This approach will better capture intra-group consistency and inter-group diversity, which may further improve the realism and generalizability of CrowdLLM.


\begin{APPENDIX}{Proofs of Theoretical Results}
\section{Proof of Theorem 1}
\begin{proof}
First of all, following the Theorem 1 in \cite{dahal2022deep}, we see that for any $\epsilon\in(0,1)$ and continuous distribution $\tilde{\mathcal{T}}$, there exists a generative model $G$ such that 
\begin{align*}
    W_1(G_{\sharp}\rho, \tilde{\mathcal{T}})<\epsilon,
\end{align*}
where $W_1$ is the 1-Wasserstein distance. It indicates that we can always find a generative model that approximates any target continuous distribution. Next, we extend the result to arbitrary distribution with mixed-type data. Denote two sets $S_1$ and $S_2$. For a vector $\boldsymbol{x}=(x_1,\cdots, x_d)$, let's denote that for $\forall i\in S_1$, $x_i$ is continuous, while for $\forall i\in S_2$, $x_i$ is discrete. Here $S_1\bigcup S_2=\{1,\cdots, d\}$ and $S_1\bigcap S_2=\varnothing$. Now we show how to replace all the discrete variables by continuous variables and generate new distributions. Suppose for any $j\in S_2$, $x_j$ follows a $K$-valued discrete distribution $f(x)=\sum_{k=1}^{K}p_k\delta(x-v_{k})\triangleq P_j$ where $\delta$ is a Dirac measure, $v_k$ is the $k$-th value of $x_j$ and $\sum_{k=1}^Kp_k=1$. Consider $\tilde{\boldsymbol{x}}=(x_1,\cdots,\tilde{x}_j,\cdots, x_d)$ where $\tilde{x}_j \sim \sum_{k=1}^Kp_k\mathcal{N}(v_k,\varepsilon^2\eta^2)\triangleq Q_j$ where $\eta$ is a fixed constant. Thus, based on the additive property of Wasserstein distance, we can obtain
\begin{align}
W_1(P_j,Q_j)\leq \sum_{k=1}^Kp_kW_1(\delta(v_k),Q^{(k)}_j),
\label{wsp}
\end{align}
where $Q_j^{(k)}\triangleq \mathcal{N}(v_k,\varepsilon^2\eta^2)$. It is easy to  see that
\begin{align*}
W_1(\delta(v_k),Q^{(j)}_k)=\mathbb{E}_{x\sim Q^{(j)}_k}|x-v_k|=\mathbb{E}_{x\sim \mathcal{N}(0,\varepsilon^2\eta^2)}|x| =\sqrt{\frac{2\eta^2}{\pi}}\varepsilon.
\end{align*}
Thus, with Eq.\eqref{wsp}, we have 
\begin{align*}
W_1(P_j,Q)\leq \sum_{k=1}^Kp_k \sqrt{\frac{2\eta^2}{\pi}}\varepsilon= \sqrt{\frac{2\eta^2}{\pi}} \varepsilon.
\end{align*}
where $j\in S_2$. It further gives $W_1(F(\boldsymbol{x}), F(\tilde{\boldsymbol{x}}))\leq \sqrt{\frac{2\eta^2}{\pi}}\varepsilon$. By running in all $j\in S_2$, we can build a sequence $\boldsymbol{z}_0,\boldsymbol{z}_1,\cdots, \boldsymbol{z}_{|S_2|}$, where $\boldsymbol{z}_0=\boldsymbol{x}$, $\boldsymbol{z}=\boldsymbol{z}_{|S_2|}$ and $\boldsymbol{z}_s$ replaces $x_{j_s}(j_s\in S_2)$ in $\boldsymbol{z}_{s-1}$, to gradually transform $\boldsymbol{x}$ to a fully continuous vector $\boldsymbol{z}$. As a result, we have
\begin{align*}
W(F(\boldsymbol{z}),F(\boldsymbol{x}))\leq \sum_{s=1}^{|S_2|}W(F(\boldsymbol{z}_{s}),F(\boldsymbol{z}_{s-1}))\leq |S_2|\sqrt{\frac{2\eta^2}{\pi}}\varepsilon.
\end{align*}
Let $\tilde{\mathcal{T}}=F(\boldsymbol{z})$ and denote the mixed distribution $F(\boldsymbol{x})$ by $\mathcal{T}$. Then, by the triangle inequality, it is easy to see that for the generative model $G'$, we have
\begin{align*}
 W_1(G_{\sharp}\rho, \mathcal{T})\leq W_1(G_{\sharp}\rho, \tilde{\mathcal{T}})+W_1(\tilde{\mathcal{T}},\mathcal{T})< (1+|S_2|\sqrt{\frac{2\eta^2}{\pi}})\varepsilon \leq  (1+ \sqrt{\frac{2}{\pi}}d\eta)\varepsilon.
\end{align*}
 
\end{proof}

\section{Proof of Theorem 2}
With a little abuse of the notation, we use $y(\boldsymbol{x}, \boldsymbol{z})$ to denote the response to task $\boldsymbol{x}$ given by the virtual participant with profile $\boldsymbol{z}$. When it won't cause confusion, we further simplify the notation and write it as $y(\boldsymbol{z})$. Before we go ahead to prove Theorem 2, we introduce the following lemma on ambiguity decomposition \cite{krogh1994neural, wood2023unified}:
\begin{lemma}
Given the ``ground truth" $\overline{y}$, and the noisy responses of a population $\tilde{y}_i (i=1,\cdots, N)$ with their average $\overline{\tilde{y}}=\frac{1}{N}\sum_{i=1}^N\tilde{y}_i$, we have the following ambiguity decomposition:
\begin{align*}
 \ell(\overline{y}, \overline{\tilde{y}})=\frac{1}{N}\sum_{i=1}^N\ell(\overline{y},\tilde{y}_i)- \frac{1}{N}\sum_{i=1}^N\ell(\overline{\tilde{y}},\tilde{y}_i).
\end{align*}
\end{lemma}
\begin{proof}
We first notice that the squared loss admits the following bias-variance decomposition: 
\begin{align*}
\mathbb{E}_{\boldsymbol{z}\sim\mathcal{T}}[\ell(\overline{y}, \tilde{y}(\boldsymbol{z}))]&=\mathbb{E}_{\boldsymbol{z}\sim\mathcal{T}}\Bigl[\ell\bigl(\overline{y}, \mathbb{E}_{\boldsymbol{z}\sim\mathcal{T}}[\tilde{y}(\boldsymbol{z})]\bigr)+\ell\bigl(\mathbb{E}_{\boldsymbol{z}\sim\mathcal{T}}[\tilde{y}(\boldsymbol{z})], \tilde{y}(\boldsymbol{z})\bigr)+2\bigl(\overline{y}-\mathbb{E}_{\boldsymbol{z}\sim\mathcal{T}}[\tilde{y}(\boldsymbol{z})]\bigr)\bigl(\mathbb{E}_{\boldsymbol{z}\sim\mathcal{T}}[\tilde{y}(\boldsymbol{z})]-\tilde{y}(\boldsymbol{z})\bigr)\Bigr]\\
&=\ell\bigl(\overline{y}, \mathbb{E}_{\boldsymbol{z}\sim\mathcal{T}}[\tilde{y}(\boldsymbol{z})]\bigr)+\mathbb{E}_{\mathcal{T}}\Bigl[\ell\bigl(\mathbb{E}_{\boldsymbol{z}\sim\mathcal{T}}[\tilde{y}(\boldsymbol{z})],\tilde{y}(\boldsymbol{z})\bigr)\Bigr].  
\end{align*}
By adopting $\mathcal{T}$ as a finite sample population $U=\{u_i\}_{i=1}^N$, we can rewrite the decomposition as
\begin{align*}
\frac{1}{N}\sum_{i=1}^N\ell(\overline{y}, \tilde{y}_i)=\ell\bigl(\overline{y}, \overline{\tilde{y}}\bigr)+\frac{1}{N}\sum_{i=1}^N\ell\bigl(\overline{\tilde{y}},\tilde{y}_i\bigr).
\end{align*}
A simple rearrangement of the terms completes the proof.
\end{proof}
With this ambiguity decomposition, by taking the expected risk of $\overline{\tilde{y}}$ we have
\begin{align*}
L&=\mathbb{E}_{\mathcal{T}}\Bigl[\mathbb{E}_{\mathcal{D}}\bigl[\mathbb{E}_{\boldsymbol{x}\sim\mathcal{X},y\sim\mathcal{Y}}[\frac{1}{N}\sum_{i=1}^N\ell(\overline{y},\tilde{y}_i)]\bigr]\Bigr]-\mathbb{E}_{\mathcal{T}}\Bigl[\mathbb{E}_{\mathcal{D}}\bigl[\mathbb{E}_{\boldsymbol{x}\sim\mathcal{X},y\sim\mathcal{Y}}[\frac{1}{N}\sum_{i=1}^N\ell(\overline{\tilde{y}},\tilde{y}_i)]\bigr]\Bigr]\\
&=\mathbb{E}_{\mathcal{X}}\Bigl[\mathbb{E}_{\mathcal{T},\mathcal{D}}\bigl[\mathbb{E}_{y\sim\mathcal{Y}|\mathcal{X},\mathcal{T}}[\frac{1}{N}\sum_{i=1}^N\ell(\overline{y},\tilde{y}_i)]\bigr]\Bigr]-\mathbb{E}_{\mathcal{T}, \mathcal{D}}\bigl[\frac{1}{N}\sum_{i=1}^N\ell(\overline{\tilde{y}},\tilde{y}_i)\bigr]\Bigr].
\end{align*}
Now we notice that the first term has the following decomposition:
\begin{align*}
\mathbb{E}_{\mathcal{T},\mathcal{D}}\Bigl[\mathbb{E}_{y,\tilde{y}\sim\mathcal{Y}|\mathcal{X},\mathcal{T}}[\frac{1}{N}\sum_{i=1}^N\ell(\overline{y},\tilde{y}_i)]\Bigr]&=\mathbb{E}_{\mathcal{T},\mathcal{D}}\Bigl[\mathbb{E}_{y,\tilde{y}\sim\mathcal{Y}|\mathcal{X},\mathcal{T}}[\frac{1}{N}\sum_{i=1}^N\ell(\overline{y},\overline{\tilde{y}}_i)]\Bigr]+\mathbb{E}_{\mathcal{T},\mathcal{D}}\Bigl[\mathbb{E}_{\tilde{y}\sim\mathcal{Y}|\mathcal{X},\mathcal{T}}[\frac{1}{N}\sum_{i=1}^N\ell(\overline{\tilde{y}}_i,\tilde{y}_i)]\Bigr]\\
&=\mathbb{E}_{\mathcal{T},\mathcal{D}}\Bigl[\mathbb{E}_{y,\tilde{y}\sim\mathcal{Y}|\mathcal{X},\mathcal{T}}[\frac{1}{N}\sum_{i=1}^N\ell(\overline{y},\overline{\tilde{y}}_i)]\Bigr]+\mathbb{E}_{\mathcal{T}}[\frac{1}{N}\sum_{i=1}^N\tilde{\eta}_i^2]\\
&=\mathbb{E}_{\mathcal{T},\mathcal{D}}\Bigl[\mathbb{E}_{y\sim\mathcal{Y}|\mathcal{X},\mathcal{T}}[\frac{1}{N}\sum_{i=1}^N\ell(\overline{y},\overline{y}_i)]\Bigr]+\mathbb{E}_{\mathcal{T},\mathcal{D}}\Bigl[\frac{1}{N}\sum_{i=1}^N\ell(\overline{y}_i,\overline{\tilde{y}}_i)\Bigr]+\mathbb{E}_{\mathcal{T}}[\frac{1}{N}\sum_{i=1}^N\tilde{\eta}_i^2].
\end{align*}
This is followed by
\begin{align*}
\mathbb{E}_{\mathcal{T},\mathcal{D}}\Bigl[\mathbb{E}_{y\sim\mathcal{Y}|\mathcal{X},\mathcal{T}}[\frac{1}{N}\sum_{i=1}^N\ell(\overline{y},\overline{y}_i)]\Bigr]&=\mathbb{E}_{\mathcal{T},\mathcal{D}}\Bigl[\mathbb{E}_{y\sim\mathcal{Y}|\mathcal{X},\mathcal{T}}[\frac{1}{N}\sum_{i=1}^N\ell(\overline{y},y_i)]\Bigr]+\mathbb{E}_{\mathcal{T},\mathcal{D}}\Bigl[\mathbb{E}_{y\sim\mathcal{Y}|\mathcal{X},\mathcal{T}}[\frac{1}{N}\sum_{i=1}^N\ell(y_i,\overline{y}_i)]\Bigr]\\
&=\mathbb{E}_{\boldsymbol{v}_1,\cdots,\boldsymbol{v}_N\sim\mathcal{T}}\Bigl[\frac{1}{N}\sum_{i=1}^N\mathbb{E}_{y_i\sim\mathcal{Y}|\mathcal{X},\boldsymbol{v}_i}[\ell(\overline{y},y_i)]\Bigr]+\mathbb{E}_{\mathcal{T}}\Bigl[\frac{1}{N}\sum_{i=1}^N\eta_i^2\Bigr].
\end{align*}
Putting all together, we can obtain the full decomposition:
\begin{align*}
 L=&\mathbb{E}_{\mathcal{X}}\Bigl[\mathbb{E}_{\mathcal{T}}\bigl[\frac{1}{N}\sum_{i=1}^N\eta_i^2\bigr]+\mathbb{E}_{\boldsymbol{v}_1,\cdots,\boldsymbol{v}_N\sim\mathcal{T}}\bigl[\frac{1}{N}\sum_{i=1}^N\mathbb{E}_{y_i\sim\mathcal{Y}|\mathcal{X},\boldsymbol{v}_i}[\ell(\overline{y},y_i)]\bigr]\\
&+\mathbb{E}_{\mathcal{T},\mathcal{D}}\bigl[\frac{1}{N}\sum_{i=1}^N\ell(\overline{y}_i,\overline{\tilde{y}}_i)\bigr]+\mathbb{E}_{\mathcal{T}}\bigl[\frac{1}{N}\sum_{i=1}^N\tilde{\eta}_i^2\bigr]-{E}_{\mathcal{T},\mathcal{D}}\bigl[\frac{1}{N}\sum_{i=1}^N\ell(\overline{\tilde{y}},\tilde{y}_i)\bigr]\Bigr].
\end{align*}

\section{Proof of Theorem 3}
By theorem 2, we have 
\begin{align*}
L=L_1+L_2+\mathbb{E}_{\mathcal{T},\mathcal{D}}\bigl[\frac{1}{N}\sum_{i=1}^N\ell(\overline{y}_i,\overline{\tilde{y}}_i)\bigr]+\mathbb{E}_{\mathcal{T}}\bigl[\frac{1}{N}\sum_{i=1}^N\tilde{\eta}_i^2\bigr]-{E}_{\mathcal{T},\mathcal{D}}\bigl[\frac{1}{N}\sum_{i=1}^N\ell(\overline{\tilde{y}},\tilde{y}_i)\bigr]\Bigr].
\end{align*}
while Proposition 3 gives us
\begin{align*}
L'=L_1+L_2+\mathbb{E}_{\mathcal{T}}\bigl[\frac{1}{N}\sum_{i=1}^N\ell(\overline{y}_i,y_{ref})\bigr]+\eta(t).
\end{align*}
Therefore, we have
\begin{align*}
L-L'=\mathbb{E}\bigl[\frac{1}{N}\sum_{i=1}^N\ell(\overline{y}_i,\overline{\tilde{y}}_i)\bigr]+\mathbb{E}\bigl[\frac{1}{N}\sum_{i=1}^N\tilde{\eta}_i^2\bigr]-{E}\bigl[\frac{1}{N}\sum_{i=1}^N\ell(\overline{\tilde{y}},\tilde{y}_i)\bigr]\Bigr]-\mathbb{E}\bigl[\frac{1}{N}\sum_{i=1}^N\ell(\overline{y}_i,y_{ref})\bigr]-\eta(t).
\end{align*}
By Eq. (8), we have $\tilde{y}_i=y_{ref}+\delta_i+\tilde{\varepsilon}_i$ where, for a specific $\boldsymbol{x}$, $y_{ref}=\Phi(\boldsymbol{x})$, $\delta_i=\delta(\boldsymbol{x},\boldsymbol{v}_i)$, $\mathbb{E}[\tilde{\varepsilon}_i]=0$ and $\mathbb{E}[\tilde{\varepsilon}_i^2]=\tilde{\eta}^2$. It further gives $\overline{\tilde{y}}_i =y_{ref}+\frac{1}{N}\sum_{i=1}^N\delta_i$. Thus, we can obtain
\begin{align*}
 L-L'=& \mathbb{E}\Bigl[\frac{1}{N}\sum_{i=1}^N\ell(\overline{y}_i,y_{ref}+\frac{1}{N}\sum_{i=1}^N\delta_i)\Bigr]+\mathbb{E}\bigl[\frac{1}{N}\sum_{i=1}^N\tilde{\eta}_i^2\bigr]-{E}\bigl[\frac{1}{N}\sum_{i=1}^N\ell(y_{ref}+\delta_i+\tilde{\varepsilon}_i,y_{ref}+\frac{1}{N}\sum_{i=1}^N\delta_i)\bigr]\Bigr]\\
 &-\mathbb{E}\bigl[\frac{1}{N}\sum_{i=1}^N\ell(\overline{y}_i,y_{ref})\bigr]-\eta(t)\\
 =&\mathbb{E}\Bigl[(\frac{1}{N}\sum_{i=1}^N\delta_i)^2\Bigr]-2\mathbb{E}\Bigl[(\frac{1}{N}\sum_{i=1}^N\delta_i)(\frac{1}{N}\sum_{i=1}^N\overline{y}_i-y_{ref})\Bigr]-{E}\bigl[\frac{1}{N}\sum_{i=1}^N(\delta_i-\frac{1}{N}\sum_{i=1}^N\delta_i)^2\bigr]\Bigr]-\eta(t)\\
 =&\mathbb{E}\Bigl[(\frac{1}{N}\sum_{i=1}^N\delta_i)^2\Bigr]-2\mathbb{E}\Bigl[(\frac{1}{N}\sum_{i=1}^N\delta_i)(y^{**}-y_{ref})\Bigr]-{E}\bigl[\frac{1}{N}\sum_{i=1}^N(\delta_i-\frac{1}{N}\sum_{i=1}^N\delta_i)^2\bigr]\Bigr]-\eta(t)\\
 =&2\mathbb{E}\Bigl[(\frac{1}{N}\sum_{i=1}^N\delta_i)^2\Bigr]-2\mathbb{E}\Bigl[(\frac{1}{N}\sum_{i=1}^N\delta_i)(y^{**}-y_{ref})\Bigr]-{E}\Bigl[\frac{1}{N}\sum_{i=1}^N\delta_i^2\Bigr]-\eta(t)\\
 =&\mathbb{E}[2P^2-2P(y^{**}-y_{ref})-Q-\eta(t)]
\end{align*}
where $y^{**}$ is the gold-standard response of the real human population that is not affected by the digital population, $P=\frac{1}{N}\sum_{i=1}^N\delta_i$ ,and $Q=\frac{1}{N}\sum_{i=1}^N\delta_i^2$. Now let $\Delta=y^{**}-y_{ref}$. Noting that $\mathbb{E}[P^2]=\mathbb{E}^2[P]+Var[P]=\mu_{\delta}^2+\frac{1}{N}\varepsilon_{\delta}^2-\frac{1}{N^2}\sum_{i=1}^N\mathbb{E}^2[\delta_i]$, $\mathbb{E}[P]=\mu_{\delta}$, and $\mathbb{E}[Q]=\varepsilon^2_{\delta}$, it yields 
\begin{align*}
 L-L'&=2\mu_{\delta}^2+\frac{2}{N}\varepsilon_{\delta}^2-\frac{2}{N^2}\sum_{i=1}^N\mathbb{E}^2[\delta_i]-2\mu_{\delta}(y^{**}-y_{ref})-\varepsilon^2_\delta-\eta(t)\\
 &\leq 2(1-\frac{1}{N})\mu_{\delta}^2-2\mu_{\delta}\Delta-(1-\frac{2}{N})\varepsilon^2_\delta-\eta(t)\\
 &=\mathbb{E}\Bigl[\frac{2(N-1)}{N}[\mu_{\delta}-\frac{N}{2(N-1)}\Delta]^2-\frac{N}{2(N-1)}\Delta^2-(1-\frac{2}{N})\varepsilon^2_\delta-\eta(t)\Bigr],
\end{align*}
For brevity, we let $A=\frac{N}{2(N-1)}, B=1-\frac{2}{N}, C=\frac{N-2}{2(N-1)}=1-A$. Obviously, when $N\geq 2$, $0<A\leq 1$ and $B,C\geq0$. The equation above gives a sufficient condition of $L-L'\leq 0$:
\begin{align*}
A\Delta-\sqrt{A^2\Delta^2+AB\varepsilon^2_\delta+A\eta(t)}\leq \mu_{\delta}\leq A\Delta+\sqrt{A^2\Delta^2+AB\varepsilon^2_\delta+A\eta(t)}
\end{align*}
It is equivalent to 
\begin{align*}
\Delta-\sqrt{A^2\Delta^2+AB\varepsilon^2_\delta+A\eta(t)}-C\Delta\leq \mu_{\delta}\leq \Delta+\sqrt{A^2\Delta^2+AB\varepsilon^2_\delta+A\eta(t)}-C\Delta  
\end{align*}
Let $\Delta_L=\sqrt{A^2\Delta^2+AB\varepsilon^2_\delta+A\eta(t)}+C\Delta$ and $\Delta_U=\sqrt{A^2\Delta^2+AB\varepsilon^2_\delta+A\eta(t)}-C\Delta$. Meanwhile, we notice the derivatives are:
\begin{align*}
&\Delta_L'=\frac{A^2\Delta}{\sqrt{A^2\Delta^2+AB\varepsilon^2_\delta+A\eta(t)}}+C\\
&\Delta_U'=\frac{A^2\Delta}{\sqrt{A^2\Delta^2+AB\varepsilon^2_\delta+A\eta(t)}}-C
\end{align*}
Let us consider $$\Delta_0=\sqrt{\frac{C^2[B\varepsilon_{\delta}^2+\eta(t)]}{A(A^2-C^2)}}=\frac{N-2}{N}\sqrt{\frac{(N-2)\varepsilon_{\delta}^2+N\eta(t)}{2}}$$ The stationary points lie at  $\Delta=-\Delta_0$ and $\Delta=\Delta_0$, respectively. Here, it is easy to see $A>C\geq 0$ for $N\geq 2$, thus $\Delta_0=\frac{C^2[B\varepsilon_{\delta}^2+\eta(t)]}{A(A^2-C^2)}\geq 0$. Therefore, we have the following observations: When $\Delta>0$, as $\Delta$ increases, $\Delta_L$ will increase, whereas $\Delta_U$ will decrease when $\Delta<\Delta_0$ and become increasing for $\Delta\geq \Delta_0$. When $\Delta<0$, as $\Delta$ increases, $\Delta_L$ will decrease until $\Delta$ reaches $-\Delta_0$ and become increasing,  whereas $\Delta_U$ will decrease. Thus, we can find the minimum of the bounds:
\begin{align*}
\Delta_L, \Delta_U\geq \frac{1}{N}\sqrt{2[(N-2)\varepsilon_{\delta}^2+N\eta(t)]}
\end{align*}
As a result, we have two cases of the bounds depending on $N$ and can build confidence interval as follows:
\begin{itemize}
    \item \textbf{$N$ is sufficiently large}. When $\Delta_0=\frac{N-2}{N}\sqrt{\frac{(N-2)\varepsilon_{\delta}^2+N\eta(t)}{2}}\geq \kappa_{\alpha}$, $\Delta_L$ and $\Delta_U$ will be monotonically increasing and decreasing, within $|\Delta|\leq \kappa_{\alpha}$. In this case, we let 
    \begin{align*}
        h_{\Delta}(\kappa_{\alpha})=\frac{\sqrt{N^2\kappa_{\alpha}^2+2(N-1)(N-2)\varepsilon_{\delta}^2+2N(N-1)\eta(t)}-(N-2)\kappa_{\alpha}}{2(N-1)}
    \end{align*}
    \item \textbf{$N$ is small}. Given $\Delta_0=\frac{N-2}{N}\sqrt{\frac{(N-2)\varepsilon_{\delta}^2+N\eta(t)}{2}}< \kappa_{\alpha}$, we let
    \begin{align*}
        h_{\Delta}(\kappa_{\alpha})= \frac{1}{N}\sqrt{2[(N-2)\varepsilon_{\delta}^2+N\eta(t)]}
    \end{align*}
\end{itemize}
We can build the interval as 
\begin{align*}
B_\alpha(\Delta)=[\Delta-h_\Delta(C), \Delta+h_\Delta(C)]    
\end{align*}
Since $|\Delta|<\kappa_{\alpha}$ with at least probability $1-\alpha$, if $\mu_{\delta}\in B_\alpha(\Delta)$, with the same probability, we can guarantee $L\leq L'$. 

\section{Proof of Theorem 4}
This proof follows \cite{angelopoulos2023prediction}. Let $\mathcal{E}_1=\{\mathbb{E}_{\boldsymbol{v}\sim\mathcal{T}}[f(\boldsymbol{x},\boldsymbol{v})-y]\in B^1_{\gamma}(\theta^*)\}$ and $\mathcal{E}_2=\{\mathbb{E}_{\boldsymbol{v}\sim\mathcal{T}}[\theta^*-f(\boldsymbol{x},\boldsymbol{v})]\in  B^2_{\alpha-\gamma}(\theta^*)\}$. From the conditions, we have $P(\mathcal{E}_1)\geq 1-\gamma$ and $P(\mathcal{E}_2)\geq 1-(\alpha-\gamma)$. Consider the event $\mathcal{E}=\mathcal{E}_1\bigcap\mathcal{E}_2$. It is easy to see $P(\mathcal{E})=1-P(\mathcal{E}_1^c\bigcup \mathcal{E}_2^c)\geq 1-P(\mathcal{E}_1^c)-P(\mathcal{E}_2^c)=P(\mathcal{E}_1)+P(\mathcal{E}_2)-1\geq 1-[\gamma+(\alpha-\gamma)]=1-\alpha$. On the event $\mathcal{E}$, we have
\begin{align*}
\mathbb{E}[\theta^*-y|\boldsymbol{x}]&=\mathbb{E}[\theta^*-y|\boldsymbol{x}]-\mathbb{E}[\theta^*-f(\boldsymbol{x},\boldsymbol{v})|\boldsymbol{x}] + \mathbb{E}[\theta^*-f(\boldsymbol{x},\boldsymbol{v})|\boldsymbol{x}]  \\
&=\mathbb{E}[f(\boldsymbol{x},\boldsymbol{v})-y|\boldsymbol{x}]+\mathbb{E}[\theta^*-f(\boldsymbol{x},\boldsymbol{v})|\boldsymbol{x}]\\
&\in B^1_{\gamma}(\theta^*)+B^2_{\alpha-\gamma}(\theta^*).
\end{align*}
Noticing $\theta^*=\mathbb{E}[y|\boldsymbol{x}]$, we have $\mathbb{E}[\theta^*-y|\boldsymbol{x}]=0$. Thus we have $0\in B^1_{\gamma}(\theta^*)+B^2_{\alpha-\gamma}(\theta^*)$, which turns to be a necessary condition. This completes the proof.

\section{Proof of Theorem 5}
We borrow the techniques from \citep{angelopoulos2023prediction} and show that $y^*\notin B^{y}_{\alpha}$ with probability at most $\alpha$ when $n,N\to\infty$. First, we denote $\theta=y^*$ and notice that
\begin{align*}
    \theta-\overline{\tilde{y}}=\theta-\frac{1}{N}\sum_{i=1}^{N}\tilde{y}_i=\theta-\frac{1}{N}\sum_{i=1}^N[\Phi^{(i)}(\boldsymbol{x})+\delta_i]=\Bigl(\mathbb{E}[\Phi(\boldsymbol{x})]-\frac{1}{N}\sum_{i=1}^N\Phi^{(i)}(\boldsymbol{x})\Bigr)+\Bigl(\theta-\mathbb{E}[\Phi(\boldsymbol{x})]-\frac{1}{N}\sum_{i=1}^N\delta_i\Bigr).
\end{align*}
Let $r_i'=-r_i$, $\Delta_\delta=\theta-\mathbb{E}[\Phi(\boldsymbol{x})]-\overline{\delta}$ and $\Delta_\Phi=\mathbb{E}[\Phi(\boldsymbol{x})]-\frac{1}{N}\sum_{i=1}^N\Phi^{(i)}(\boldsymbol{x})$. Thus, we have $\theta-\overline{\tilde{y}}=\Delta_{\delta}+\Delta_{\Phi}$. By central limit theorem, we have
\begin{align*}
&\sqrt{n}(\overline{r}'-\mathbb{E}[\overline{r}'])\xrightarrow{d}\mathcal{N}(0,\sigma_r^2),\\
&\sqrt{N}(\Delta_{\delta}-\mathbb{E}[\Delta_{\delta}])\xrightarrow{d}\mathcal{N}(0,\sigma_\delta^2),\\
&\sqrt{N}(\Delta_{\Phi}-\mathbb{E}[\Delta_{\Phi}])\xrightarrow{d}\mathcal{N}(0,\eta_{\Phi}(t)).
\end{align*}
Thus, we have
\begin{align*}
\sqrt{N}(\Delta_{\delta}+\Delta_{\Phi}+\overline{r}'-\mathbb{E}[\Delta_{\delta}+\Delta_{\Phi}+\overline{r}'])&=\sqrt{n}\cdot\sqrt{\frac{N}{n}}(\overline{r}'-\mathbb{E}[\overline{r}'])+\sqrt{N}\Bigl\{(\Delta_{\delta}-\mathbb{E}[\Delta_{\delta}])+(\Delta_{\Phi}-\mathbb{E}[\Delta_{\delta}])\Bigr\}\\
&\to \mathcal{N}(0,\frac{1}{p}\sigma_r^2+\sigma_\delta^2+\eta(t)).
\end{align*}
Let $\hat{\sigma}^2=\frac{1}{p}\hat{\sigma}_r^2+\hat{\sigma}_\delta^2+\eta(t)=\frac{N}{n}\hat{\sigma}_r^2+\hat{\sigma}_\delta^2+\eta(t)$. It is easy to see that this is a consistent estimate of the variance $\frac{1}{p}\sigma_r^2+\sigma_\delta^2+\eta(t)$. Therefore, we have
\begin{align*}
\lim_{n,N\to\infty}P(|(\Delta_{\delta}+\Delta_{\Phi}+\overline{r}')-\mathbb{E}[\Delta_{\delta}+\Delta_{\Phi}+\overline{r}']|\geq z_{1-\frac{\alpha}{2}}\frac{\hat{\sigma}}{\sqrt{N}})\leq \alpha.
\end{align*}
Since we know
\begin{align*}
\Delta_{\delta}+\Delta_{\Phi}+\overline{r}'=\Delta_{\delta}+\Delta_{\Phi}-\overline{r}=\theta-\overline{\tilde{y}}+\overline{\tilde{y}}-\overline{y}=\theta-\overline{y},
\end{align*}
we can easily obtain
\begin{align*}
  \mathbb{E}[\Delta_{\delta}+\Delta_{\Phi}+\overline{r}']=\mathbb{E}[\theta-\overline{y}]=0  .
\end{align*}
Thus, we have
\begin{align*}
\lim_{n,N\to\infty}P(|\Delta_{\delta}+\Delta_{\Phi}+\overline{r}'|\geq z_{1-\frac{\alpha}{2}}\frac{\hat{\sigma}}{\sqrt{N}})\leq \alpha,   
\end{align*}
which is equivalent to 
\begin{align*}
\lim_{n,N\to\infty}P(|\theta-\overline{\tilde{y}}-\overline{r}|\geq z_{1-\frac{\alpha}{2}}\sqrt{\frac{\eta(t)}{N}+\frac{\sigma_\delta^2}{N}+\frac{\sigma_r^2}{n}})\leq \alpha .      
\end{align*}
This results in
\begin{align*}
\lim_{n,N\to\infty}P(|\theta-\overline{\tilde{y}}|\geq|\overline{r}|+z_{1-\frac{\alpha}{2}}\sqrt{\frac{\eta(t)}{N}+\frac{\sigma_\delta^2}{N}+\frac{\sigma_r^2}{n}})\leq \alpha .      
\end{align*}
Noticing that 
\begin{align*}
|\overline{r}|^2=|\frac{1}{n}\sum_{i=1}^ny_i-\frac{1}{n}\sum_{i=1}^n\tilde{y}_i|^2\leq \frac{1}{n}\sum_{i=1}^n(y_i-\tilde{y}_i)^2\leq \varepsilon_0^2,
\end{align*}
which gives $|\overline{r}|\leq \varepsilon_0$, we get
\begin{align*}
\lim_{n,N\to\infty}P(|\theta-\overline{\tilde{y}}|\geq\varepsilon_0+z_{1-\frac{\alpha}{2}}\sqrt{\frac{\eta(t)}{N}+\frac{\sigma_\delta^2}{N}+\frac{\sigma_r^2}{n}})\leq \alpha .      
\end{align*}

\end{APPENDIX}

\begin{APPENDIX}{Experiments}
\section{Examples of Prompts}
\subsection{Zero-shot prompt}
By default, for all the LLM-based decision-making methods, we use zero-shot prompts which provides LLM with the problem description $\boldsymbol{x}_t$, the specific requirements $\mathcal{R}_t$ on the decisions and the context $\mathcal{C}_t$. $\boldsymbol{x}_t$ describes the specific problem the workers need to make decisions on. $\mathcal{R}_t$ describes what kind of decisions need to be made, e.g., if the decisions should be a score, the scale of the score should be included in $\mathcal{R}_t$. And $\mathcal{C}_t$ encodes information about the decision-making scenario and the decision-maker LLM should simulate. An example of zero-shot prompt is given below in Figure \ref{fig:prompt_zero}.

\begin{figure}[!htbp]
    \centering
    \includegraphics[width=0.7\textwidth]{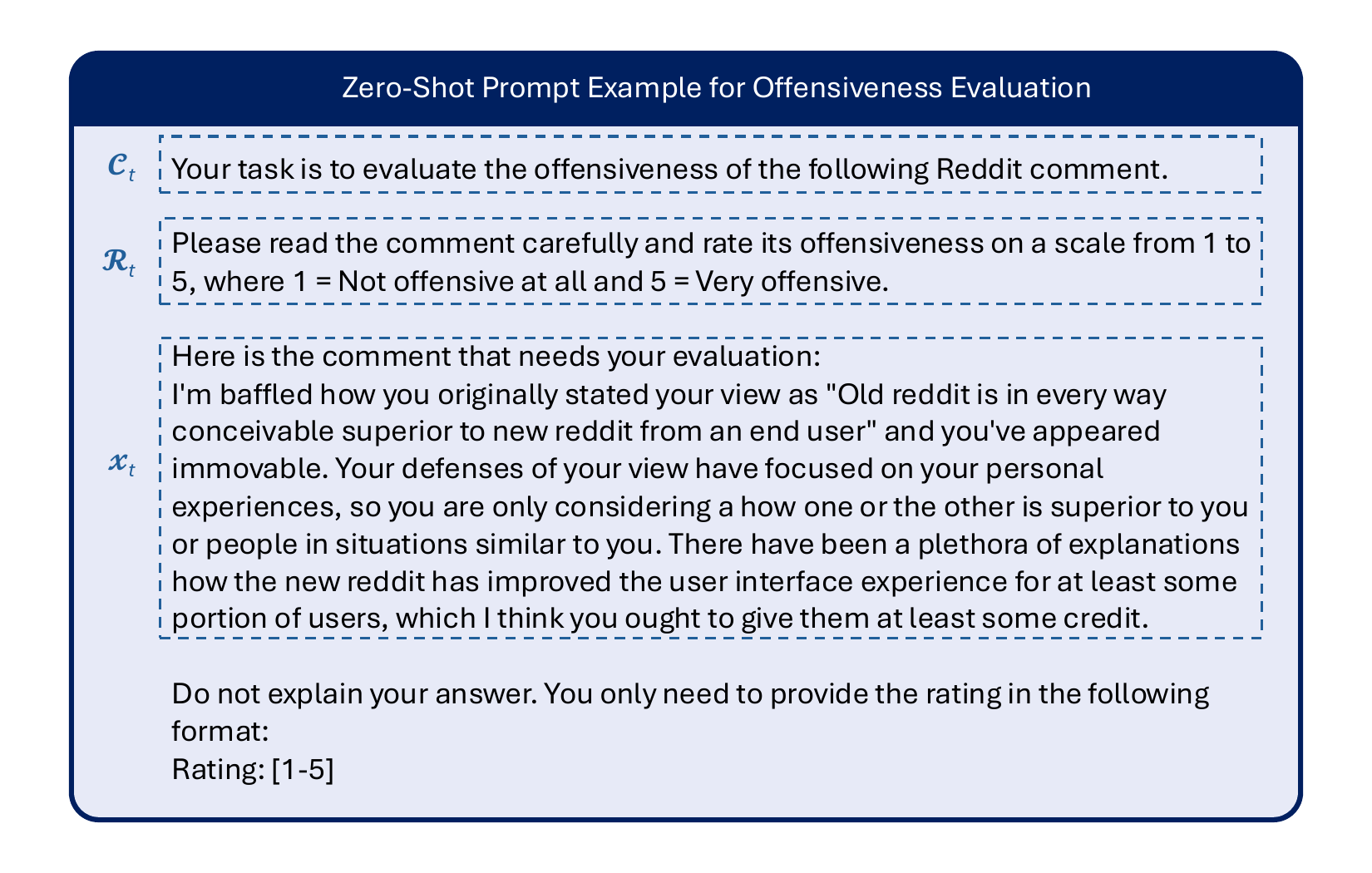}
    \caption{Example of zero-shot prompt for offensiveness evaluation.}
    \label{fig:prompt_zero}
\end{figure}

\subsection{Multi-Persona prompt}
\begin{figure}[!htbp]
    \centering
    \includegraphics[width=0.7\textwidth]{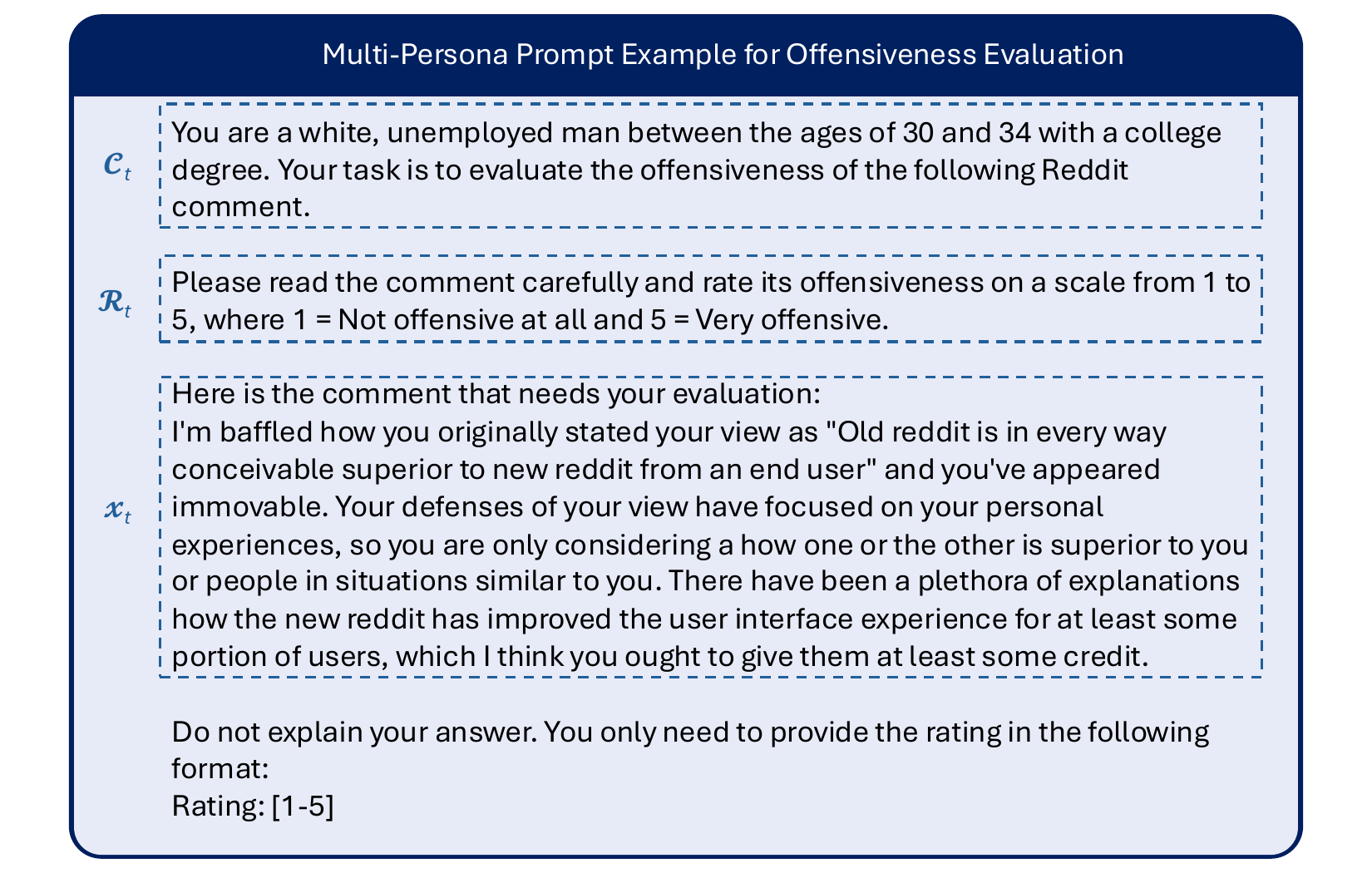}
    \caption{Example of multi-persona prompt for offensiveness evaluation.}
    \label{fig:prompt_persona}
\end{figure}

With the idea of using multiple personas in prompting for self-collaboration \cite{olea2024evaluating, hu2024quantifying}, we use LLM to simulate different individuals by building multiple personas with additional context on their profiles. We use similar prompts as zero-shot prompting for making decisions but only change the context $\mathcal{C}_t$ to assign the persona. An example of multi-persona prompt is provided below in Figure \ref{fig:prompt_persona}.

\subsection{Self-Consistency (SC) prompt}
Following \cite{wangself, liu2024dellma}, we perform self-consistency prompting by taking the majority vote of multiple decisions generated by LLM with temperature 0.5. For the generation of each decision, we use the same prompt as zero-shot prompting. 


\section{Additional Experimental Results}
\textbf{Section B.1} discusses the impact of different LLM backbones on the performance. \textbf{Section B.2-6} provides experimental results in addition to \textbf{Section 5} of the main text.

\subsection{The impact of the number of problems}

\begin{figure}[!t]
    \centering
    \begin{subfigure}{\dimexpr\textwidth/2\relax}
        \centering
        \includegraphics[width=\textwidth]{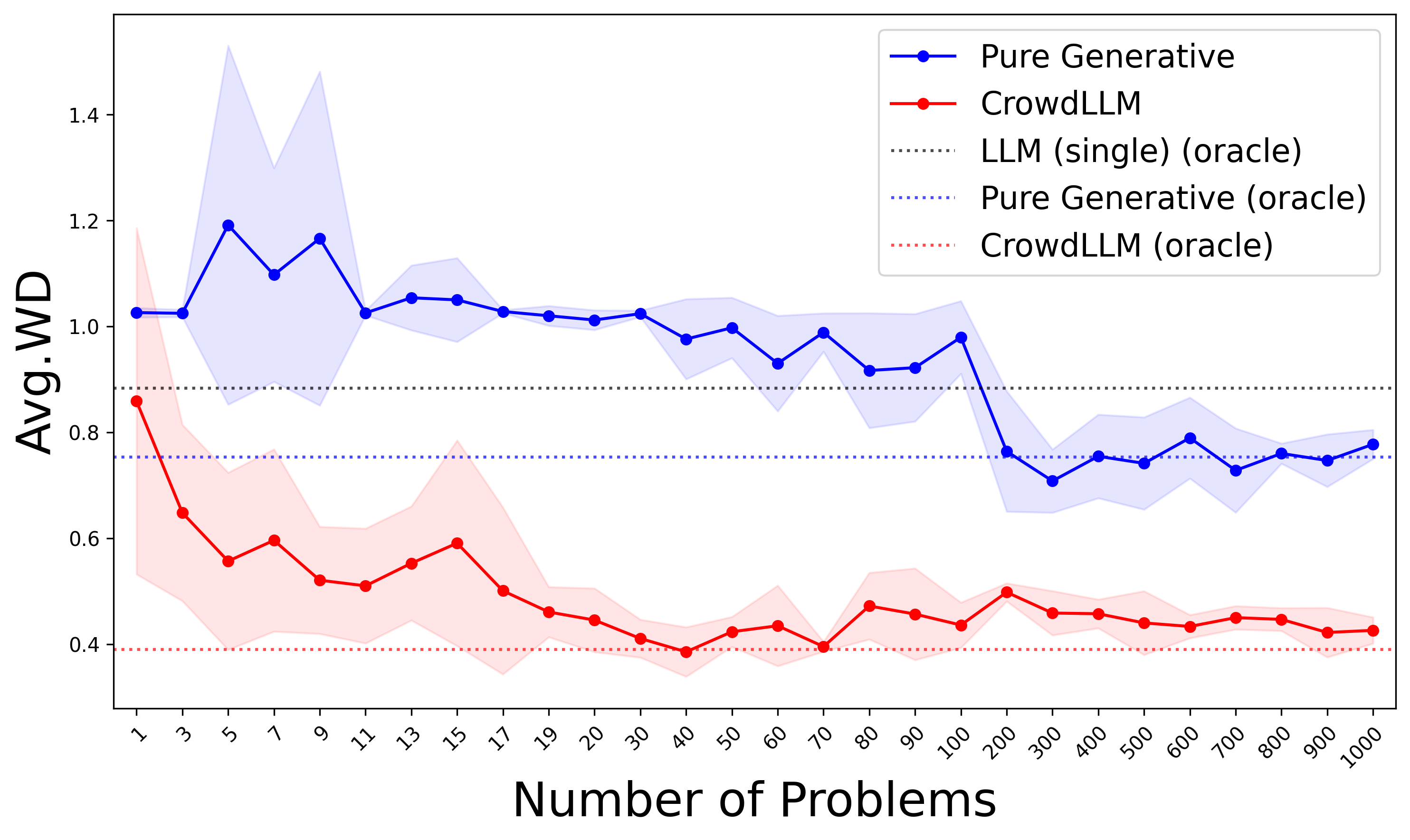}
        \caption{Offensiveness dataset}
        \label{fig:sub_offensiveness_wd_problem}
    \end{subfigure}%
    \begin{subfigure}{\dimexpr\textwidth/2\relax}
        \centering
        \includegraphics[width=\textwidth]{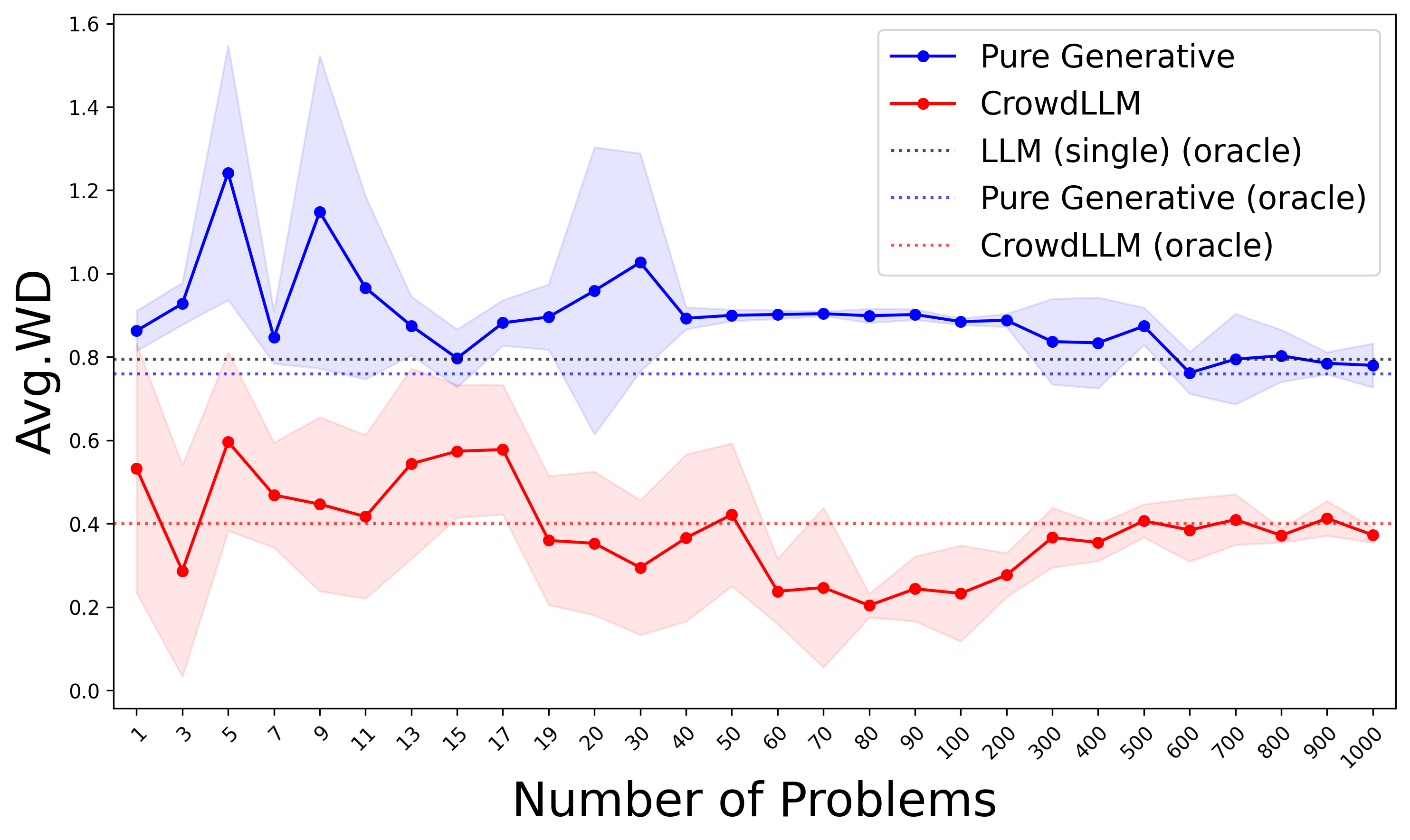}
        \caption{QA Difficulty dataset}
        \label{fig:sub_qa_difficulty_wd_problem}
    \end{subfigure}%
    \caption{Avg. WD vs. Number of Problems.}
    \label{fig:wd_vs_num_problems}
\end{figure}

\begin{figure}[!t]
    \centering
    \begin{subfigure}{\dimexpr\textwidth/2\relax}
        \centering
        \includegraphics[width=\textwidth]{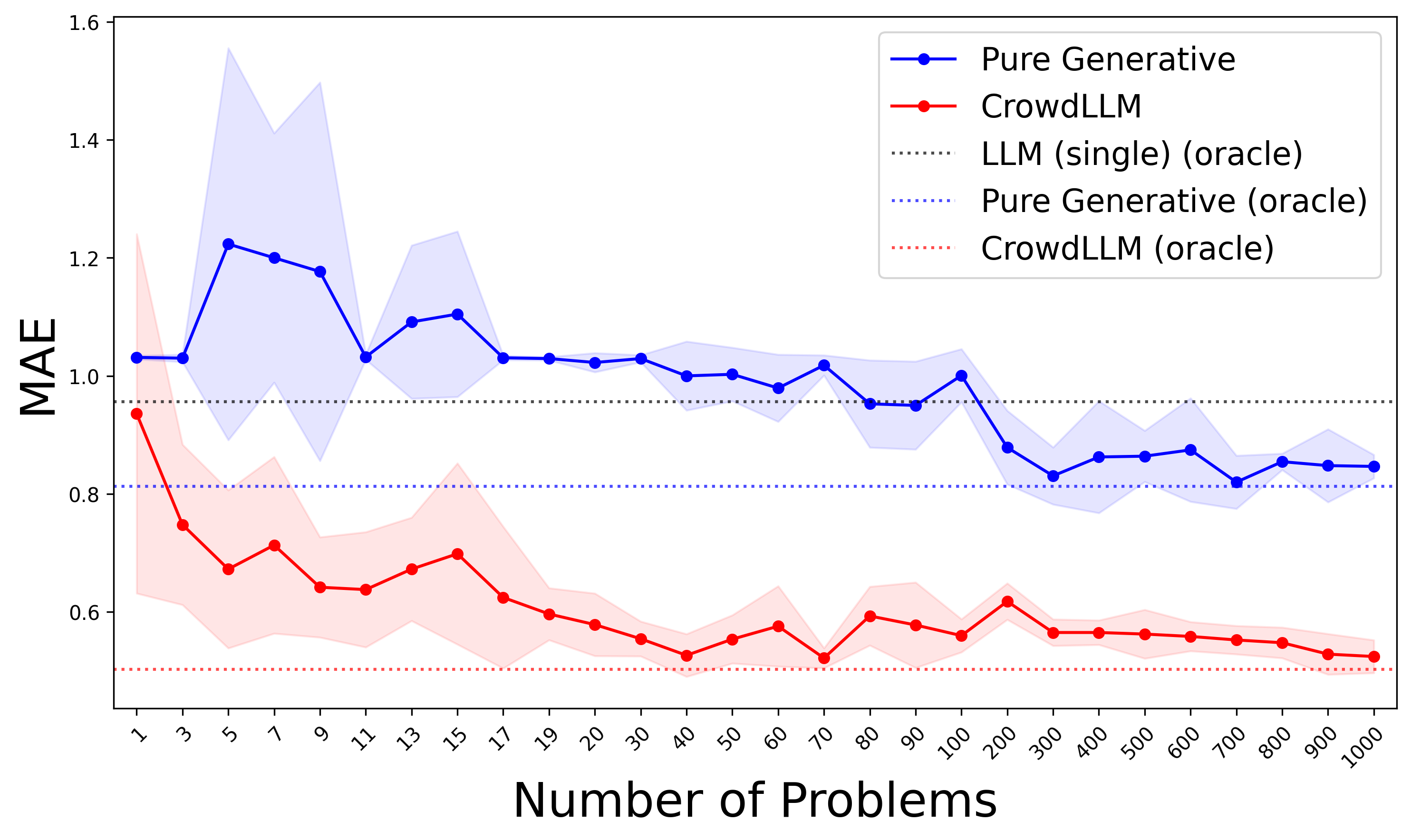}
        \caption{Offensiveness dataset}
        \label{fig:sub_offensiveness_mae_problem}
    \end{subfigure}%
    \begin{subfigure}{\dimexpr\textwidth/2\relax}
        \centering
        \includegraphics[width=\textwidth]{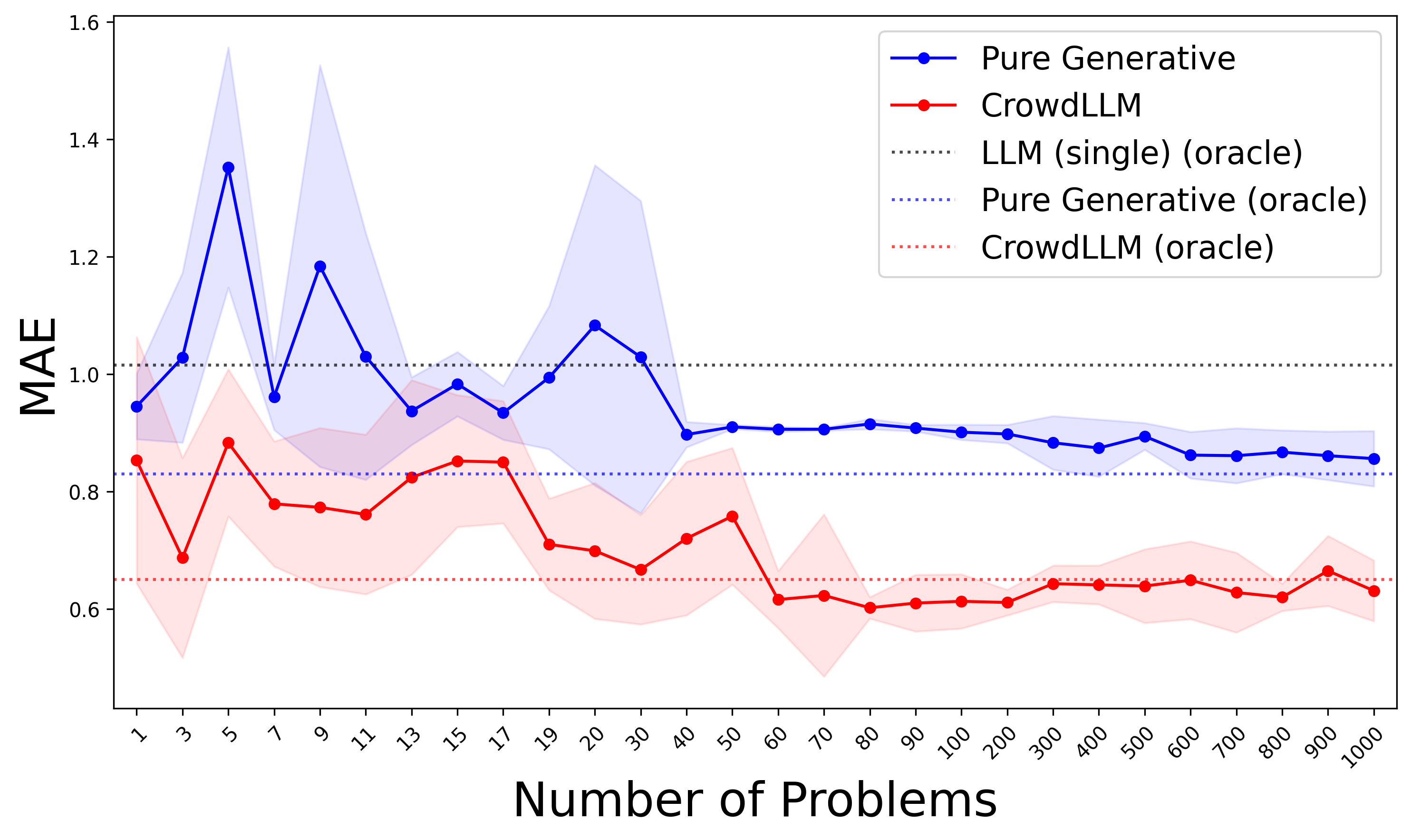}
        \caption{QA Difficulty dataset}
        \label{fig:sub_qa_difficulty_mae_problem}
    \end{subfigure}%
    \caption{MAE vs. Number of Problems.}
    \label{fig:mae_vs_num_problems}
\end{figure}

This subsection examines how the total number of problems affects model performance (MAE, RMSE, Avg. WD) for Pure Generative Model and our proposed CrowdLLM on the \textit{Offensiveness} and \textit{QA Difficulty} datasets. As Figures \ref{fig:wd_vs_num_problems}-\ref{fig:rmse_vs_num_problems} demonstrates, CrowdLLM consistently outperforms Pure Generative with lower errors across all datasets as problem numbers increase. Both models improve as problems grow from few to moderate (e.g., 1 to 100-200), after which gains diminish and metrics stabilize. Notably, CrowdLLM achieves better absolute errors and reaches stable, high-quality performance with fewer problems than Pure Generative (e.g., on \textit{Offensiveness} and \textit{QA Difficulty}, CrowdLLM stabilizes around 20-100 problems, while Pure Generative improves more gradually). This highlights CrowdLLM's data efficiency.
\begin{figure}[!t]
    \centering
    \begin{subfigure}{\dimexpr\textwidth/2\relax}
        \centering
        \includegraphics[width=\textwidth]{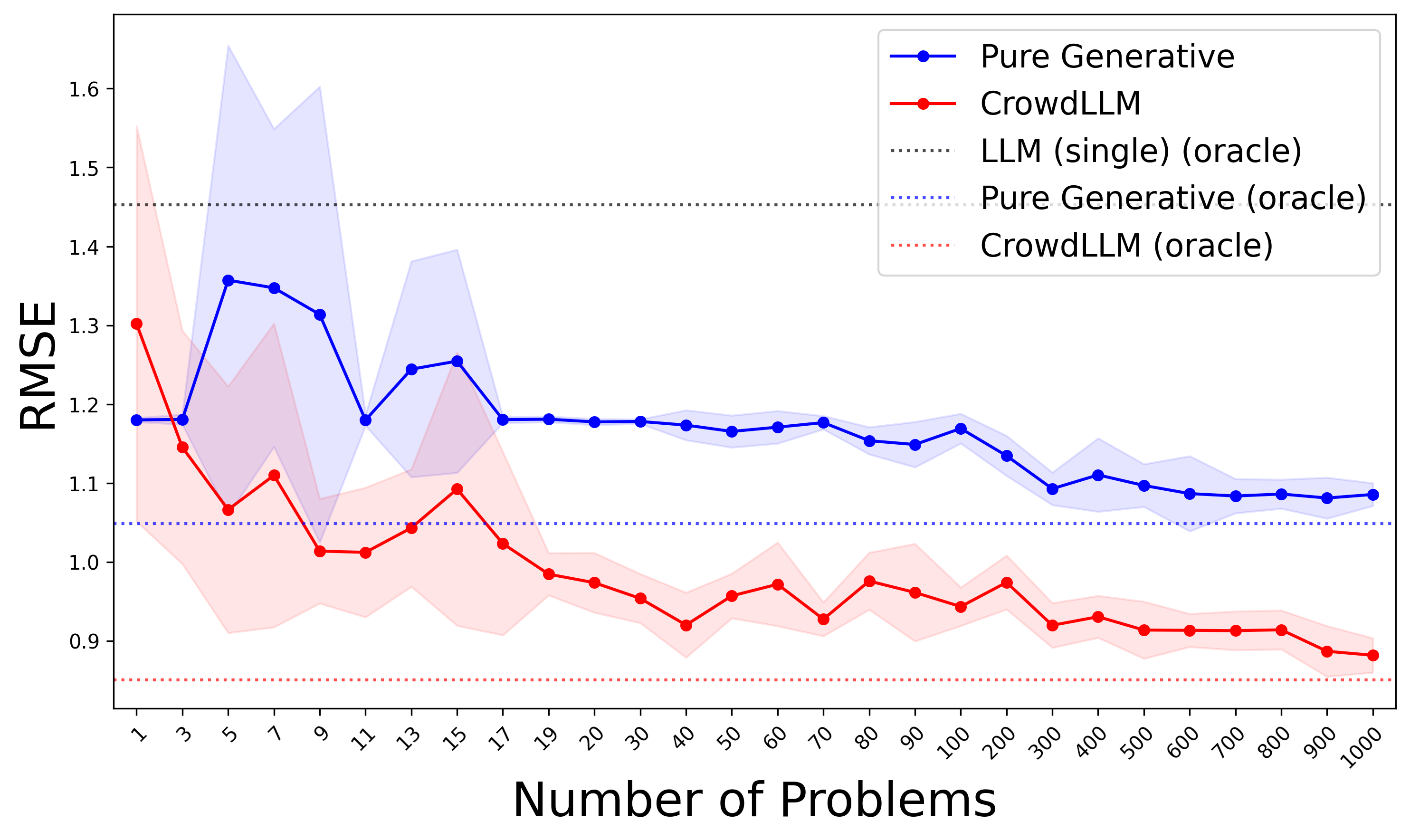}
        \caption{Offensiveness dataset}
        \label{fig:sub_offensiveness_rmse_problem}
    \end{subfigure}%
    \begin{subfigure}{\dimexpr\textwidth/2\relax}
        \centering
        \includegraphics[width=\textwidth]{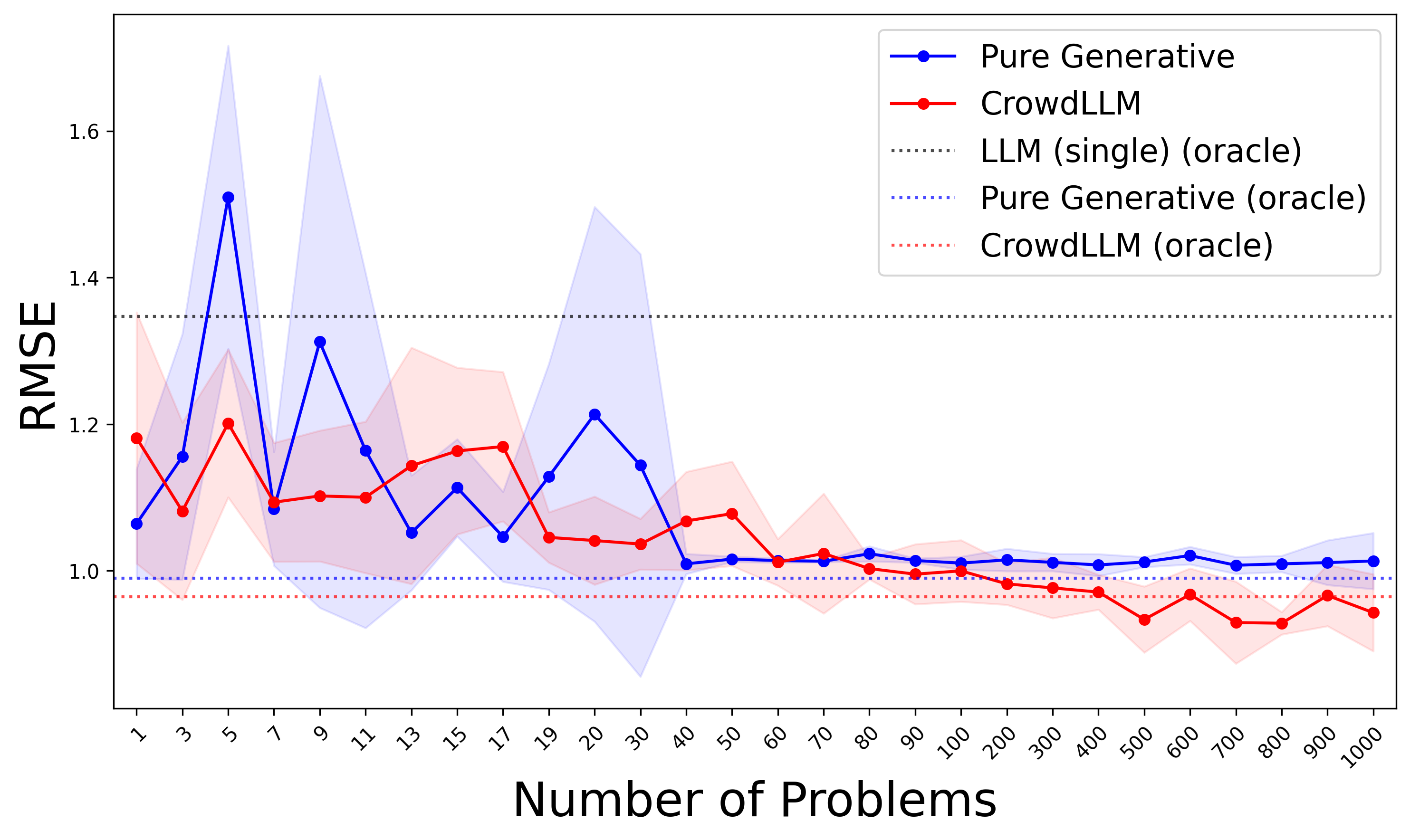}
        \caption{QA Difficulty dataset}
        \label{fig:sub_qa_difficulty_rmse_problem}
    \end{subfigure}%
    \caption{RMSE vs. Number of Problems.}
    \label{fig:rmse_vs_num_problems}
\end{figure}

\subsection{The impact of the number of unique workers}
This section analyzes how the number of unique workers affects Pure Generative Model and CrowdLLM performance (MAE, RMSE, Avg. WD) on the \textit{Offensiveness} and \textit{QA Difficulty} datasets. As Figures \ref{fig:wd_vs_num_unique_workers}-\ref{fig:rmse_vs_num_unique_workers} shows, CrowdLLM consistently shows lower error rates than Pure Generative across all metrics, regardless of unique worker count. For both models, errors substantially decrease when unique workers increase from few (e.g., 1) to moderate (e.g., 20-50), as diverse perspectives improve data quality. Beyond a point (e.g., ~50-100 workers for CrowdLLM, potentially more for Pure Generative), benefits diminish and metrics stabilize. CrowdLLM generally reaches a better performance plateau with fewer unique workers.

\begin{figure}[!t]
    \centering
    \begin{subfigure}{\dimexpr\textwidth/2\relax}
        \centering
        \includegraphics[width=\textwidth]{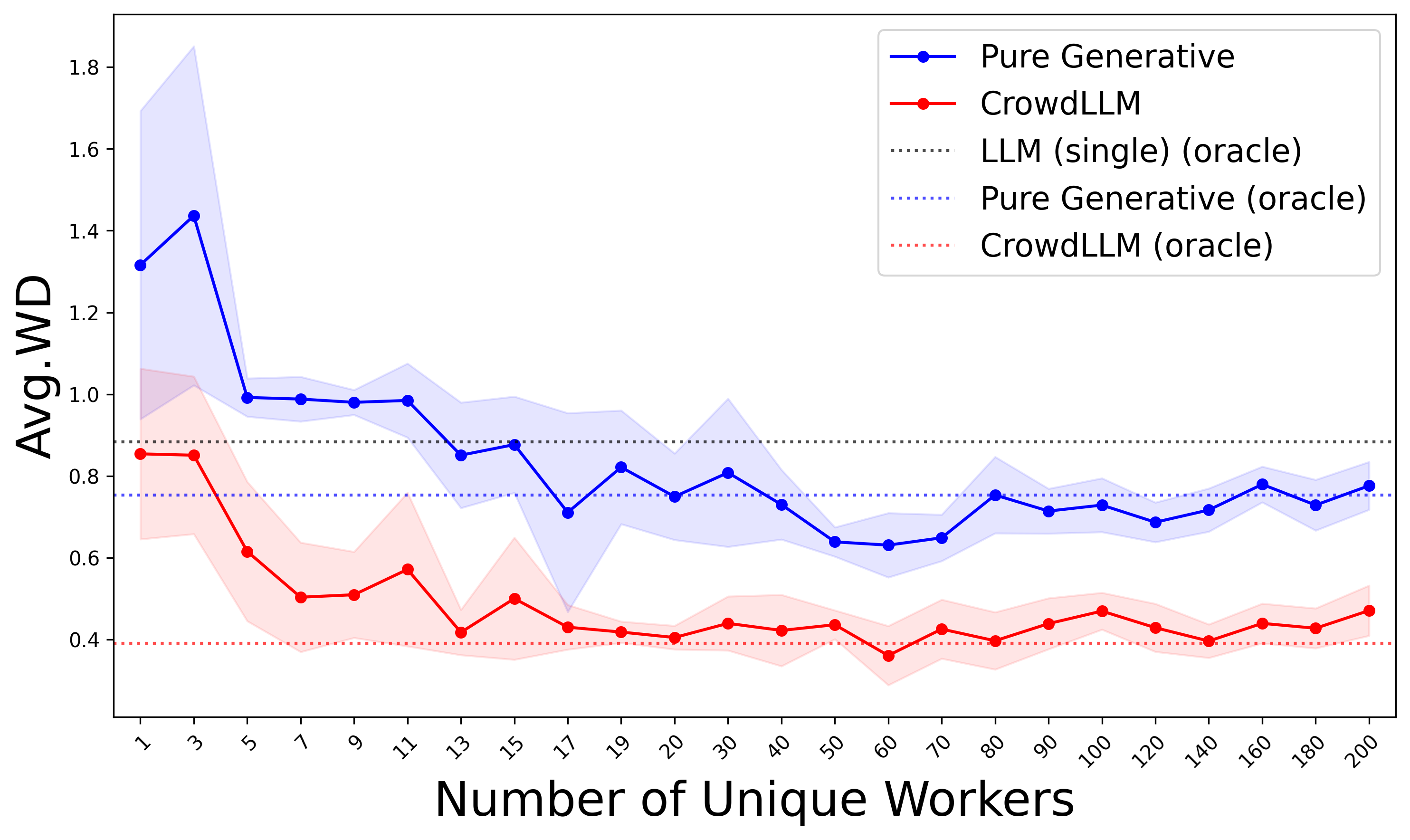}
        \caption{Offensiveness dataset}
        \label{fig:sub_offensiveness_wd_unique_workers}
    \end{subfigure}%
    \begin{subfigure}{\dimexpr\textwidth/2\relax}
        \centering
        \includegraphics[width=\textwidth]{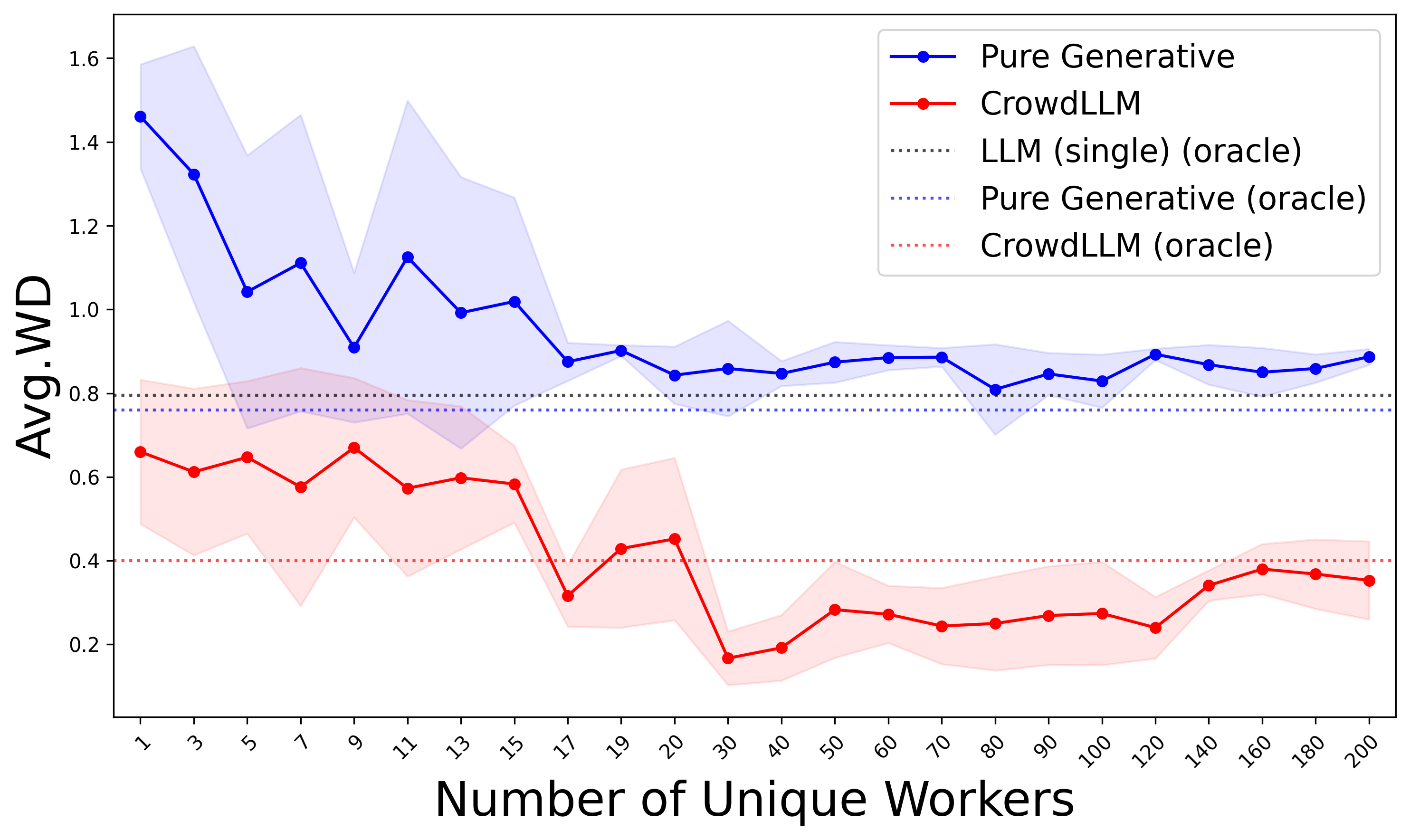}
        \caption{QA Difficulty dataset}
        \label{fig:sub_qa_difficulty_wd_unique_workers}
    \end{subfigure}%
    \caption{Avg. WD vs. Number of Unique Workers.}
    \label{fig:wd_vs_num_unique_workers}
\end{figure}

\begin{figure}[!t]
    \centering
    \begin{subfigure}{\dimexpr\textwidth/2\relax}
        \centering
        \includegraphics[width=\textwidth]{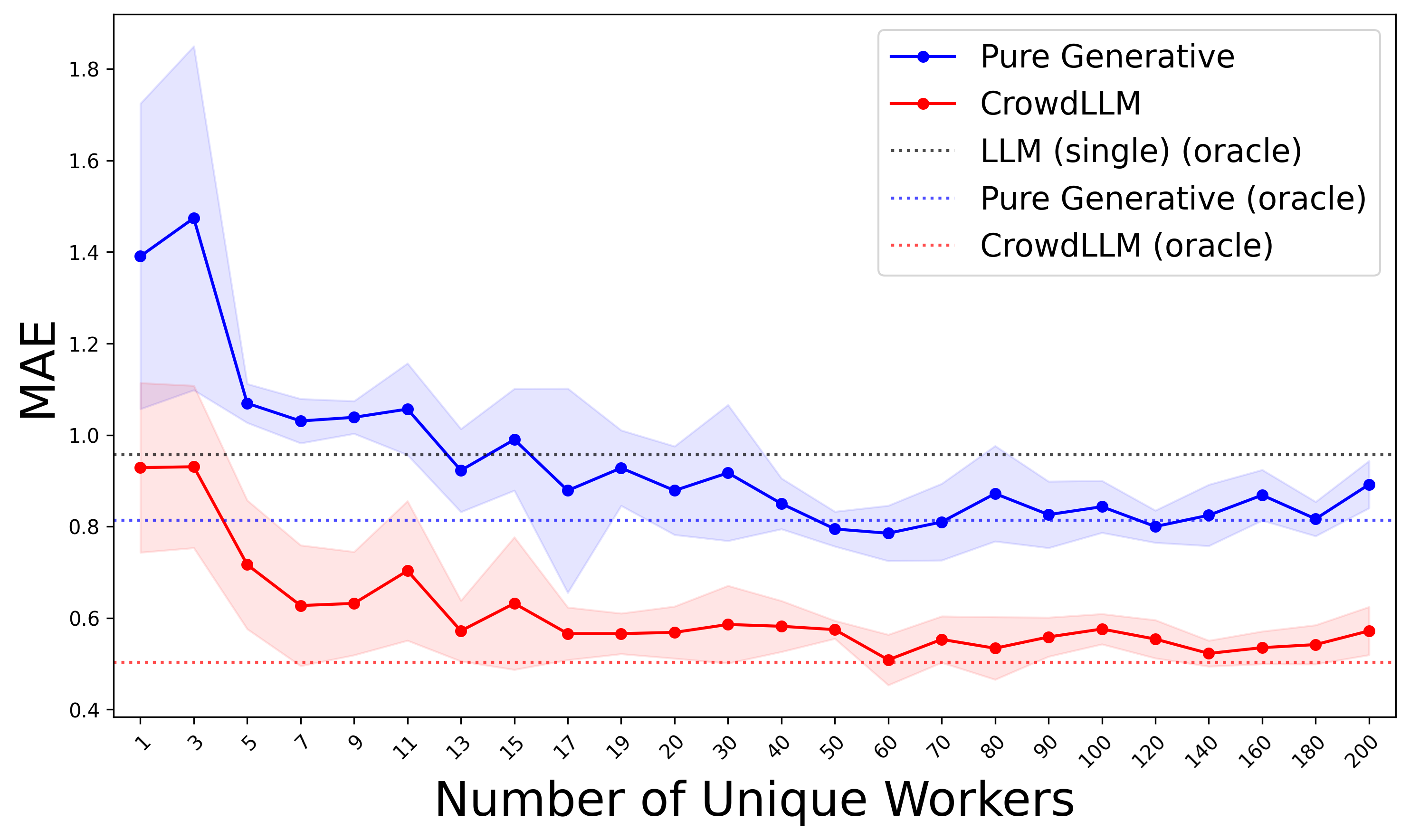}
        \caption{Offensiveness dataset}
        \label{fig:sub_offensiveness_mae_unique_workers}
    \end{subfigure}%
    \begin{subfigure}{\dimexpr\textwidth/2\relax}
        \centering
        \includegraphics[width=\textwidth]{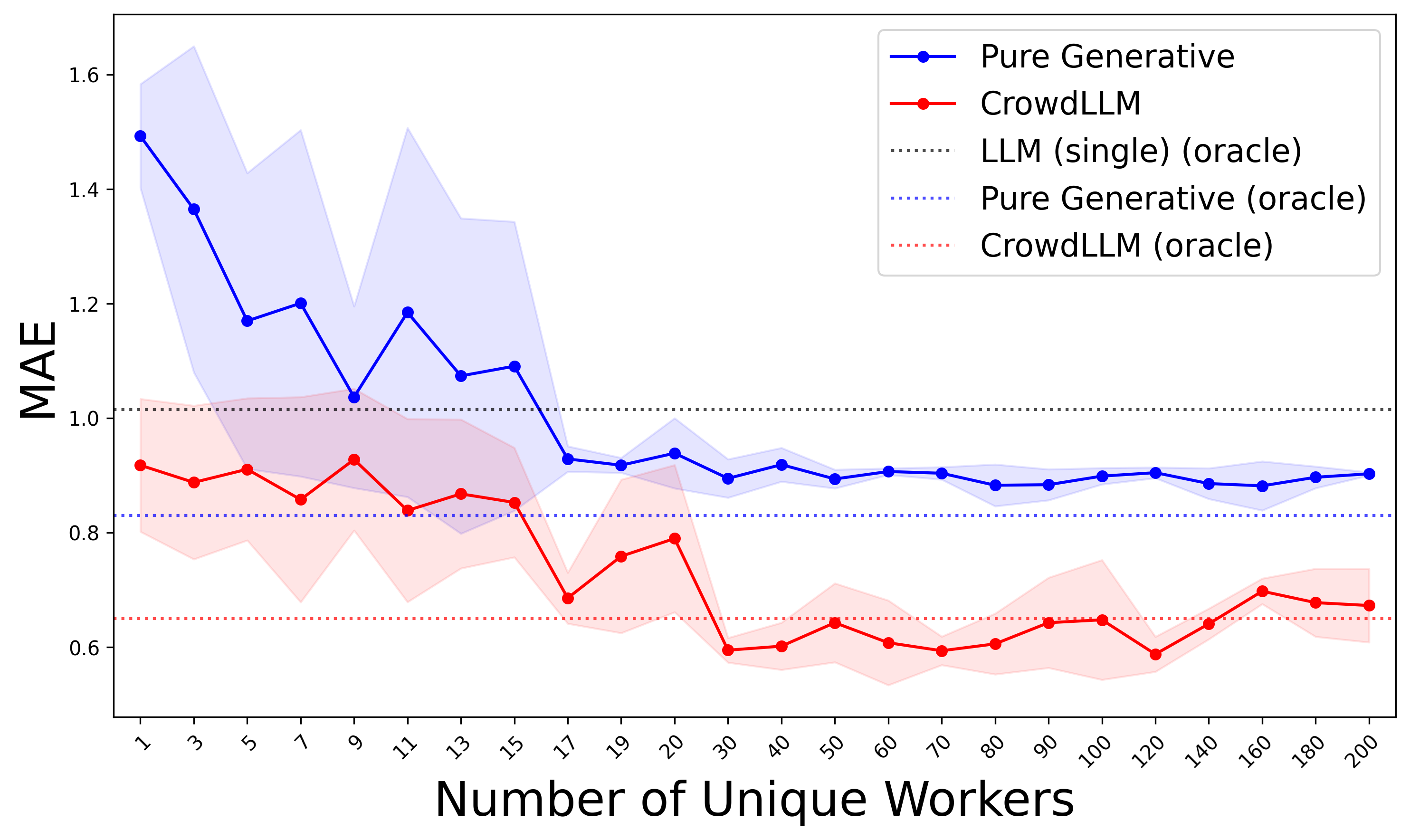}
        \caption{QA Difficulty dataset}
        \label{fig:sub_qa_difficulty_mae_unique_workers}
    \end{subfigure}%
    \caption{MAE vs. Number of Unique Workers.}
    \label{fig:mae_vs_num_unique_workers}
\end{figure}

\begin{figure}[!t]
    \centering
    \begin{subfigure}{\dimexpr\textwidth/2\relax}
        \centering
        \includegraphics[width=\textwidth]{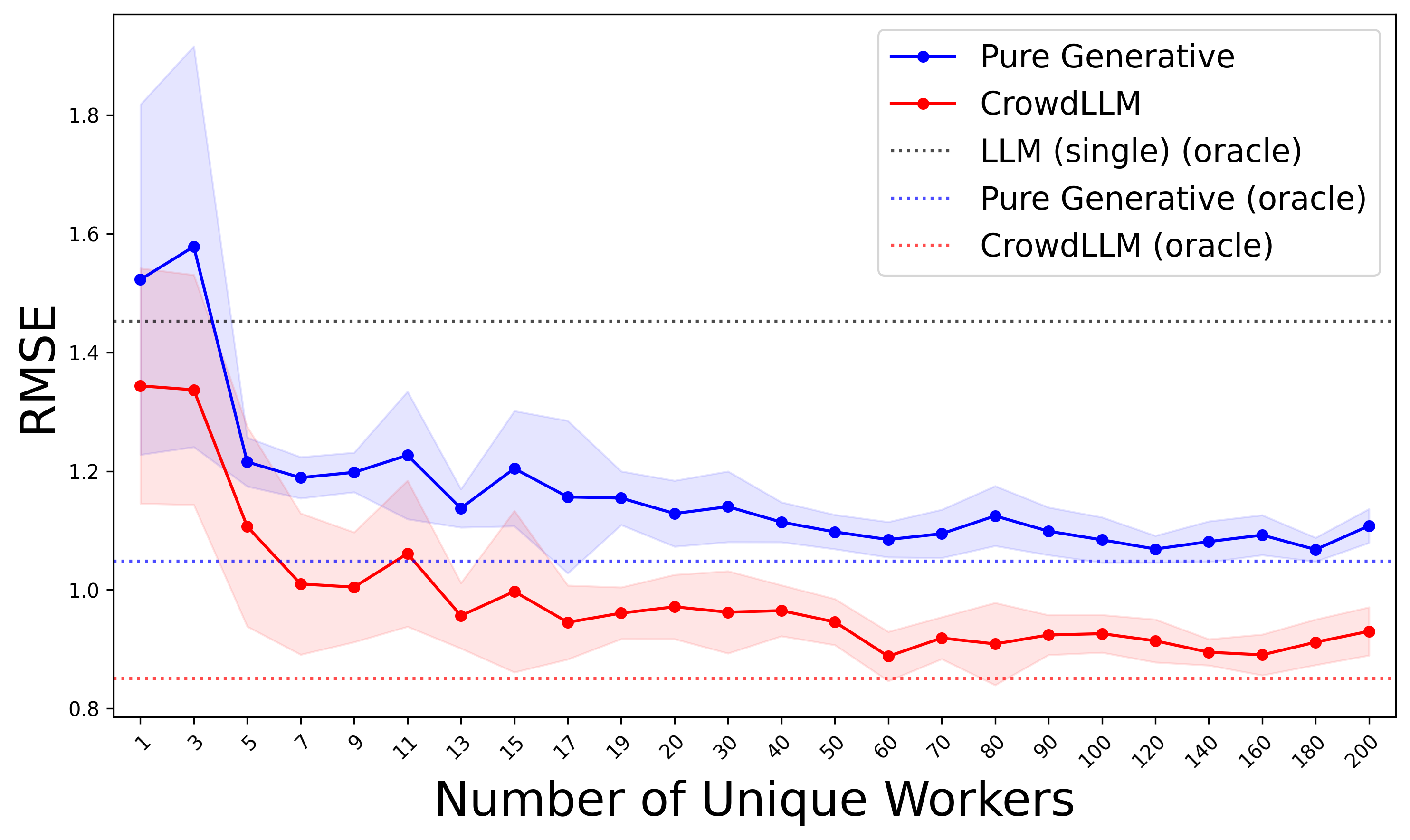}
        \caption{Offensiveness dataset}
        \label{fig:sub_offensiveness_rmse_unique_workers}
    \end{subfigure}%
    \begin{subfigure}{\dimexpr\textwidth/2\relax}
        \centering
        \includegraphics[width=\textwidth]{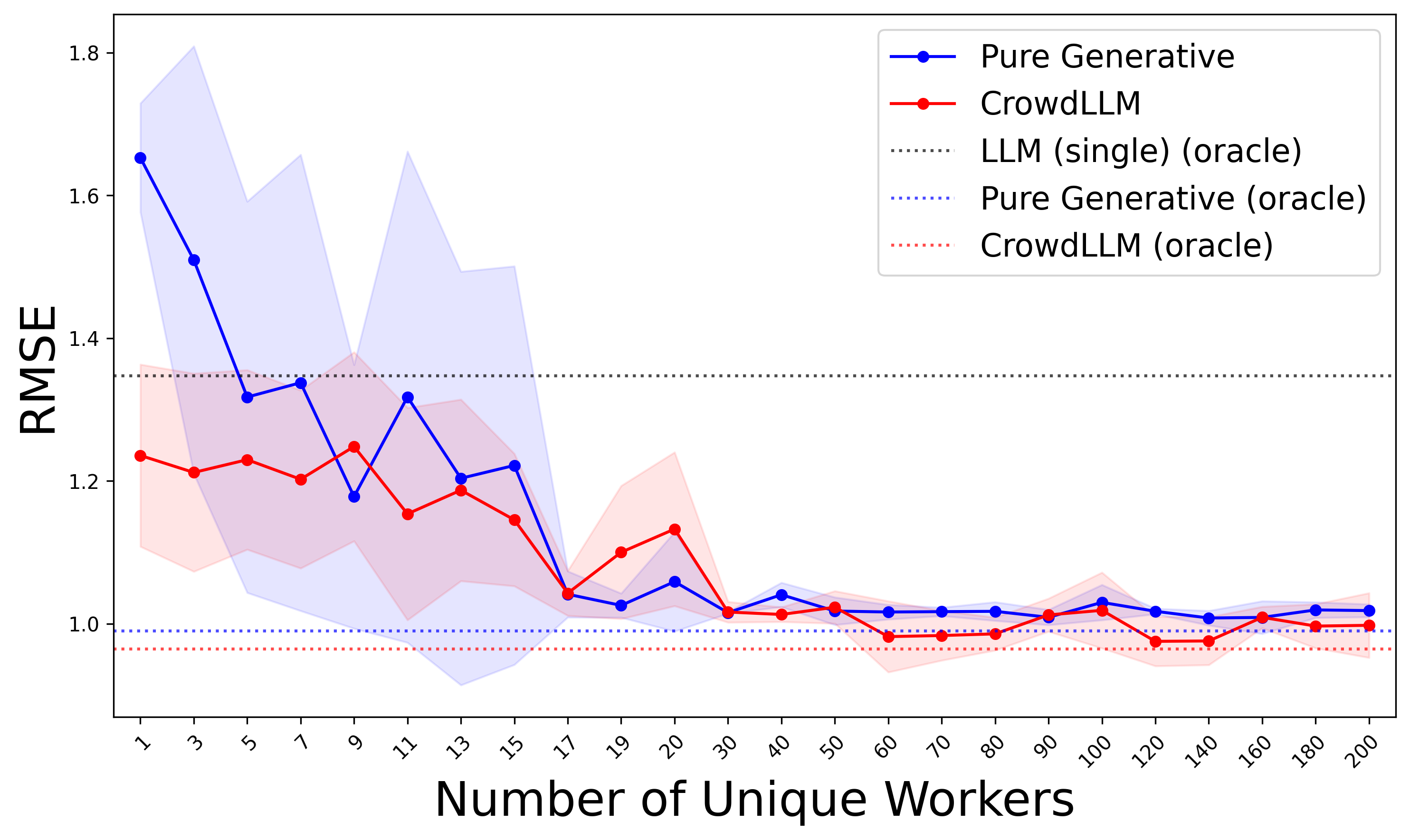}
        \caption{QA Difficulty dataset}
        \label{fig:sub_qa_difficulty_rmse_unique_workers}
    \end{subfigure}%
    \caption{RMSE vs. Number of Unique Workers.}
    \label{fig:rmse_vs_num_unique_workers}
\end{figure}

\subsection{The impact of the number of ratings}

\begin{figure}[!hb]
    \centering
    \begin{subfigure}{\dimexpr\textwidth/2\relax}
        \centering
        \includegraphics[width=\linewidth]{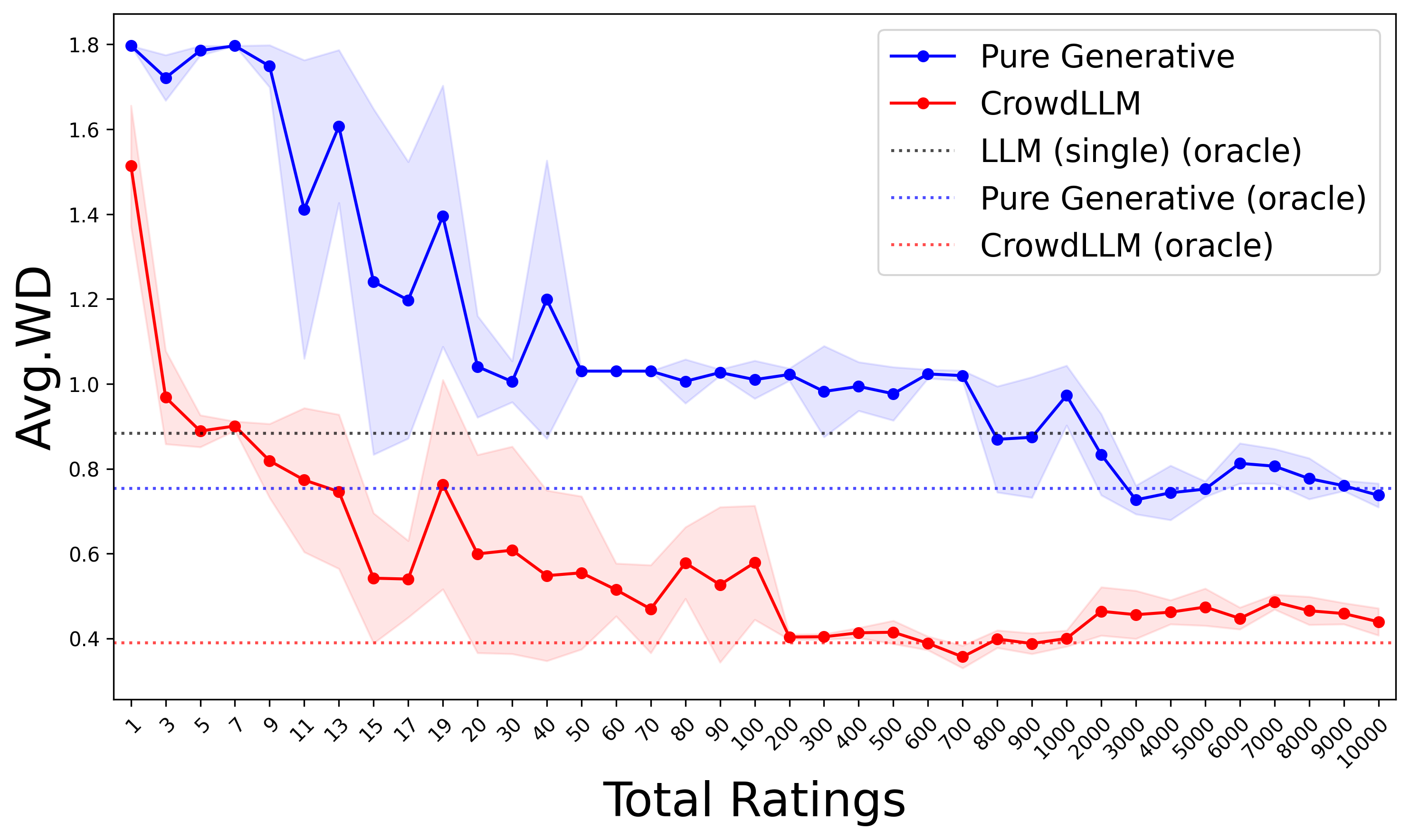}
        \caption{Offensiveness dataset}
        \label{fig:sub_offensiveness_wd_ratings}
    \end{subfigure}%
    \begin{subfigure}{\dimexpr\textwidth/2\relax}
        \centering
        \includegraphics[width=\linewidth]{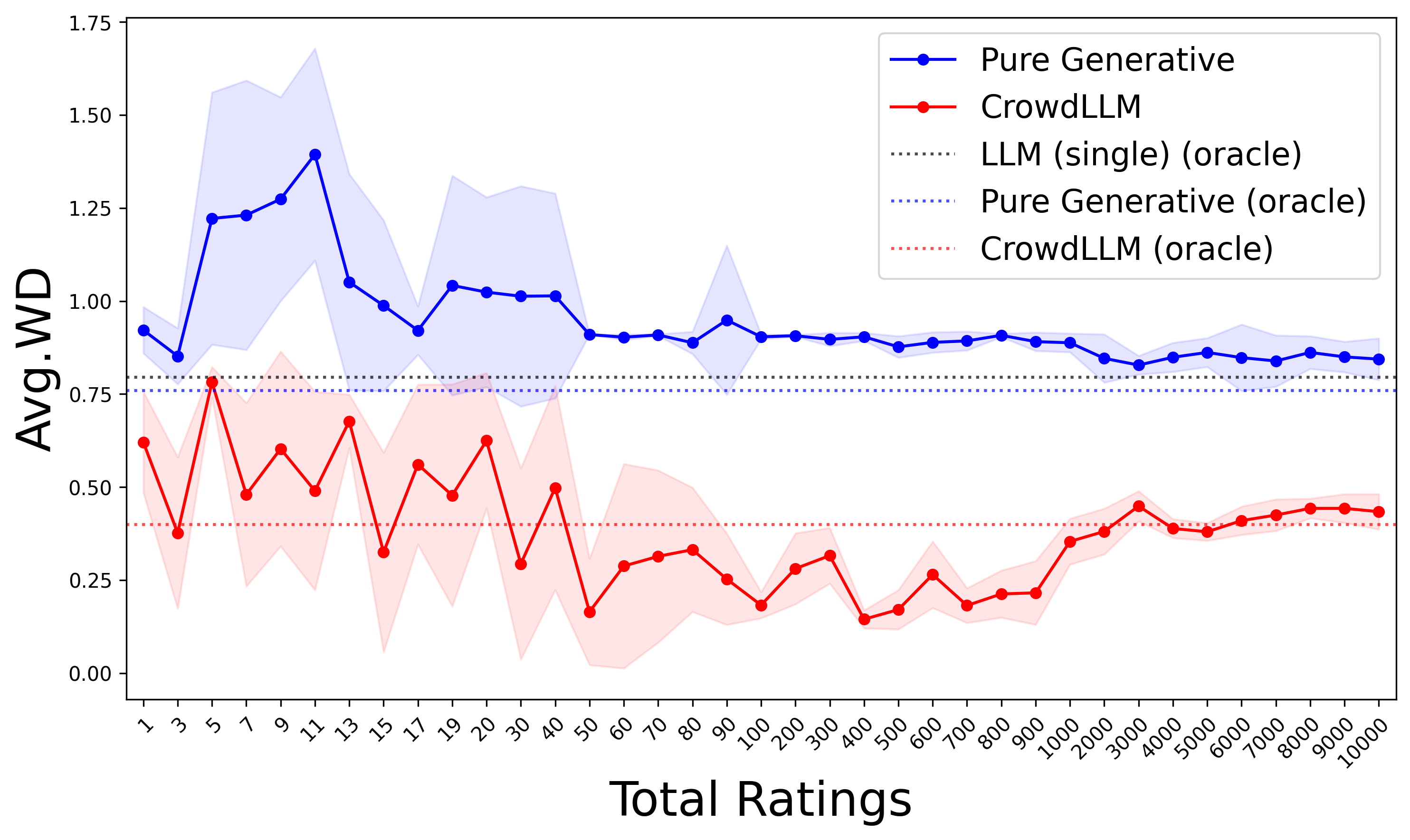}
        \caption{QA Difficulty dataset}
        \label{fig:sub_qa_difficulty_wd_ratings}
    \end{subfigure}%
    \caption{Avg. WD vs. Number of Ratings.}
    \label{fig:wd_vs_num_ratings}
\end{figure}

\begin{figure}[!ht]
    \centering
    \begin{subfigure}{\dimexpr\textwidth/2\relax}
        \centering
        \includegraphics[width=\linewidth]{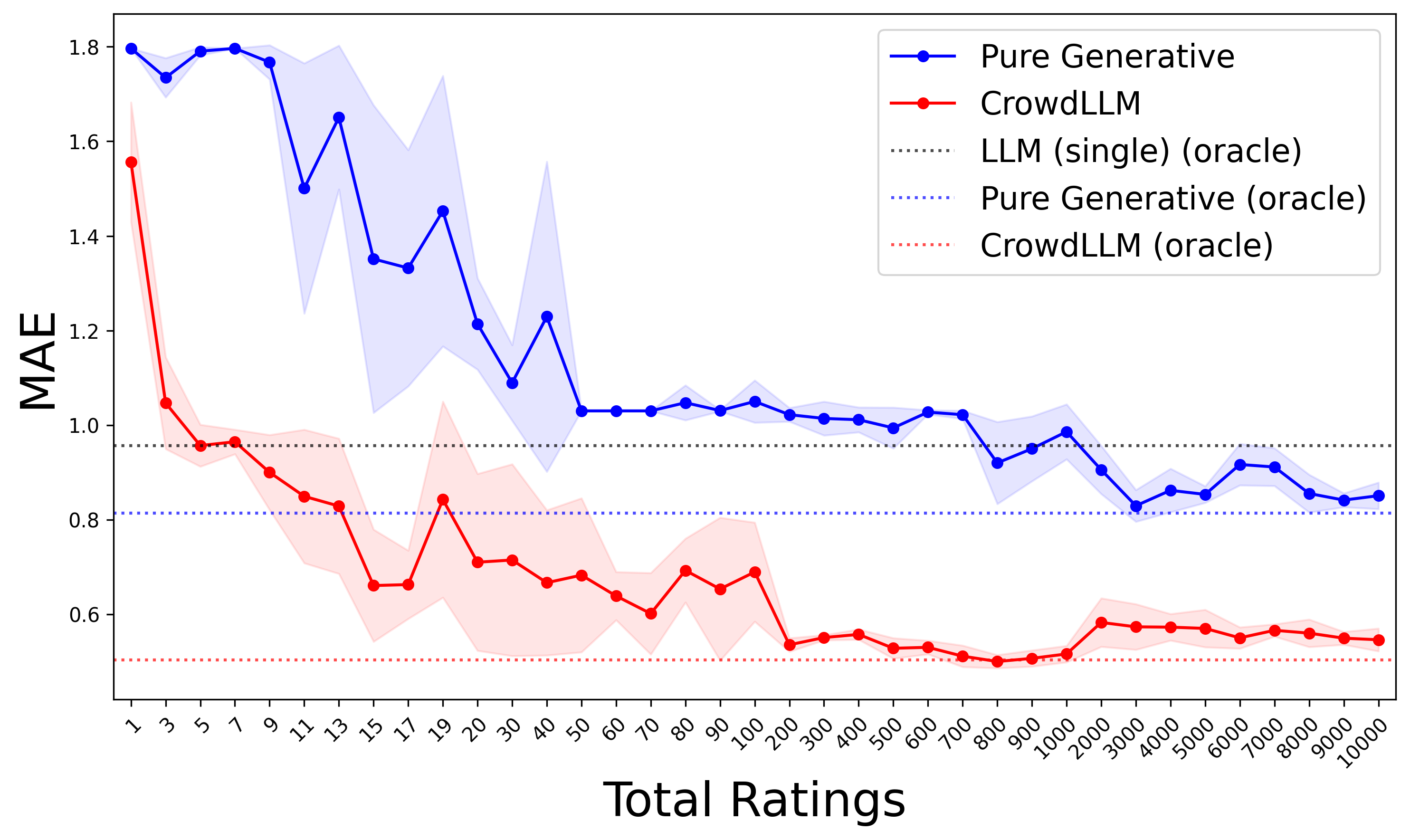}
        \caption{Offensiveness dataset}
        \label{fig:sub_offensiveness_mae_ratings}
    \end{subfigure}%
    \begin{subfigure}{\dimexpr\textwidth/2\relax}
        \centering
        \includegraphics[width=\linewidth]{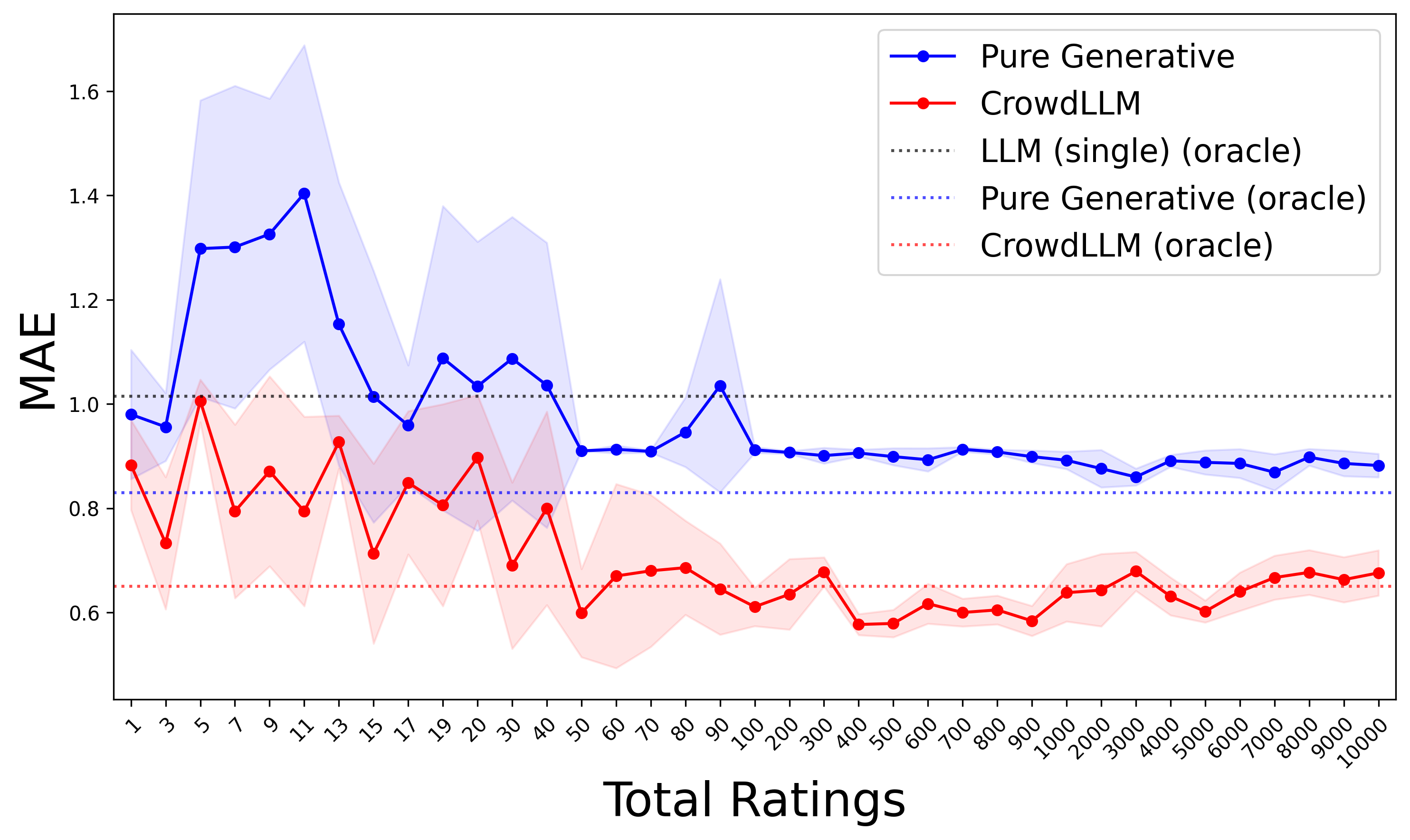}
        \caption{QA Difficulty dataset}
        \label{fig:sub_qa_difficulty_mae_ratings}
    \end{subfigure}%
    \caption{MAE vs. Number of Ratings.}
    \label{fig:mae_vs_num_ratings}
\end{figure}

\begin{figure}[!ht]
    \centering
    \begin{subfigure}{\dimexpr\textwidth/2\relax}
        \centering
        \includegraphics[width=\linewidth]{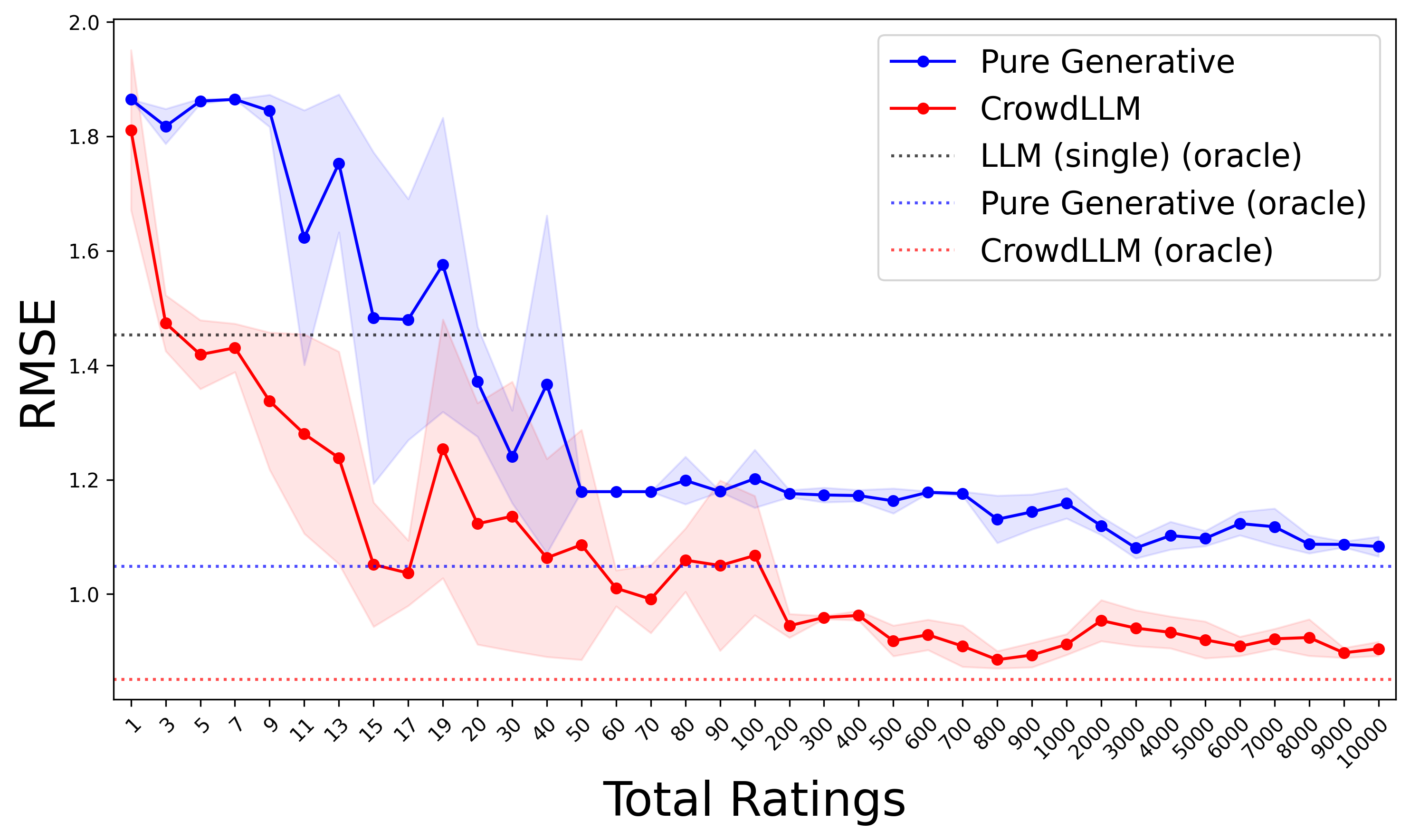}
        \caption{Offensiveness dataset}
        \label{fig:sub_offensiveness_rmse_ratings}
    \end{subfigure}%
    \begin{subfigure}{\dimexpr\textwidth/2\relax}
        \centering
        \includegraphics[width=\linewidth]{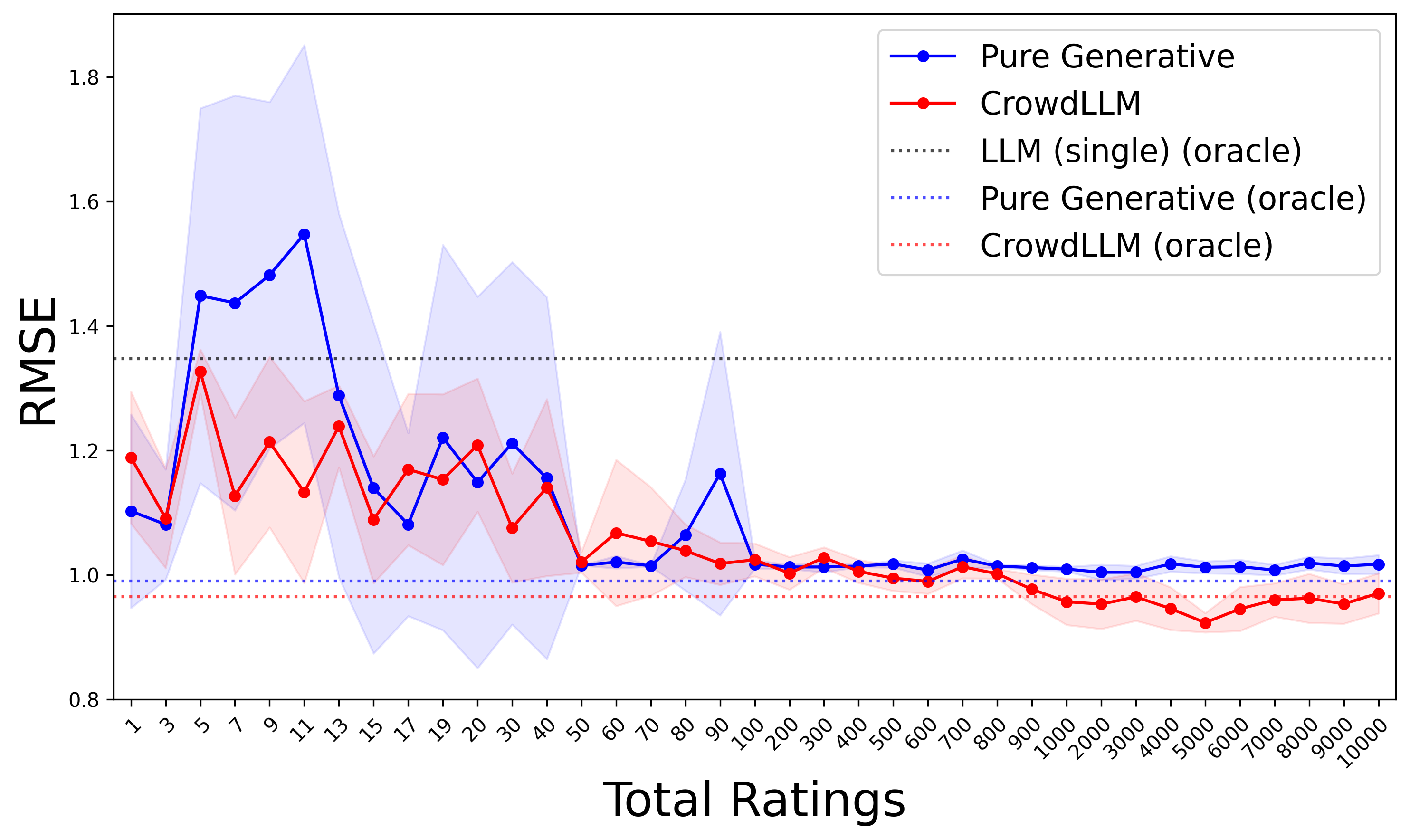}
        \caption{QA Difficulty dataset}
        \label{fig:sub_qa_difficulty_rmse_ratings}
    \end{subfigure}%
    \caption{RMSE vs. Number of Ratings.}
    \label{fig:rmse_vs_num_ratings}
\end{figure}

This section assesses how total ratings impact Pure Generative and CrowdLLM performance (MAE, RMSE, Avg. WD) on the \textit{Offensiveness}, and \textit{QA Difficulty} datasets. From Figures \ref{fig:wd_vs_num_ratings}-\ref{fig:rmse_vs_num_ratings}, we see that CrowdLLM consistently outperforms Pure Generative with lower errors as ratings increase. Both models improve with more ratings, with significant gains when increasing from few (e.g., 1-100) to hundreds/thousands (e.g., 1000-2000). Beyond a high volume (e.g., $>$2000-5000), improvements slow down and performance stabilizes. CrowdLLM achieves superior performance and often reaches optimal levels with fewer total ratings than Pure Generative, underscoring its data efficiency.

\subsection{The impact of the number of workers per problem}
This section analyzes how workers per problem (ranging from 1 to 8) influence Pure Generative and CrowdLLM performance (MAE, RMSE, Avg. WD) across the three datasets. The results are shown in Figures \ref{fig:wd_vs_num_workers_per_problem}-\ref{fig:rmse_vs_num_workers_per_problem}. 
\begin{figure}[!b]
    \centering
    \begin{subfigure}{\dimexpr\textwidth/2\relax}
        \centering
        \includegraphics[width=\linewidth]{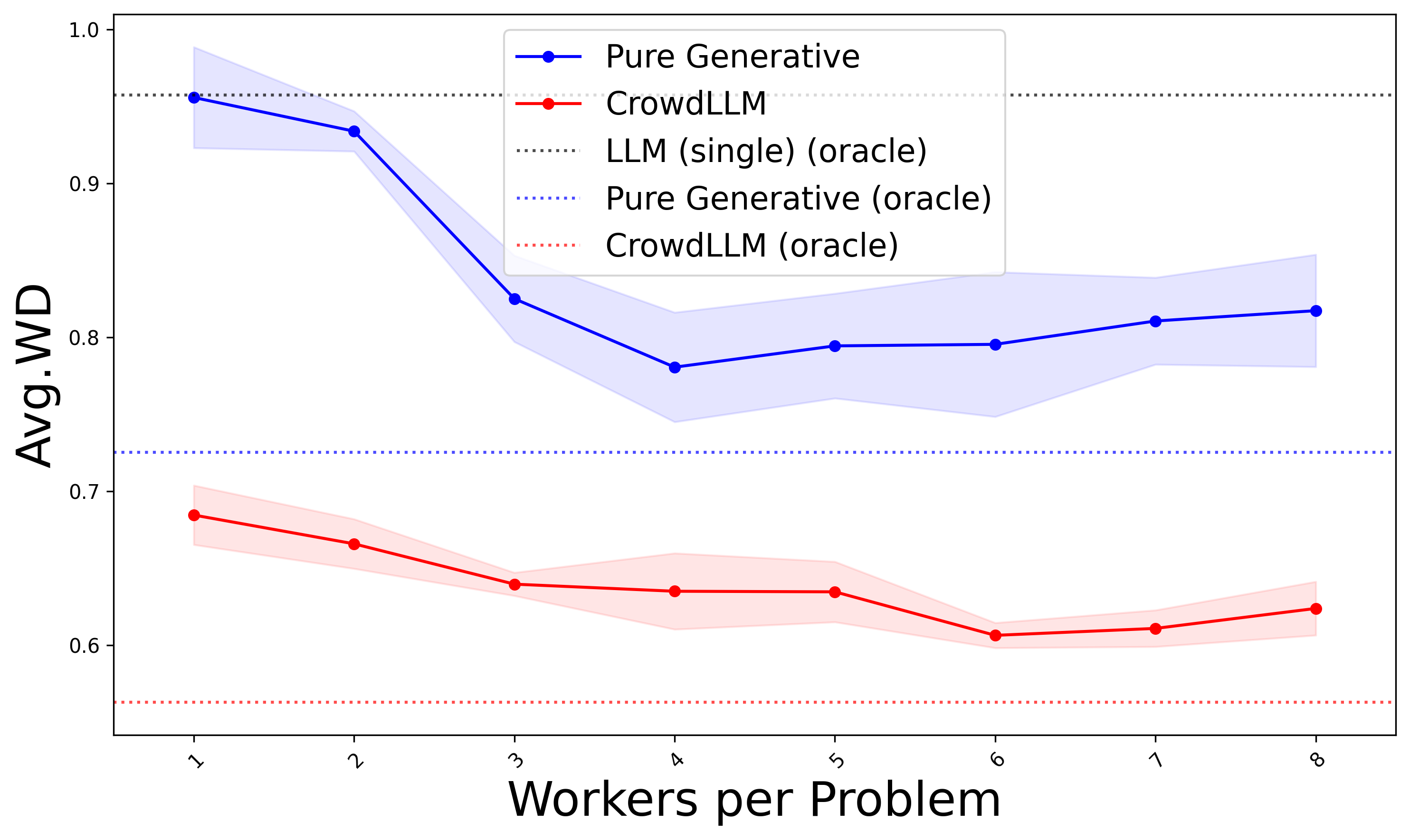}
        \caption{Offensiveness dataset}
        \label{fig:sub_offensiveness_wd_workers_per_problem}
    \end{subfigure}%
    \begin{subfigure}{\dimexpr\textwidth/2\relax}
        \centering
        \includegraphics[width=\linewidth]{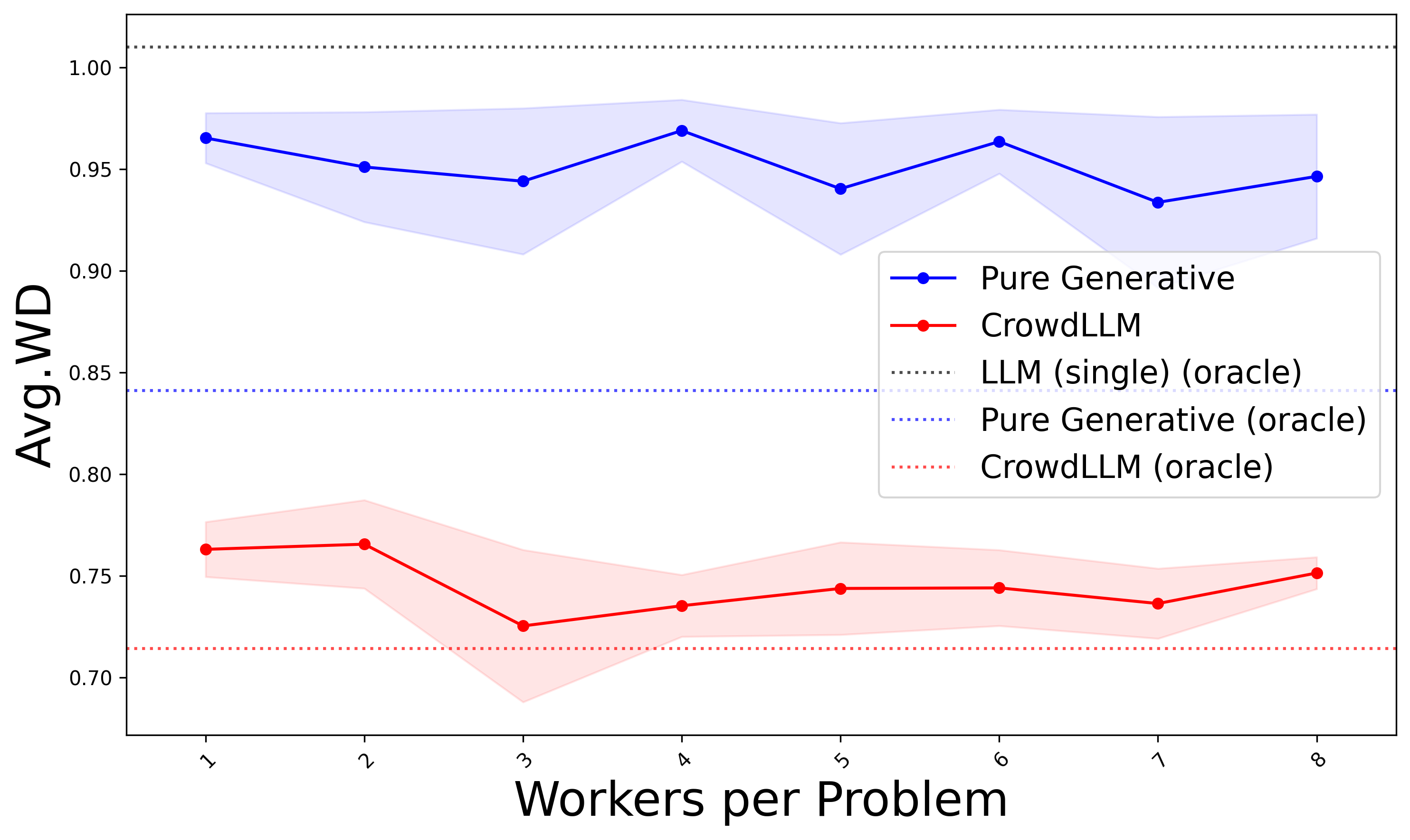}
        \caption{QA Difficulty dataset}
        \label{fig:sub_qa_difficulty_wd_workers_per_problem}
    \end{subfigure}%
    \caption{Avg. WD vs. Number of Workers per Problem.}
    \label{fig:wd_vs_num_workers_per_problem}
\end{figure}

\begin{figure}[!b]
    \centering
    \begin{subfigure}{\dimexpr\textwidth/2\relax}
        \centering
        \includegraphics[width=\linewidth]{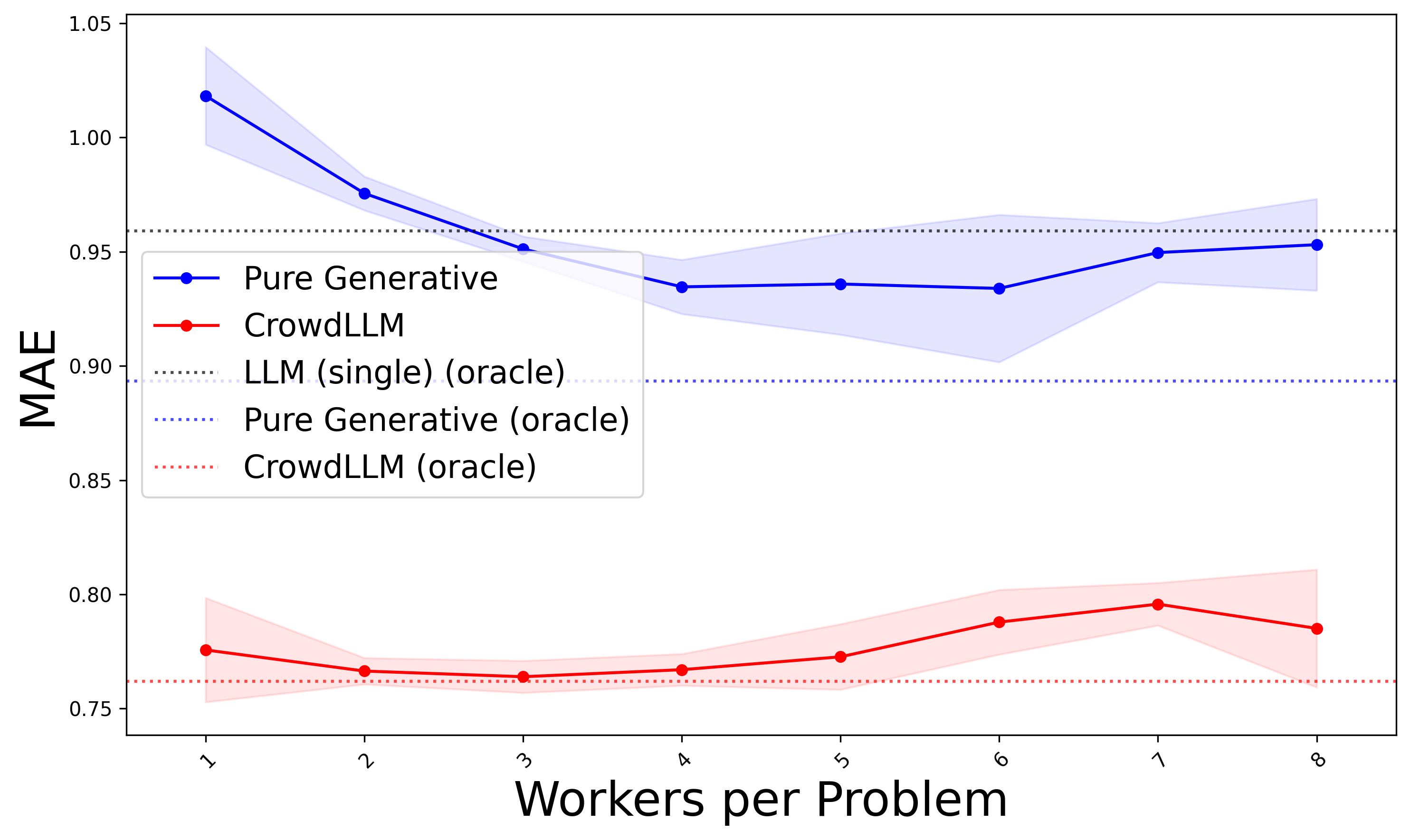}
        \caption{Offensiveness dataset}
        \label{fig:sub_offensiveness_mae_workers_per_problem}
    \end{subfigure}%
    \begin{subfigure}{\dimexpr\textwidth/2\relax}
        \centering
        \includegraphics[width=\linewidth]{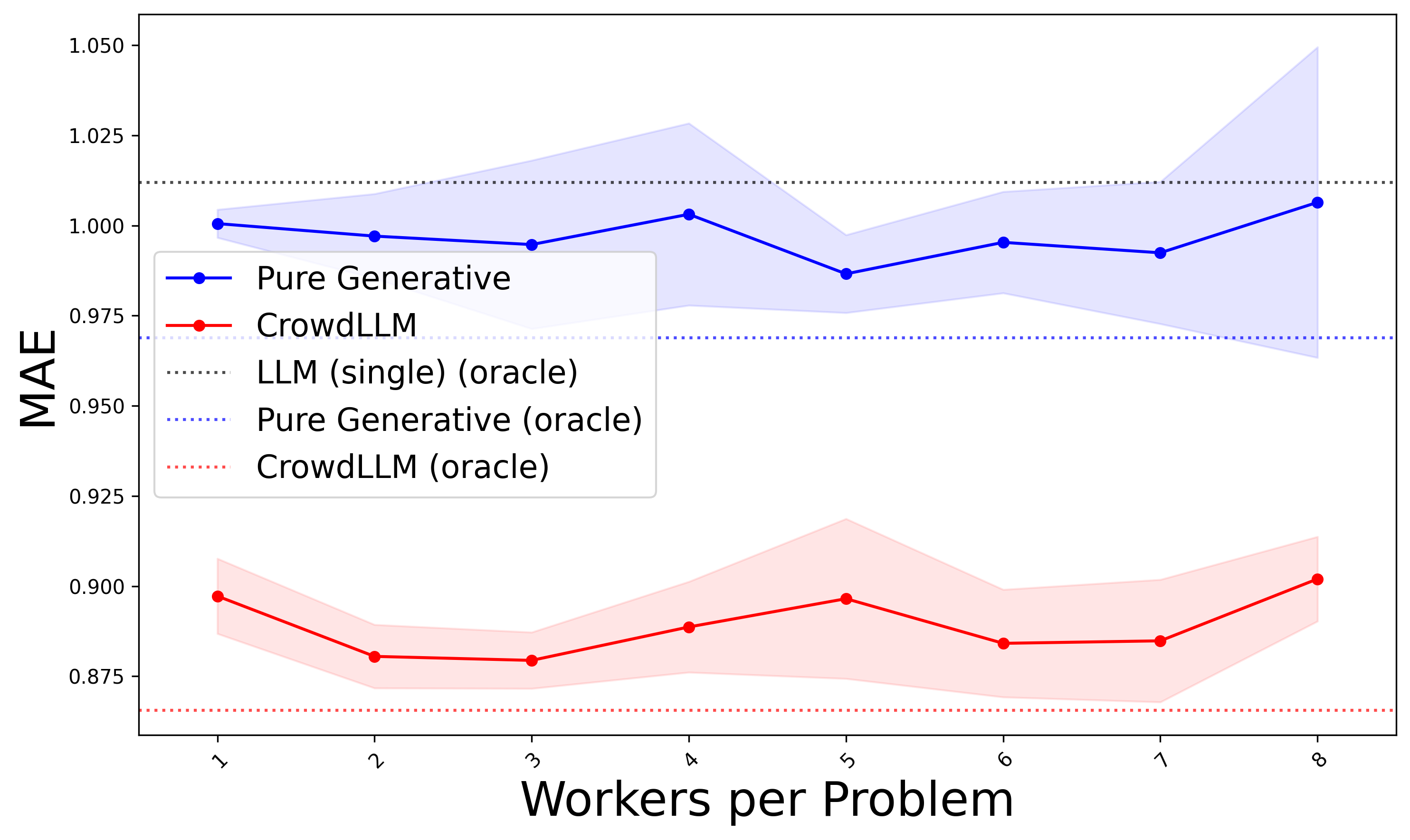}
        \caption{QA Difficulty dataset}
        \label{fig:sub_qa_difficulty_mae_workers_per_problem}
    \end{subfigure}%
    \caption{MAE vs. Number of Workers per Problem.}
    \label{fig:mae_vs_num_workers_per_problem}
\end{figure}

\begin{figure}[!t]
    \centering
    \begin{subfigure}{\dimexpr\textwidth/2\relax}
        \centering
        \includegraphics[width=\linewidth]{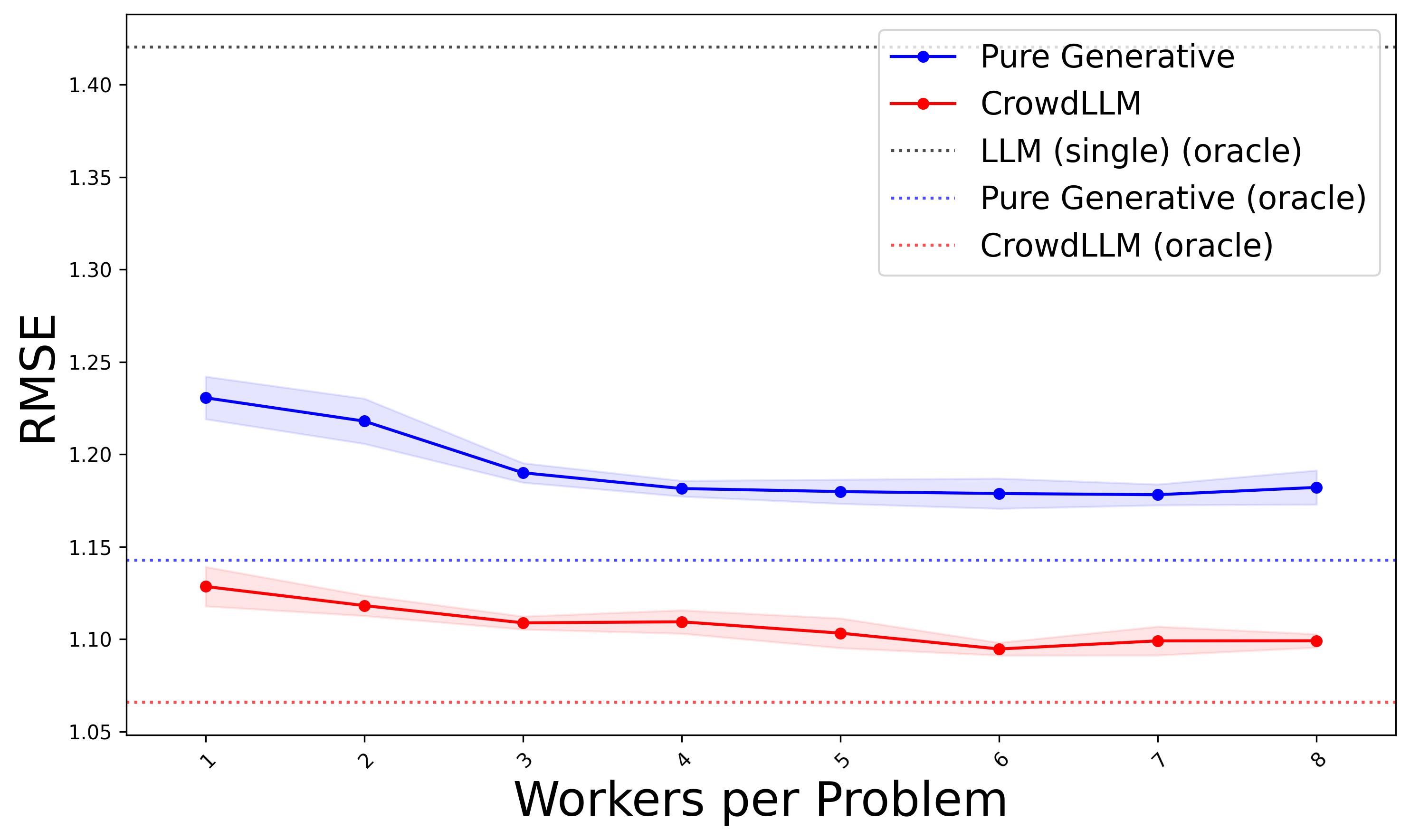}
        \caption{Offensiveness dataset}
        \label{fig:sub_offensiveness_rmse_workers_per_problem}
    \end{subfigure}%
    \begin{subfigure}{\dimexpr\textwidth/2\relax}
        \centering
        \includegraphics[width=\linewidth]{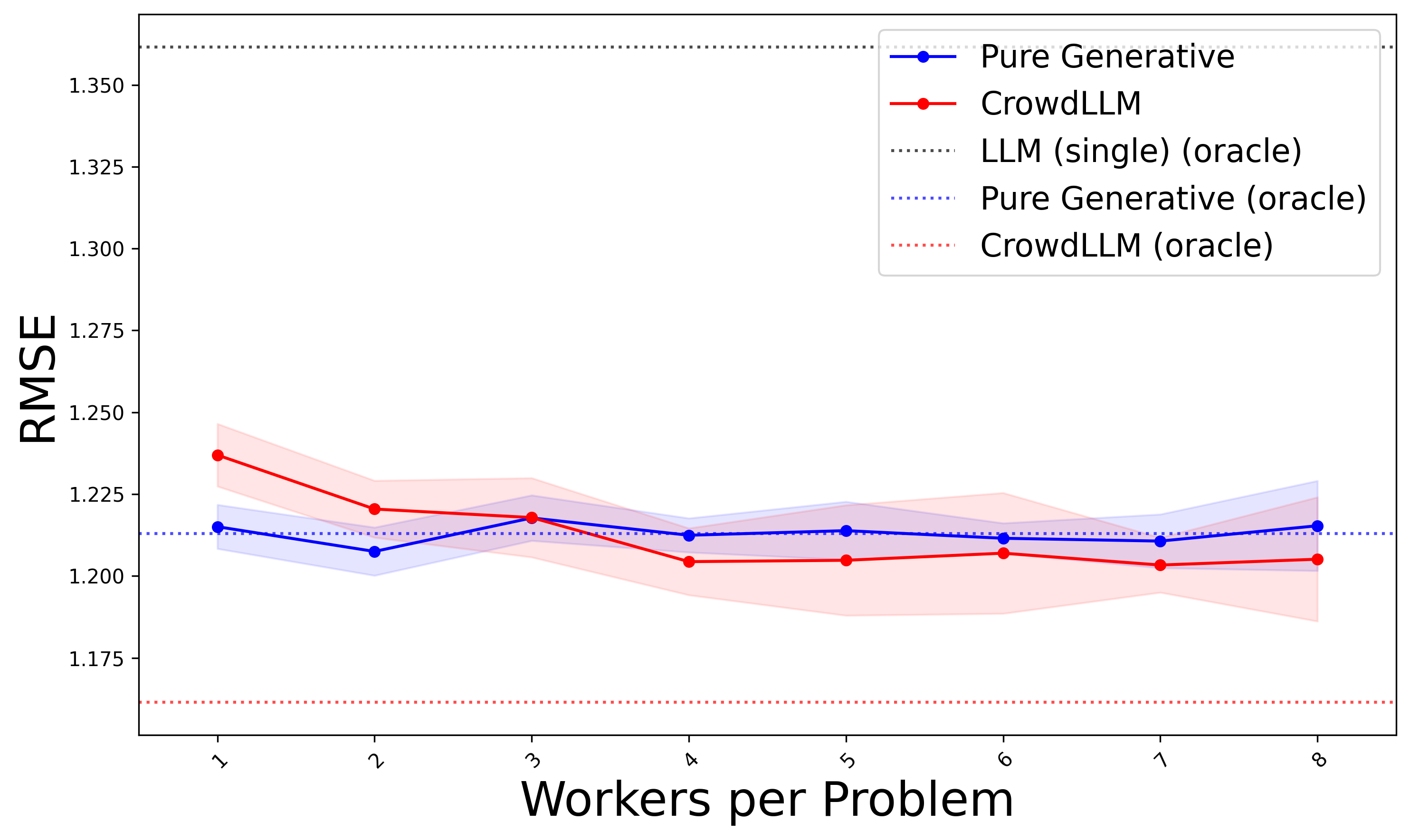}
        \caption{QA Difficulty dataset}
        \label{fig:sub_qa_difficulty_rmse_workers_per_problem}
    \end{subfigure}%
    \caption{RMSE vs. Number of Workers per Problem.}
    \label{fig:rmse_vs_num_workers_per_problem}
\end{figure}

CrowdLLM consistently shows markedly lower errors than Pure Generative. Increasing workers per problem from one generally improves aggregated label quality and reduces errors for both models, most notably when moving from 1 to 2-3 workers. CrowdLLM typically reaches optimal performance or diminishing returns with few workers (e.g., 2-4); for instance, on \textit{Offensiveness}, its Avg. WD stabilizes after 3 workers. Adding more workers (up to 8) offers little further benefit for CrowdLLM. Pure Generative Model also improves but maintains higher error rates than CrowdLLM.


\subsection{Cost-saving prospect of CrowdLLM}\label{performance_tolerance_appendix}
To assess CrowdLLM's cost-effectiveness, we analyze the LLMs needed to approximate ground-truth ratings within various tolerances ($\pm \{0.1, 0.2, 0.3, 0.4\}$) across datasets. 

\begin{figure}[!htbp]
    \centering
    \begin{subfigure}{0.33\textwidth}
        \centering
        \includegraphics[width=\textwidth]{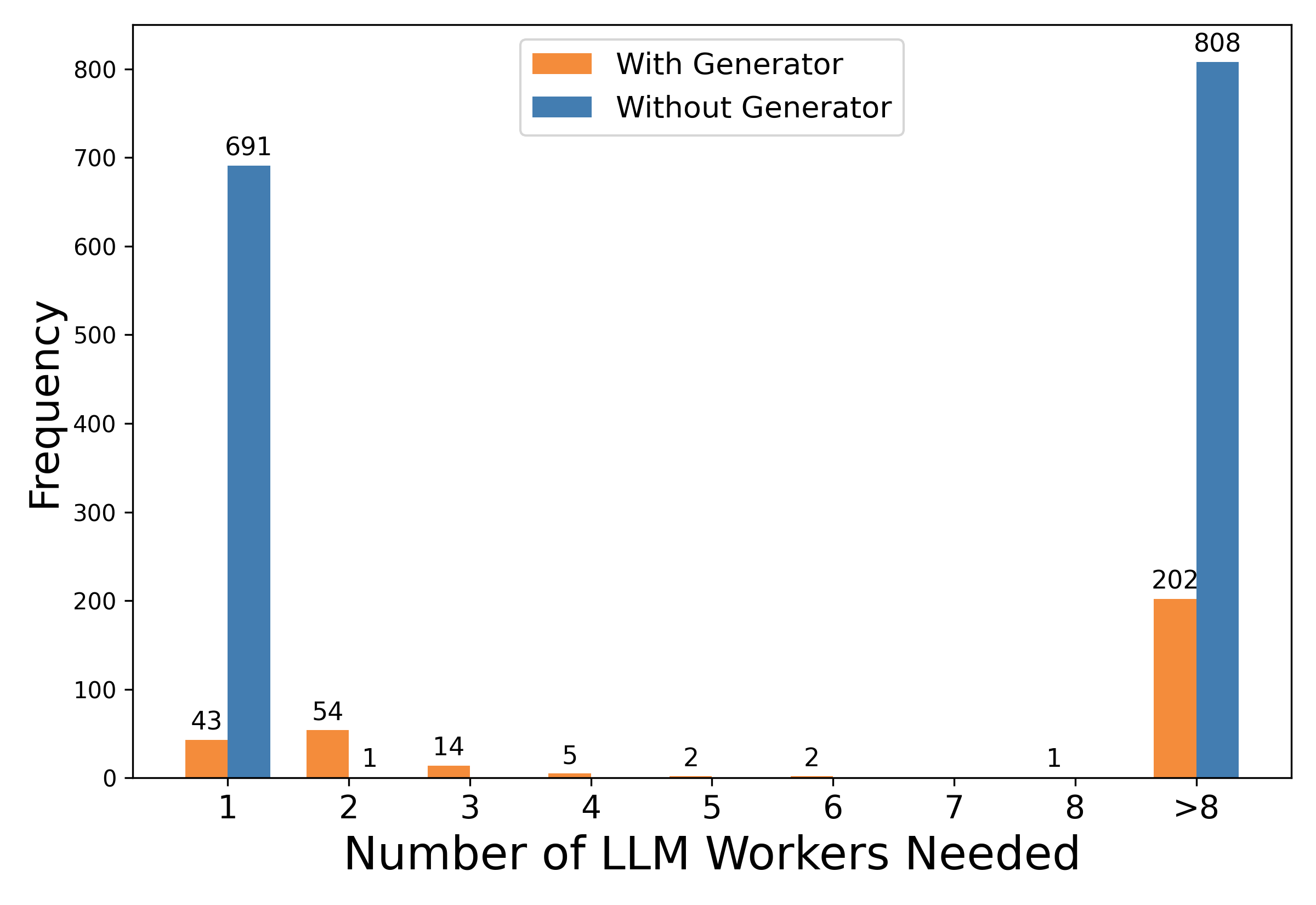}
        \caption{Ground truth rating $\pm 0.1$}
    \end{subfigure}%
    \begin{subfigure}{0.33\textwidth}
        \centering
        \includegraphics[width=\textwidth]{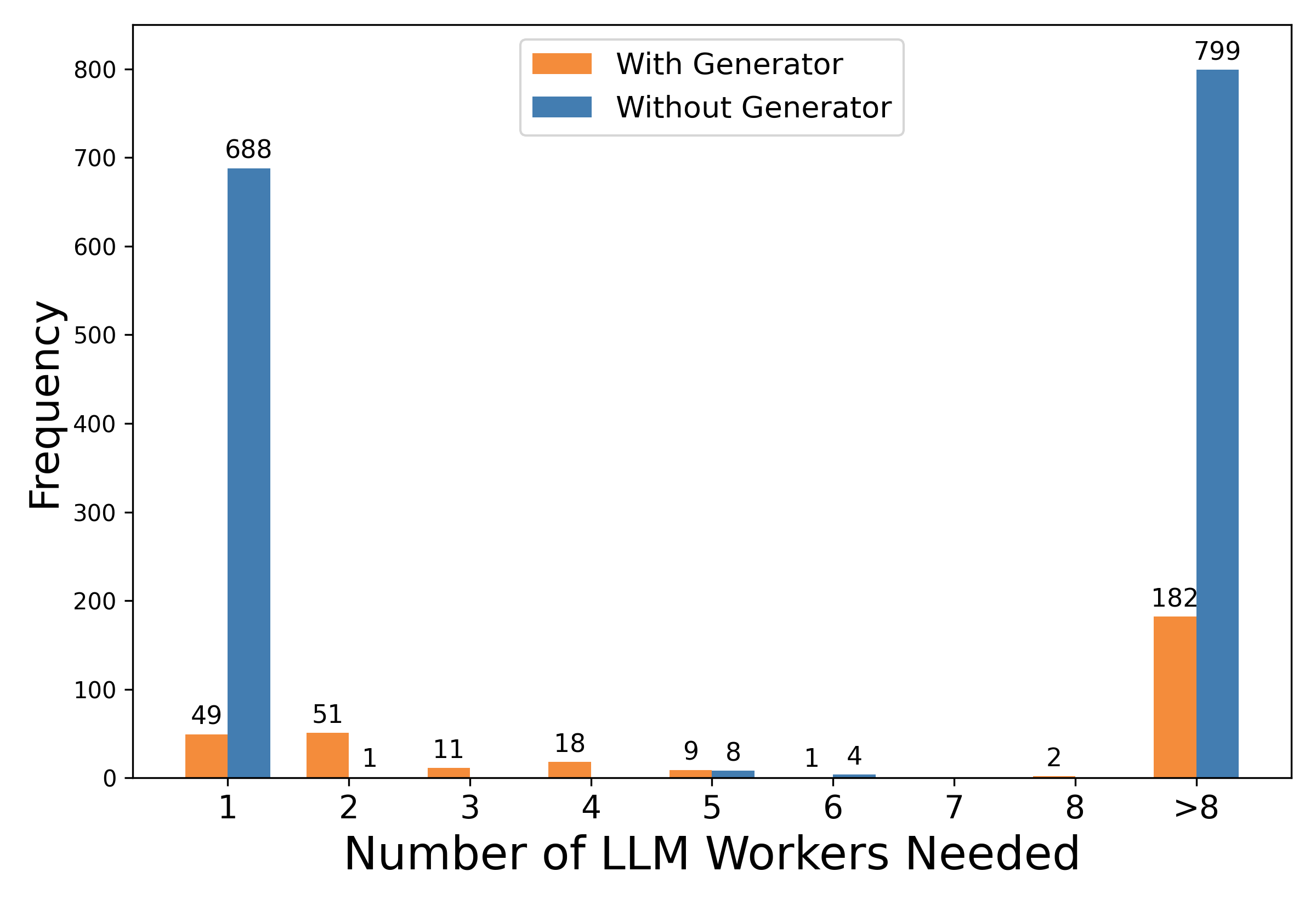}
        \caption{Ground truth rating $\pm 0.2$}
    \end{subfigure}%
    \begin{subfigure}{0.33\textwidth}
        \centering
        \includegraphics[width=\textwidth]{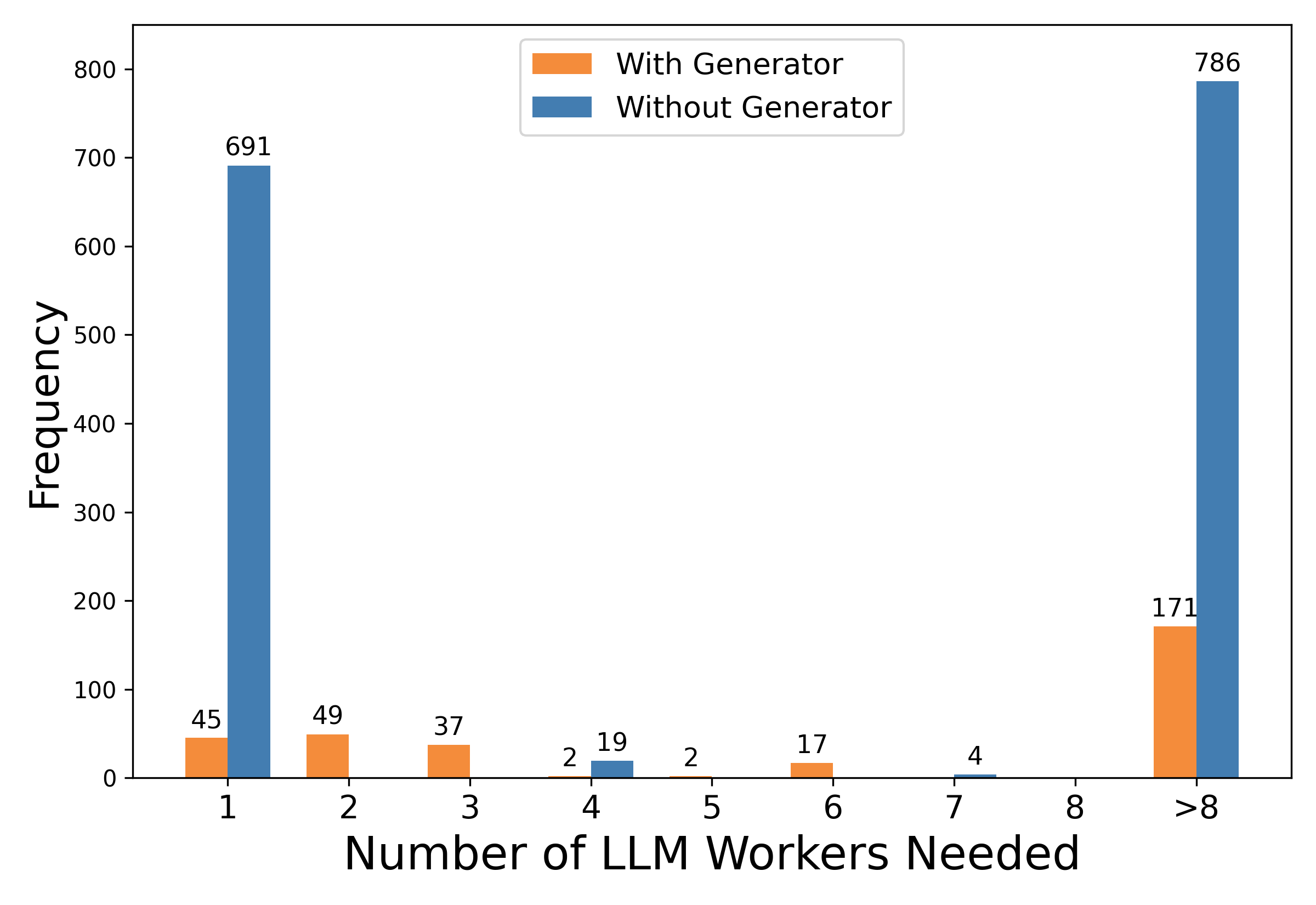}
        \caption{Ground truth rating $\pm 0.3$}
    \end{subfigure}%
    \caption{Number of LLM workers needed to reach different performance ranges with and without generator-based rating on Offensiveness dataset. The generator is trained on more than 13,000 instances.}
    \label{vae_llm_tolerance_0.1-0.3}
\end{figure}

\begin{figure}[!ht]
    \centering
    \begin{subfigure}{0.33\textwidth}
        \centering
        \includegraphics[width=\textwidth]{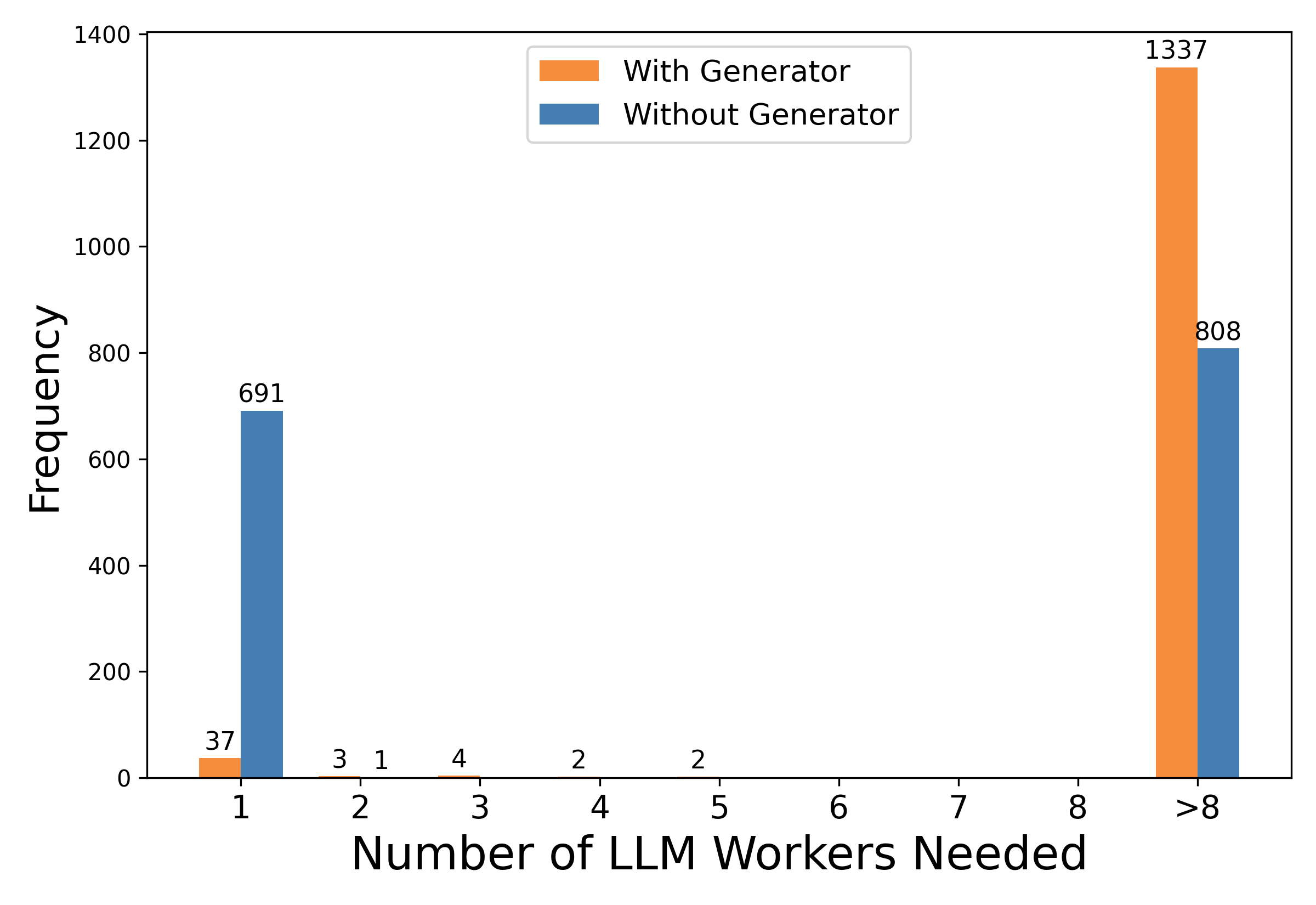}
        \caption{Ground truth rating $\pm 0.1$}
    \end{subfigure}%
    \begin{subfigure}{0.33\textwidth}
        \centering
        \includegraphics[width=\textwidth]{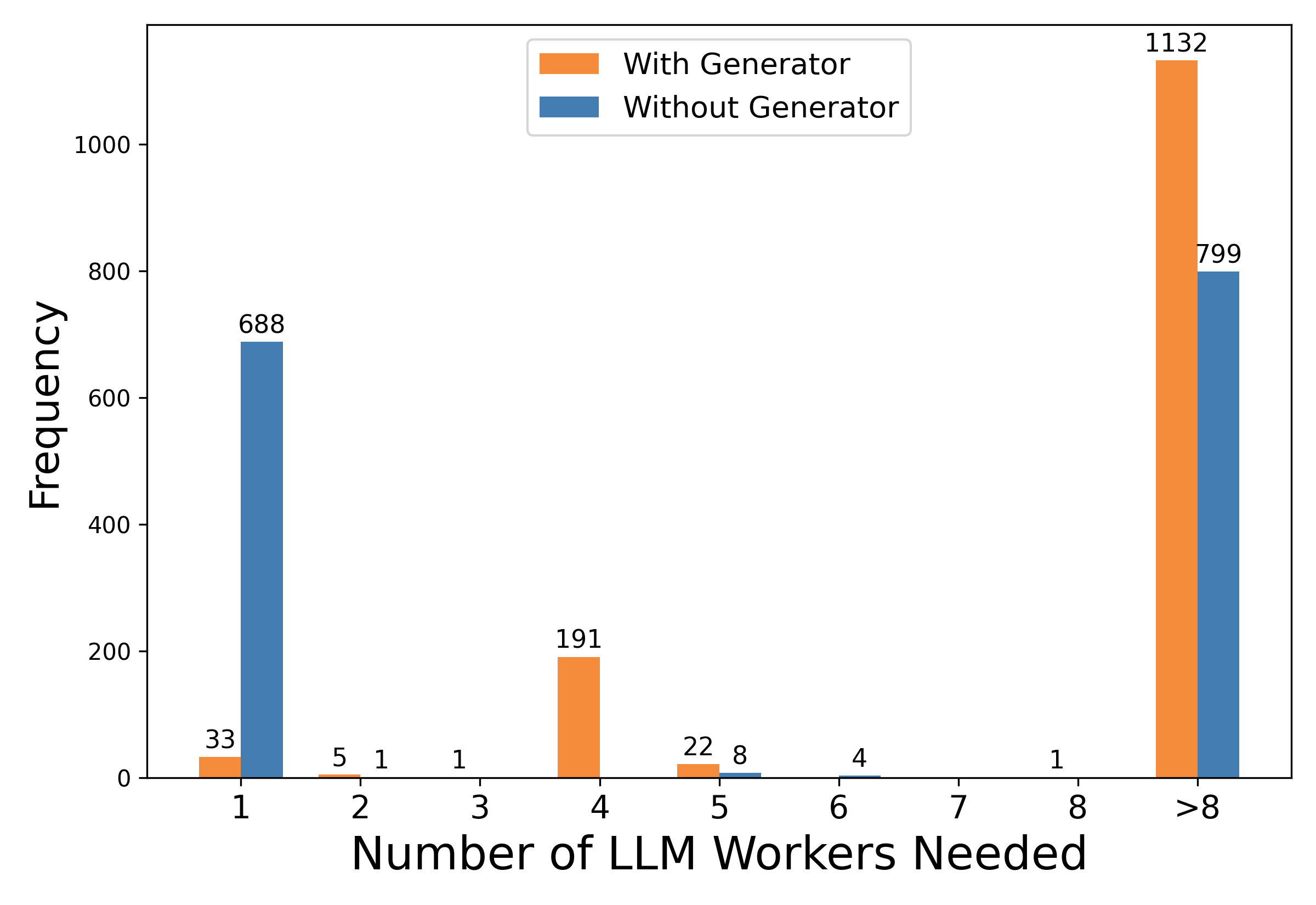}
        \caption{Ground truth rating $\pm 0.2$}
    \end{subfigure}\\
    \begin{subfigure}{0.33\textwidth}
        \centering
        \includegraphics[width=\textwidth]{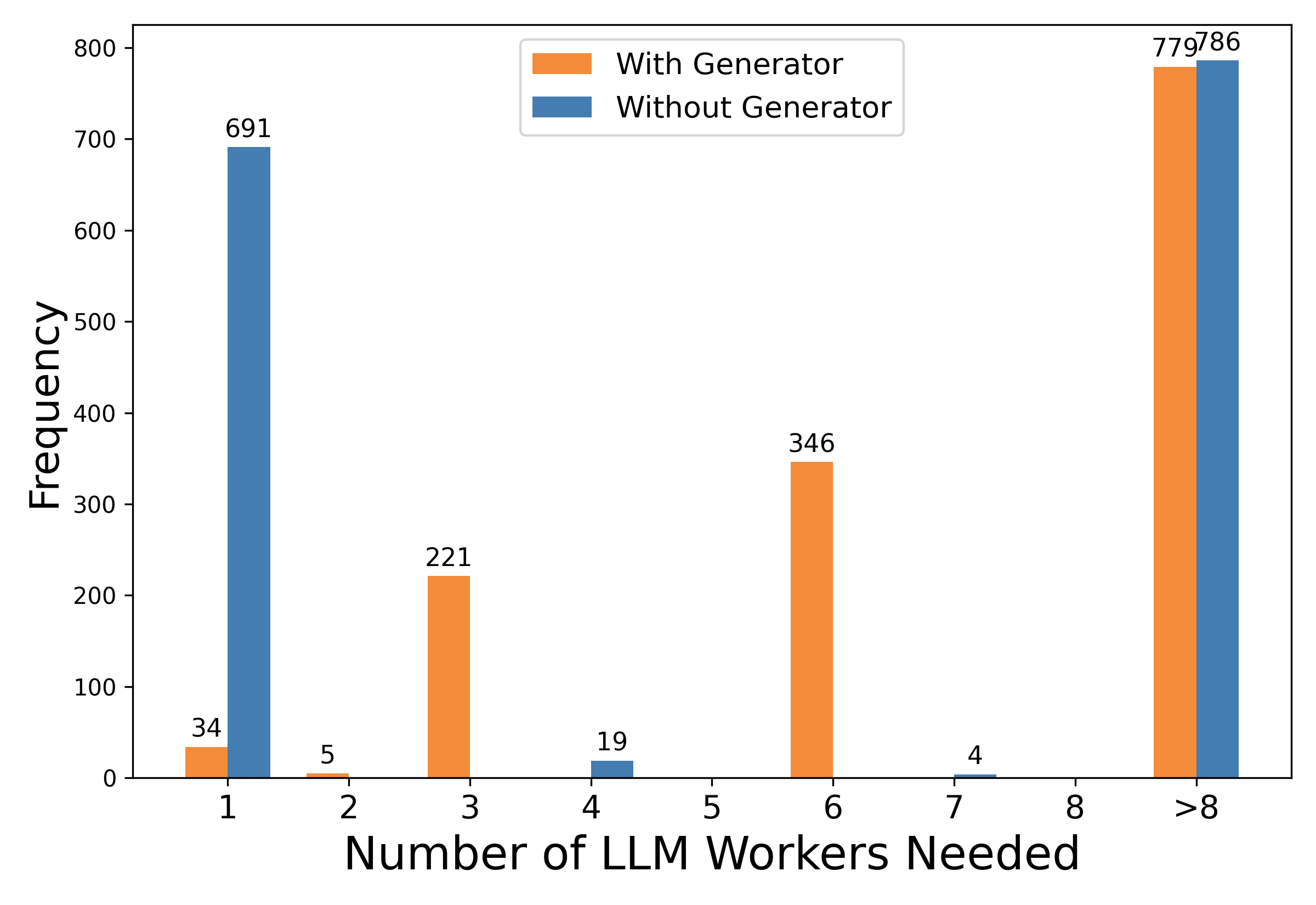}
        \caption{Ground truth rating $\pm 0.3$}
    \end{subfigure}%
    \begin{subfigure}{0.33\textwidth}
        \centering
        \includegraphics[width=\textwidth]{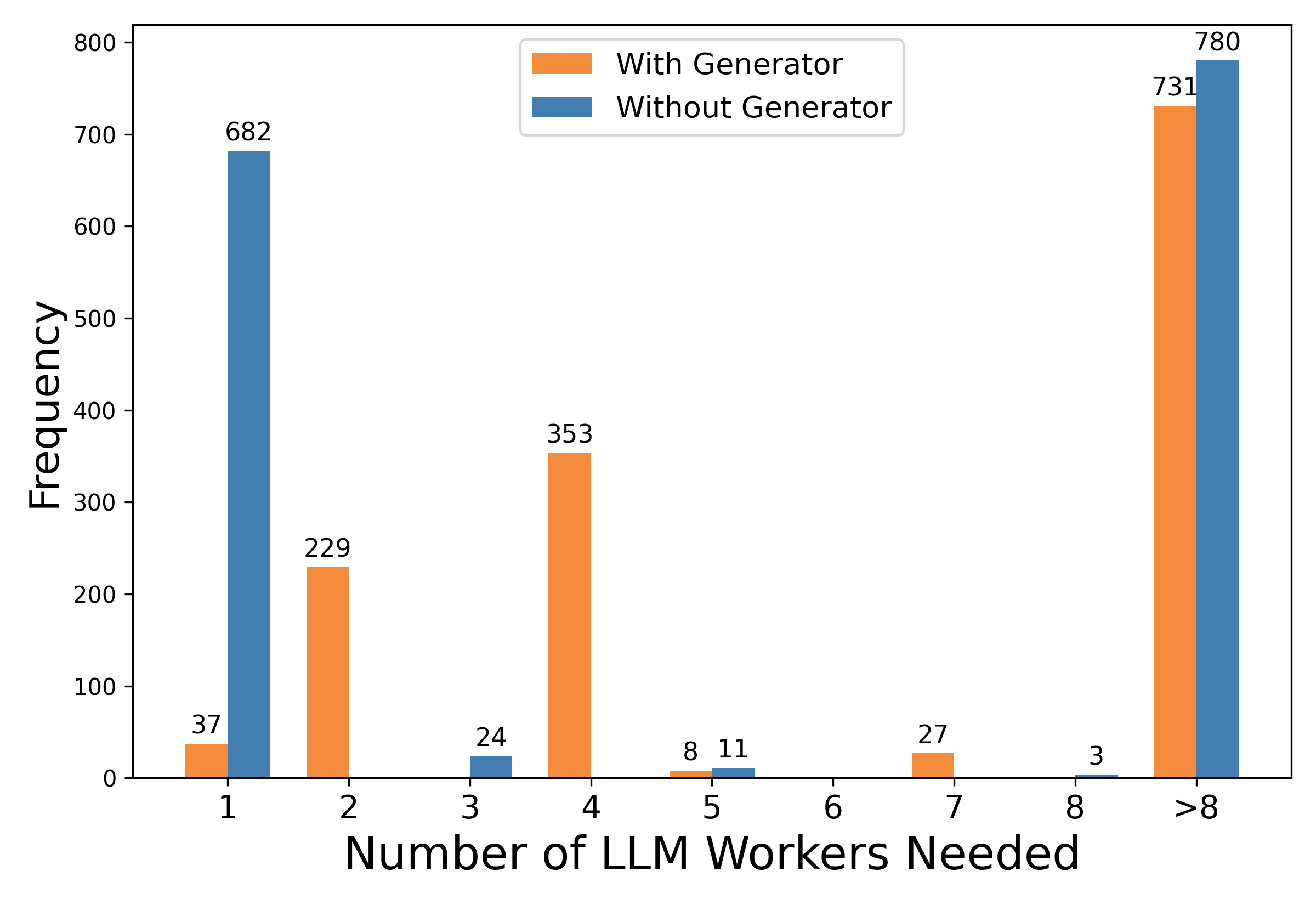}
        \caption{Ground truth rating $\pm 0.4$}
    \end{subfigure}%
    \caption{Number of LLM workers needed to reach different performance ranges with and without generator-based rating on Offensiveness dataset. The generator is trained on only 9 instances.}
    \label{vae_llm_tolerance_9_Offensiveness}
\end{figure}

On \textit{Offensiveness}, Figure \ref{vae_llm_tolerance_0.1-0.3} details LLM requirements for tolerances $\pm \{0.1, 0.2, 0.3\}$ with/without a generator trained on more than 13,000 instances. Figure \ref{vae_llm_tolerance_9_Offensiveness} shows results for a generator trained on only 9 instances. While this lightly-trained generator saw 1337 instances needing more than 8 LLMs at $\pm 0.1$ tolerance (more than its well-trained counterpart), it still outperformed the no-generator baseline. As tolerance loosens, many instances meet targets with 2-6 LLMs using the 9-instance generator, which requires fewer LLMs than the baseline at $\pm 0.3$ and $\pm 0.4$ tolerances, where the baseline still struggles.

\begin{figure}[!ht]
    \centering
    \begin{subfigure}{0.33\textwidth}
        \centering
        \includegraphics[width=\textwidth]{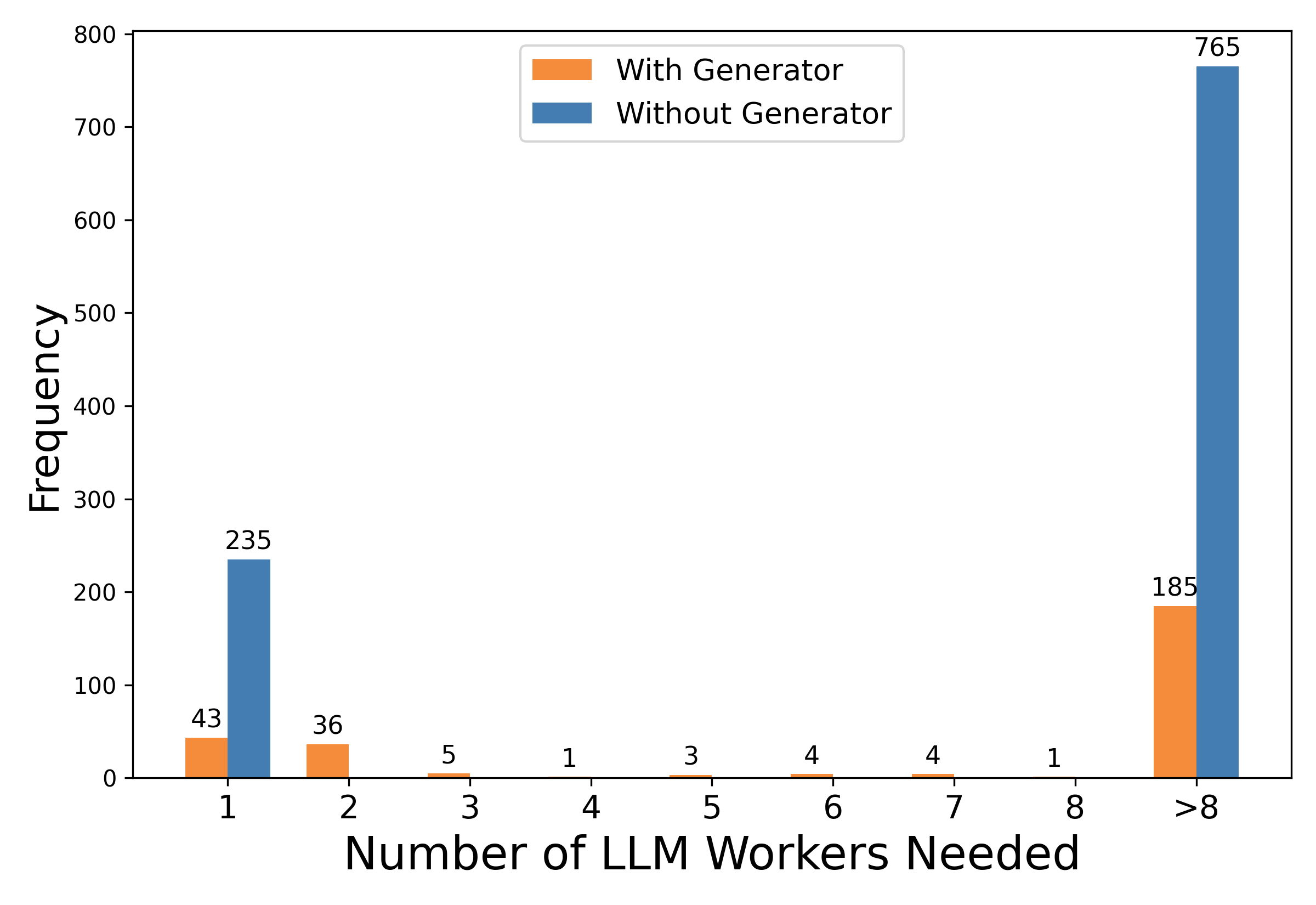}
        \caption{Ground truth rating $\pm 0.1$}
    \end{subfigure}%
    \begin{subfigure}{0.33\textwidth}
        \centering
        \includegraphics[width=\textwidth]{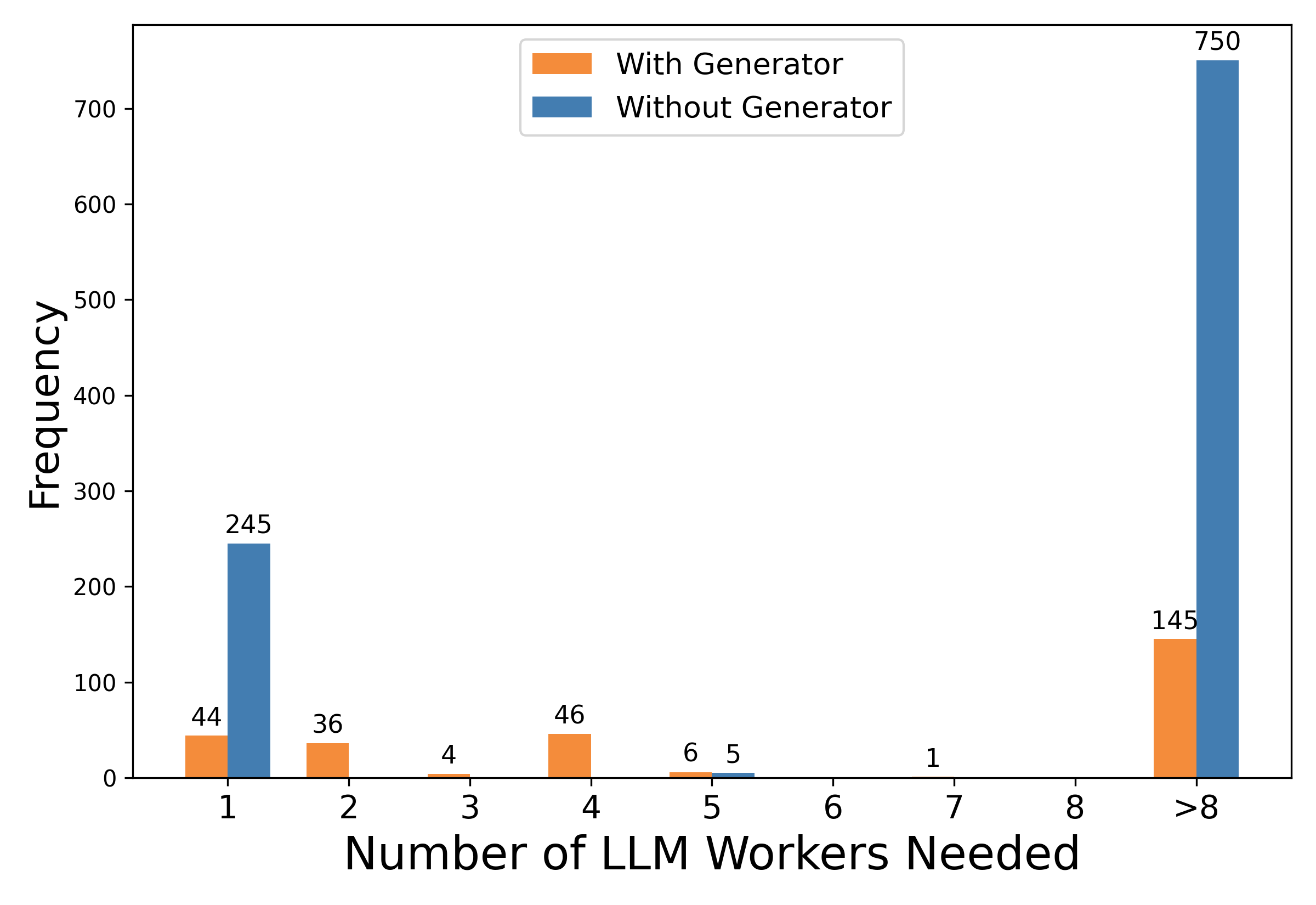}
        \caption{Ground truth rating $\pm 0.2$}
    \end{subfigure}\\
    \begin{subfigure}{0.33\textwidth}
        \centering
        \includegraphics[width=\textwidth]{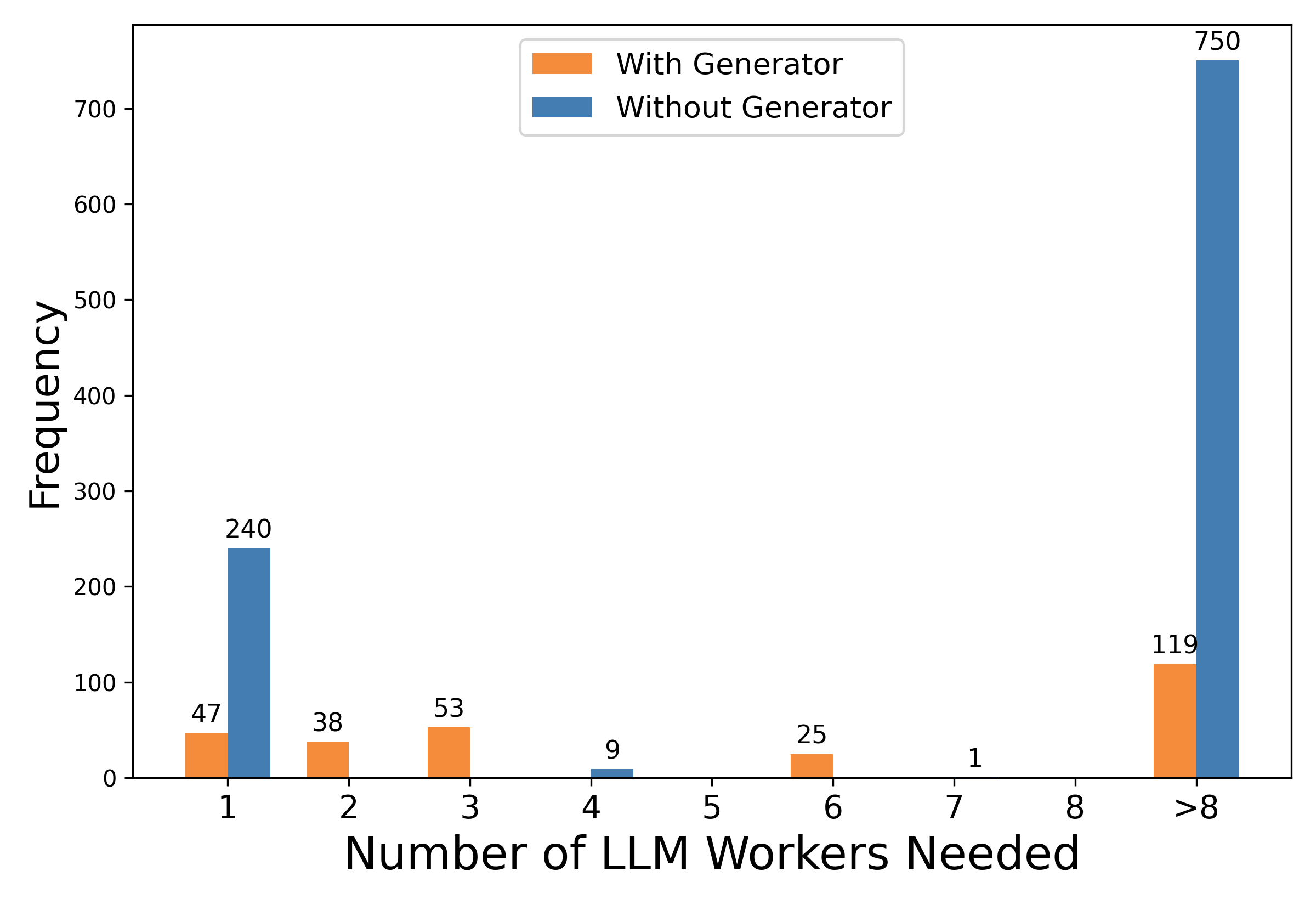}
        \caption{Ground truth rating $\pm 0.3$}
    \end{subfigure}%
    \begin{subfigure}{0.33\textwidth}
        \centering
        \includegraphics[width=\textwidth]{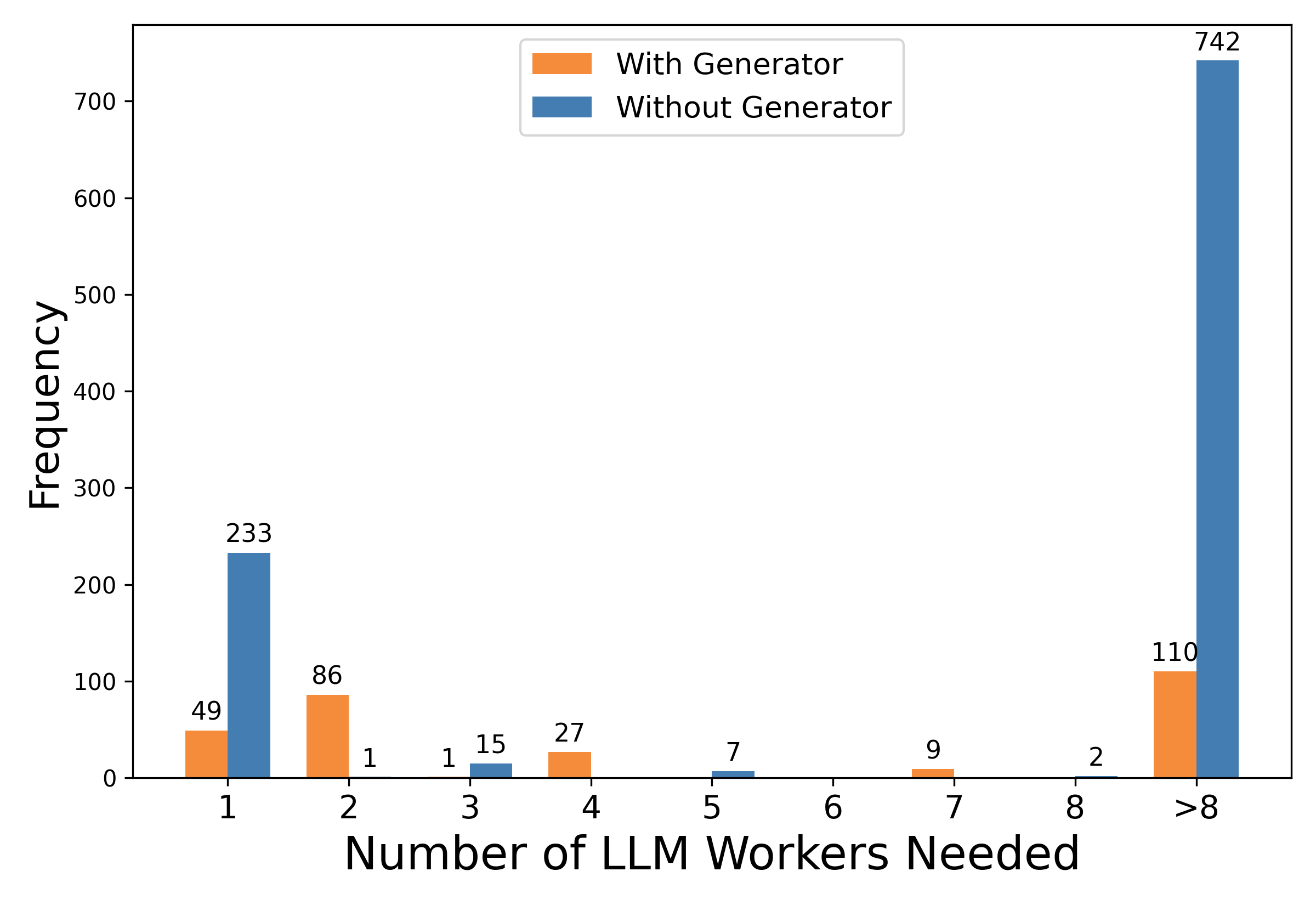}
        \caption{Ground truth rating $\pm 0.4$}
    \end{subfigure}%
    \caption{Number of LLM workers needed to reach different performance ranges with and without generator-based rating on QA Difficulty dataset.}
    \label{vae_llm_tolerance_QAdifficulty}
\end{figure}

Figures \ref{vae_llm_tolerance_QAdifficulty} (\textit{QA Difficulty}) presents results with/without this 9-instance generator. On both datasets, even this lightweight generator significantly reduces the required LLMs. Without a generator, the number of instances needing $>$8 LLMs remains high even at relaxed tolerances (e.g., \textit{QA Difficulty}: 750 at $\pm 0.3$, 742 at $\pm 0.4$). With the generator, as tolerance loosens, 1-5 LLMs resolve more instances, confirming its efficiency.
\end{APPENDIX}
\newpage
\bibliographystyle{informs2014}
\bibliography{informs2014}

\end{document}